\documentclass{CSML}

\def\dOi{12(1:8)2016}
\lmcsheading%
{\dOi}
{1--49}
{}
{}
{Oct.~20, 2014}
{Mar.~31, 2016}
{}

\ACMCCS{[{\bf Theory of computation}]: Semantics and
  Reasoning---Program Semantics---Operational semantics; [{\bf
    Software and its engineering}]: Software Notations and
  Tools---System Description Languages---System modeling languages}


\usepackage{hyperref}

\usepackage[active]{srcltx}
\usepackage[latin1]{inputenc}
\usepackage{mathpartir}
\usepackage{latexsym}
\usepackage{amssymb}
\usepackage{stmaryrd}
\usepackage{fixmath}
\usepackage{mathtools}
\usepackage{color}
\usepackage{graphicx}
\usepackage{url}
\usepackage{soul}
\usepackage{xspace}

\renewcommand{\vec}[1]{\tilde{#1}}

\newcommand{\one}{\mathbf{1}} 

\newcommand{\parop}{\;|\;}      
\newcommand{\sdot}{\, . \,}     

\makeatletter

\newcommand{\goesto}[1]{\@@transition\rightarrowfill{#1}}
\newcommand{\Goesto}[1]{\@@transition\Rightarrowfill{#1}}
\newcommand{\Goestoplus}[1]{\@@transition\Rightarrowplusfill{#1}}
\newcommand{\Goestotrad}[1]{\@@transition\Hookarrowfill{#1}}
\newcommand{\sgoesto}[1]{\@@transition\mapstofill{#1}}
\newcommand{\Sgoesto}[1]{\@@transition\Mapstofill{#1}}

\newbox\@transbox
\newbox\@arrowbox
\def\Rightarrowfill{$\m@th\mathord=\mkern-6mu%
  \cleaders\hbox{$\mkern-2mu\mathord=\mkern-2mu$}\hfill
  \mkern-6mu\mathord=\mkern-6mu\mathord\Rightarrow$}
\def\Hookarrowfill{$\m@th\mathord\lhook\mathord=\mkern-16mu%
  \cleaders\hbox{$\mkern-2mu\mathord=\mkern-2mu$}\hfill
  \mkern-6mu\mathord=\mkern-6mu\mathord\Rightarrow$}
\def\Rightarrowplusfill{$\m@th\mathord=\mkern-6mu%
  \cleaders\hbox{$\mkern-2mu\mathord=\mkern-2mu$}\hfill
  \mkern-6mu\mathord=\mkern-6mu\mathord\Rightarrow_+$}
\def\mapstofill{$\m@th\mathord\mapstochar\mathord-\mkern-6mu%
  \cleaders\hbox{$\mkern-2mu\mathord-\mkern-2mu$}\hfill
  \mkern-6mu\mathord\rightarrow$}
\def\Mapstofill{$\m@th\mathord|\mkern-3mu\mathord=\mkern-6mu%
  \cleaders\hbox{$\mkern-2mu\mathord=\mkern-2mu$}\hfill
  \mkern-6mu\mathord=\mkern-6mu\mathord\Rightarrow$}
\def\hookarrowfill{$\m@th\mathord\lhook\mathord-\mkern-16mu%
  \cleaders\hbox{$\mkern-2mu\mathord-\mkern-2mu$}\hfill
  \mkern-6mu\mathord\rightarrow$}
\def\nhookarrowfill{$\m@th\mathord\lhook\mathord-\mkern-16mu%
  \cleaders\hbox{$\mkern-2mu\mathord-\mkern-2mu$}\hfill
  \mkern-6mu\mathord\not\mathord-\mkern-16mu
  \cleaders\hbox{$\mkern-2mu\mathord-\mkern-2mu$}\hfill
  \mkern-6mu\mathord\rightarrow$}
\def\nrightarrowfill{$\m@th\mathord-\mkern-6mu%
  \cleaders\hbox{$\mkern-2mu\mathord-\mkern-2mu$}\hfill
  \mkern-6mu\mathord\not\mathord-\mkern-6mu
  \cleaders\hbox{$\mkern-2mu\mathord-\mkern-2mu$}\hfill
  \mkern-6mu\mathord\rightarrow$}
\def\nmapstofill{$\m@th\mathord\mapstochar\mathord-\mkern-6mu%
  \cleaders\hbox{$\mkern-2mu\mathord-\mkern-2mu$}\hfill
  \mkern-6mu\mathord\not\mathord-\mkern-6mu
  \cleaders\hbox{$\mkern-2mu\mathord-\mkern-2mu$}\hfill
  \mkern-6mu\mathord\rightarrow$}
\def\@@transition#1#2%
   {\setbox\@transbox\hbox
      {\strut\hskip0.25em$\scriptstyle#2$\hskip0.35em}
   \ifdim\wd\@transbox<1.5em
      \setbox\@transbox\hbox to 1.5em{\hfil\box\@transbox\hfil}\fi
   \setbox\@arrowbox\hbox to \wd\@transbox{#1}
   \ht\@arrowbox 0.2\ht\@arrowbox\dp\@arrowbox\z@
   \ht\@transbox 0.65\ht\@transbox
   \setbox\@transbox\hbox{$\mathop{\box\@arrowbox}\limits^{\box\@transbox}$}
   \mathrel{\box\@transbox}}

\makeatother
\newcommand{\gt}[1]{\goesto{#1}}

\newcommand{\gtt}{\gt{\tau}}

\newcommand{\bisim}{\stackrel{\mbox{\bf .}}{\sim}} 
\newcommand{\wbisim}{\stackrel{\mbox{\bf .}}{\approx}} 

\newcommand{\names}[1]{{\rm n}(#1)}

\newcommand{\bn}[1]{\mbox{\rm bn($#1$)}}
\newcommand{\fn}[1]{\mbox{\rm fn($#1$)}}

\newcommand{\dom}[1]{\operatorname{dom}(#1)}

\newcommand{\isdef}{\stackrel{\mbox{\scriptsize\it def}}{=}} 

\newcommand{\tinfrule}{\rule[-0.5ex]{0ex}{3ex}} 

\newcommand{\atabrule}[2]   
{\frac{\tinfrule\displaystyle #1}{\tinfrule\displaystyle#2}}

\newcommand{\pic}{pi-calculus}

\definecolor{sgreen}{rgb}{0,0.3,0}

\newcommand{\ve}[1]{\widetilde{#1}}

\newcommand{\apitransarrow}[1]{\goesto{#1}}
\newcommand{\trans}[3]{#1 \; \apitransarrow{#2} \; #3}
\newcommand{\transPrim}[4][e]{#2 \mathrel{{\apitransarrow{#3}}{}^{#1}} #4}

\newcommand{\transw}[3]{#1 \; \stackrel{#2}{\Longrightarrow} \; #3}

\newcommand{\out}[2]{\overline{#1}\; #2 }
\newcommand{\inn}[2]{\underline{#1}\; #2}
\newcommand{\innn}[2]{\underline{#1} (#2)}
\newcommand{\inlabel}[2]{\underline{#1}\;#2}
\newcommand{\lin}[2]{\underline{#1}(\lambda #2)}
\newcommand{\boutlabel}[3]{\overline{#1}\:({\nu}#2)\:{#3}}

\newcommand{\noboutlabel}[2]{\overline{#1}\langle #2 \rangle}

\newcommand{\frameof}[1]{\mathcal{F}(#1)}

\newcommand{\fr}[1]{\mathcal{F}(#1)}

\newcommand{\frames}{\,\rhd\,}

\newcommand{\res}[1]{({\nu}#1)}

\newcommand{\freshin}{\#}

\newcommand{\ifthen}[2]{\mbox{\rm\ensuremath{\textbf{if } #1 \textbf{ then } #2}}}
\newcommand{\ci}[2]{#1:#2} 
\newcommand{\boldword}[1]{\mbox{\rm\ensuremath{\textbf{#1}}}}
\newcommand{\caseonly}{\boldword{case}}

\newcommand{\casesep}{\mathrel{[\hspace{-0.1ex}]}} 
\newcommand{\casesprod}[1]{\caseonly\casesep #1}
\newcommand{\pll}{\:|\:}

\newcommand{\vect}[1]{\langle #1 \rangle}

\newcommand{\nil}{\mathbf{0}}

\newcommand{\R}{\mathcal{R}}

\newcommand{\simplies}{\leq}

\newcommand{\sequivalent}{\simeq}

\newcommand{\subst}[2]{\{{}^{#1}\!\!/\!{}_{#2}\}}

\newcommand{\lsubst}[2]{[#2 := #1]}

\newcommand{\n}{\mbox{{\rm n}}}

\newcommand{\vars}{\mathit{v}}

\newcommand{\true}{\rm true}

\newcommand{\ftimes}{{\otimes}}

\newcommand{\emptyframe}{{\bf 1}}

\newcommand{\sch}{\stackrel{{\hspace{.05ex}\hbox{\ensuremath{\mathbf{.}}}}}{\leftrightarrow}}

\newcommand{\framedtransempty}[4]{ #1 \frames \trans{#2}{#3}{#4}}
\newcommand{\frnames}[1]{\ve{b_{#1}}}
\newcommand{\frass}[1]{\Psi_{#1}}

\newcommand{\nameset}[1][]{\mathcal{N}_{{#1}}}

\newcommand{\pass}[1]{\llparenthesis #1 \rrparenthesis}
\newcommand{\framepair}[2]{(\nu #1)#2}

\newcommand{\trms}{{\rm\bf T}}
\newcommand{\conditions}{{\rm\bf C}}
\newcommand{\assertions}{{\rm\bf A}}
\newcommand{\pats}{{\rm\bf X}} 
\newcommand{\processes}{{\rm\bf P}}

\newcommand{\CanSend}{\mathrel{\overline{\propto}}}
\newcommand{\CanReceive}{\mathrel{\underline{\propto}}}
\newcommand{\CanSubstitute}{\mathrel{\Yleft}}
\newcommand{\Restrictionsorts}{\sortset_\nu}
\newcommand{\Set}[1]{\{#1\}}
\newcommand{\Match}{\textsc{match}}

\newcommand{\SORT}{\textsc{sort}}
\newcommand{\sort}[1]{\SORT(#1)}
\newcommand{\sortset}{\mathcal{S}}
\newcommand{\namesorts}{\mathcal{S_{N}}}

\newcommand{\pto}{\rightharpoonup}
\newcommand{\kw}[1]{\ensuremath{\mathtt{#1}}}
\renewcommand{\vars}{\textsc{vars}}

\newcommand{\Pow}{\mathcal{P}}
\newcommand{\PowFin}{\mathcal{P}_{\mathrm{fin}}}

\newcommand{\Pwb}{\textsc{Pwb}\xspace}

\newcommand{\abs}[1]{\lvert {#1}\rvert}
\newcommand{\semb}[1]{\llbracket #1 \rrbracket}
\newcommand{\instancename}[1]{\textbf{#1}}
\newcommand{\INSTANCE}[2]{
  \[\boxed{
     \begin{array}{c}
        \instancename{#1} \\\hline\\[-2.2em]\\
          {#2}
      \end{array}
     }
   \]
  }

\newcommand{\ARRAYOF}[1]{\begin{array}[t]{l} #1 \end{array}}
\newcommand{\instance}[2]{\INSTANCE{#1}{\ARRAYOF{#2}}}
\newcommand{\instanceFrom}[3]{\INSTANCE{#1}{
    \begin{array}{l}
      \text{Everything as in \instancename{#2} except:}\\
      #3
    \end{array}}}

\newcommand{\instanceTwo}[4][\quad]{\INSTANCE{#2}{\ARRAYOF{#3}#1\ARRAYOF{#4}}}

\newcommand{\instanceTwoFrom}[5][\quad]{\instanceFrom{#2}{#3}{\ARRAYOF{#4}#1\ARRAYOF{#5}}}

\newcommand{\RelR}{\mathrel{\mathcal{R}}}
\newcommand{\bmes}[1]{\overline{#1}}

\newcommand{\ErrChoice}{\talloblong}
     
\newcommand{\GAP}{0.8em}

\newcommand{\BnfOr}{\;\;|\;\;}
\newcommand{\BnfDef}{\mathrel{::=}}

\theoremstyle{definition}

\begin{document}
\title[A Sorted Semantic Framework for Applied Process Calculi]{A Sorted Semantic Framework\\ for Applied Process Calculi}
\author[J.~Borgstr{\"o}m]{Johannes Borgstr{\"o}m}
\address{Computing Science Division
Department of Information Technology
Uppsala University}
\email{\{johannes.borgstrom, ramunas.gutkovas, Joachim.Parrow,  Bjorn.Victor,  	johannes.aman-pohjola\}@it.uu.se}
\thanks{This project is financially supported by the Swedish Foundation for Strategic Research}
\author[R.~Gutkovas]{Ram{\=u}nas Gutkovas}
\address{\vspace{-18 pt}}
\author[J.~Parrow]{Joachim Parrow}
\address{\vspace{-18 pt}}
\author[B.~Victor]{Bj{\"o}rn Victor}
\address{\vspace{-18 pt}}
\author[J.~{\AA}man Pohjola]{Johannes {\AA}man Pohjola}
\address{\vspace{-18 pt}}
%

\begin{abstract}
Applied process calculi include advanced programming constructs such as type systems, communication with pattern matching, encryption primitives, concurrent constraints, nondeterminism, process creation, and dynamic connection topologies. 
Several such formalisms, e.g.~the applied pi calculus,  are extensions of the the pi-calculus; 
a growing number is geared towards particular applications or computational paradigms. 

Our goal is a unified framework to represent different process calculi and notions of computation.
To this end, we extend our previous work on psi-calculi with novel abstract patterns and pattern matching, 
and add sorts to the data term language, 
giving sufficient criteria for subject reduction to hold.
Our framework can directly represent several existing process calculi; 
the resulting transition systems are isomorphic to the originals up to strong bisimulation. 
We also demonstrate  different notions of computation on data terms,
including cryptographic primitives and a lambda-calculus with erratic choice.
Finally, we prove standard congruence and structural properties of bisimulation;
the proof has been machine-checked using Nominal Isabelle in the case
of a single name sort. 
\end{abstract}

\maketitle

\section{Introduction}
\label{sec:introduction}
There is today a growing number of high-level constructs in the area of concurrency. Examples include type systems, communication with pattern matching, encryption primitives, concurrent constraints, nondeterminism, %
and dynamic connection topologies. Combinations of such constructs are included in a variety of application oriented process calculi. For each such calculus its internal consistency, in terms of congruence results and algebraic laws, must be established independently. Our aim is a framework where many such calculi fit and where such results are derived once and for all, eliminating the need for individual proofs about each calculus.

Our effort in this direction is the framework of psi-calculi~\cite{LMCS11.PsiCalculi}, which
provides machine-checked proofs that important meta-theoretical
properties, such as compositionality of bisimulation, hold in all instances of the
framework.
We claim that the theoretical development is more robust than that of other calculi of comparable complexity, 
since we use a structural operational semantics given by a single inductive definition, 
and since we have checked most results in the interactive theorem prover Nominal Isabelle~\cite{U07:NominalTechniquesInIsabelleHOL}.

In this paper we introduce a novel generalization of pattern matching, decoupled from the definition of substitution, 
and add sorts for data terms and names. 
The generalized pattern matching is a new contribution that holds general interest; 
here it allows us to directly capture computation on data in advanced process calculi, without elaborate encodings. 

We evaluate our framework by providing instances that correspond to standard calculi, 
and instances that use several different notions of computation.  
We define strong criteria for a psi-calculus to \emph{represent} another process calculus, meaning that they are for all practical purposes one and the same. Representation is stronger than the standard \emph{encoding} correspondences e.g.~by Gorla~\cite{Gorla:encoding}, which define criteria for one language to encode the behaviour of another.
The representations that we provide of other standard calculi 
 advance our previous work, where we had to resort to nontrivial
encodings with an unclear formal correspondence to the source calculus.

An extended abstract~\cite{borgstrom13:sorted} of the present paper
has previously been published.

\subsection{Background: Psi-calculi}
In the following we assume the reader to be acquainted with the basic ideas of process algebras based on the pi-calculus, 
and explain psi-calculi by a few simple examples. 
Full definitions can be found in the references above, 
and for a reader not acquainted with our work we recommend the first few sections of~\cite{LMCS11.PsiCalculi} for an introduction.

A psi-calculus has a notion of data terms, ranged over by $K,L,M,N$, 
and we write $\out{M}N \sdot P$ to represent an agent sending the term $N$ along the channel $M$ (which is also a data term), continuing as the agent $P$. 
We write $\lin{K}{\ve{x}}X \sdot Q$ to represent an agent that can input along the channel $K$, 
receiving some object matching the pattern $X$, where $\ve{x}$ are the variables bound by the prefix.
These two agents can interact under two conditions. 
First, the two channels must be \emph{channel equivalent}, as defined by the channel equivalence predicate $M \sch K$. Second, $N$ must match the pattern $X$. 

Formally, a \emph{transition} is of kind $\framedtransempty{\Psi}{P}{\alpha}{P'}$, 
meaning that in an environment represented by the \emph{assertion} $\Psi$ 
the agent $P$ can do an action $\alpha$ to become $P'$. 
An assertion embodies a collection of facts used to infer \emph{conditions} such as the channel equivalence predicate $\sch$. 
To continue the example, if $N=X\lsubst{\ve{L}}{\ve{x}}$ we will have
$\framedtransempty{\Psi}{\out{M}N \sdot P\parop \lin{K}{\ve{x}}X \sdot Q}{\tau}{P \parop Q\lsubst{\ve{L}}{\ve{x}} }$
when additionally $\Psi \vdash M \sch K$, i.e.~when the assertion $\Psi$ entails that $M$ and $K$ represent the same channel.  
In this way we may introduce a parametrised equational theory over a data structure for channels. 
Conditions, ranged over by $\varphi$, can be tested in the \textbf{if} construct: we have that
$\framedtransempty{\Psi}{\ifthen{\varphi}{P}}{\alpha}{P'}$ when $\Psi \vdash \varphi$ and $\framedtransempty{\Psi}{P}{\alpha}{P'}$. 
In order to represent concurrent constraints and local knowledge, assertions can be used as agents: 
$\pass{\Psi}$ stands for an agent that asserts $\Psi$ to its environment. 
Assertions may contain names and these can be scoped; 
for example, in $P \parop (\nu a)(\pass{\Psi} \parop Q)$ the agent $Q$ uses all entailments provided by $\Psi$, 
while $P$ only uses those that do not contain the name $a$.

Assertions and conditions can, in general, form any logical theory. Also the data terms can be drawn from an arbitrary set.  
One of our major contributions has been to pinpoint the precise requirements on the data terms and logic for a calculus to be useful in the sense that the natural formulation of bisimulation satisfies the expected algebraic laws (see Section~\ref{sec:definitions}).
It turns out that it is necessary to view the terms and logics as \emph{nominal}~\cite{PittsAM:nomlfo-jv}. This means that there is a distinguished set of names, and for each term a well defined notion of \emph{support}, intuitively corresponding to the names occurring in the term.  Functions and relations must be \emph{equivariant}, meaning that they treat all names equally. 
In addition, we impose straight-forward requirements on the combination of assertions, on channel equivalence, and on substitution.
Our requirements are quite general, and therefore our framework accommodates a wide variety of applied process calculi.

\subsection{Extension: Generalized pattern matching}
\label{sec:patterns-intro}

In our original definition of psi-calculi (\cite{LMCS11.PsiCalculi}, called ``the original psi-calculi'' below), patterns are just terms and pattern matching is defined by substitution in the usual way: the output object $N$ matches the pattern $X$ with binders $\ve{x}$ iff $N=X\lsubst{\ve{L}}{\ve{x}}$.
In order to increase the generality we now introduce a function $\Match$ which takes a term $N$, a sequence of names $\ve{x}$ and a pattern $X$, returning a set of sequences of terms; the intuition is that if  $\ve{L}$ is in $\Match(N,\ve{x},X)$ then the term $N$ matches the pattern $X$ by instantiating $\ve{x}$ to $\ve{L}$.
The receiving agent $\lin{K}{\ve{x}}X \sdot Q$ then continues as $Q\lsubst{\ve{L}}{\ve{x}}$. 

As an example we consider a term algebra with two function symbols: $\kw{enc}$ of arity three and $\kw{dec}$ of arity two. Here $\kw{enc}(N,n,k)$ means encrypting $N$ with the key $k$ and a random nonce $n$ and
and $\kw{dec}(N,k)$ represents symmetric key decryption, discarding the nonce. Suppose an agent sends an encryption, as in $\out{M}{\kw{enc}(N,n,k)} \sdot P$.
If we allow all terms to act as patterns, a receiving agent can use $\kw{enc}(x,y,z)$ as a pattern, as in $\lin{c}{x,y,z}\kw{enc}(x,y,z) \sdot Q$, and in this way decompose the encryption and extract the message and key. Using the encryption function as a destructor in this way is clearly not the intention of a cryptographic model.
With the new general form of pattern matching, we can simply 
limit the patterns to not bind names in terms at key position. 
Together with the separation between patterns and terms, this allows to directly represent dialects of the spi-calculus 
as in Sections~\ref{sec:DYalg} and~\ref{sec:pmSpi}.

Moreover, the generalization makes it possible to safely use rewrite rules such as $\kw{dec}(\kw{enc}(M,N,K),K)\to M$. 
In the psi-calculi framework such evaluation is not a primitive concept, but it can be part of the substitution function, 
with the idea that with each substitution all data terms are normalized according to rewrite rules.
Such evaluating substitutions are dangerous for two reasons. First, in the original psi-calculi they can introduce ill-formed input prefixes.
The input prefix $\lin{M}{\ve{x}}N$ is well-formed when $\ve{x}\subseteq\n(N)$, i.e.~the names $\ve{x}$ must all occur in $N$; a rewrite of the well-formed $\lin{M}{y}{\kw{dec}(\kw{enc}(N,y,k),k) \sdot P}$ to $\lin{M}{y}{N \sdot P}$ yields an ill-formed agent when $y$ does not appear in $N$.
Such ill-formed agents could also arise from input transitions in some original psi-calculi; with the current generalization preservation of well-formedness is guaranteed.

Second, in the original psi-calculi there is a requirement that 
substituting $\ve{L}$ for $\ve{x}$ in $M$ must yield a term containing all names in $\ve{L}$ whenever $\ve{x}\subseteq\n(M)$. 
The reason is explained at length in~\cite{LMCS11.PsiCalculi}; briefly put, without this requirement the scope extension law is unsound. 
If rewrites such as $\kw{dec}(\kw{enc}(M,N,K),K)\to M$ are performed by substitutions this requirement is not fulfilled, since a substitution may then erase the names in $N$ and~$K$. However, a closer examination reveals that this requirement is only necessary for some uses of substitution. In the transition
\[\trans{\lin{M}{\ve{x}}{N}.P}{\inn{K}{N\lsubst{\ve{L}}{\ve{x}}}}{P\lsubst{\ve{L}}{\ve{x}}}\]
the non-erasing criterion is important for the substitution above the arrow ($N\lsubst{\ve{L}}{\ve{x}}$) but unimportant for the substitution after the arrow ($P\lsubst{\ve{L}}{\ve{x}}$). In the present paper, we replace the former of these uses by the $\Match$ function, where a similar non-erasing criterion applies. All other substitutions may safely use arbitrary rewrites, even erasing ones.

In this paper, we address these three issues by introducing explicit notions of patterns, pattern variables and matching.
This allows us to control precisely which parts of messages can be bound by pattern-matching and how messages can be deconstructed,
admit computations such as $\kw{dec}(\kw{enc}(M,N,K),K)\to M$.
We obtain criteria that ensure that well-formedness is preserved by transitions, 
and apply these to the original psi-calculi~\cite{LMCS11.PsiCalculi} (Theorem~\ref{thm:origpats}) 
and to pattern-matching spi calculus~\cite{haack.jeffrey:pattern-matching-spi} (Lemma~\ref{lem:dypres}).

\subsection{Extension: Sorting}
\label{sec:sorting-intro}
Applied process calculi often make use of a sort system.
The applied pi-calculus~\cite{abadi.fournet:mobile-values} has a name sort and a data sort; 
terms of name sort must not appear as subterms of terms of data sort.
It also makes a distinction between input-bound variables (which may be substituted) 
and restriction-bound names (which may not).
The pattern-matching spi-calculus~\cite{haack.jeffrey:pattern-matching-spi} 
uses a sort of patterns and a sort of implementable 
terms; every implementable term can also be used as a pattern.

To represent such calculi, we admit a user-defined sort system on names, terms and patterns. 
Substitutions are only well-defined if they conform to the sorting discipline.
To specify which terms can be used as channels, and which values can be received on them, 
we use compatibility predicates on the sorts of the subject and the object in input and output prefixes.
The conditions for preservation of sorting by transitions (subject reduction) are very weak,
allowing for great flexibility when defining instances.

The restriction to well-sorted substitution also allows to avoid ``junk'': terms that exist solely to make substitutions total. A prime example is representing the polyadic pi-calculus as a psi-calculus. The terms that can be transmitted between agents are tuples of names. Since a tuple is a term it can be substituted for a name, even if that name is already part of a tuple. The result is that the terms must admit nested tuples of names, which do not occur in the original calculus. 
Such anomalies disappear when introducing an appropriate sort system; cf.~Section~\ref{sec:polyPi}.

\subsection{Related work.}
Pattern-matching is in common use in functional programming languages. 
Scala admits pattern-matching of objects~\cite{Odersky.ECOOP07.matching.objects} 
using a method \texttt{unapply} that turns the receiving object into a matchable value (e.g.~a tuple).
F\# admits the definition of pattern cases independently of the type that they should match~\cite{Syme.ICFP07.active.patterns}, 
facilitating interaction with third-party and foreign-language code.
Turning to message-passing systems,
LINDA~\cite{Gelernter.1985.LINDA} uses pattern-matching when receiving from a tuple space.
Similarly, in Erlang, message reception from a mailbox is guarded by a pattern.

These notions of patterns, with or without computation, are easily supported by the \textsc{match} construct.
The standard first-match policy can be encoded by extending the pattern language
with mismatching and conjunction~\cite{Krishnaswami:POPL09:focusing.pattern.matching}.

\subsubsection*{Pattern matching in process calculi}
The pattern-matching spi-calculus~\cite{haack.jeffrey:pattern-matching-spi}
limits which variables may be binding in a pattern in order to 
match encrypted messages without binding unknown keys (cf.~Section~\ref{sec:pmSpi}).
The Kell calculus~\cite{GC04Kell} also uses pattern languages equipped with a match function.
However, in the Kell calculus the channels are single names and appear as part of the pattern in the input prefix,
patterns may match multiple communications simultaneously (\`a la join calculus), 
and first-order pattern variables only match names (not composite messages) 
which reduces expressiveness~\cite{given-wilson.express14.intensional}.

The applied pi-calculus~\cite{abadi.fournet:mobile-values} models deterministic computation 
by using for data language a term algebra modulo an equational logic.
ProVerif~\cite{BlanchetBook09} is a specialised tool for security
protocol verification in an extension of applied pi, including a pattern matching construct.
Its implementation allows pattern matching of tagged tuples modulo a user-defined
rewrite system; this is strictly less general than the psi-calculus 
pattern matching described in this paper (cf.~Section~\ref{sec:nform}).

Other tools for process calculi extended with datatypes include
mCRL2~\cite{mCRL2.TACAS13} for ACP, which allows higher order sorted
term algebras and equational logic, and
PAT3~\cite{dblp:conf/issre/liusd11} which includes a CSP$\sharp$
\cite{Sun:2009:ISP:1607726.1608426} module where actions built over
types like booleans and integers are extended with C$\sharp$-like
programs. %
In all these cases, the pattern matching is defined by substitution in the usual way.

\subsubsection*{Sort systems for mobile processes}
Sorts for the pi-calculus were first described by Milner~\cite{milner:polyadic-tutorial},
and were developed in order to remove nonsensical processes using polyadic communication,
similar to the motivation for the present work.

In contrast, H{\"u}ttel's dependently typed psi-calculi~\cite{Hyttel.CONCUR11.TypedPsi,hyttel.tgc13.resources}
is intended for a more fine-grained control of the behaviour of processes, 
and is capable of capturing a wide range of earlier type systems for pi-like
calculi formulated as instances of psi-calculi. 
In H{\"u}ttel's typed psi-caluli
the term language is a free term algebra (without name binders),
using the standard notions of substitution and matching, 
and not admitting any computation on terms.

In contrast, in our sorted psi-calculi terms and substitution are general. 
A given term always has a fixed sort, not dependent on any term or value and independent of its context.
We also have important meta-theoretical results, with machine-checked proofs for the case of a single name sort, 
including congruence results and structural equivalence laws for well-sorted bisimulation, 
and the preservation of well-sortedness under structural equivalence;
no such results exist for H{\"u}ttel's typed psi-calculi.
Indeed, our sorted psi-calculi can be seen as a foundation for H{\"u}ttel's typed psi-calculi: 
we give a formal account of the separation between variables and names used in H{\"u}ttel's typed psi-calculi, 
and substantiate H{\"u}ttel's claim that ``the set of well-[sorted] terms is closed under well-[sorted] substitutions, which suffices'' (Theorem~\ref{thm:sortstructcong}).

The state-of-the art report~\cite{BETTY13.stateoftheart} of WG1 of the BETTY project (EU COST Action IC1201) 
is a comprehensive guide to behavioural types for process calculi. 

Fournet et al.~\cite{FGM05:ATypeDisciplineForAuthiorizationPolicies}
add type-checking for a general authentication logic to a process calculus with destructor matching;
there the authentication logic is only used to specify program correctness,
and does not influence the operational semantics in any way.

\subsection{Results and outline}
\label{sec:outline}
In Section~\ref{sec:definitions} we define psi-calculi with the above extensions and prove preservation of well-formedness. 
In Section~\ref{sec:meta-theory} we prove the usual algebraic properties of bisimilarity. The proof is in two steps: a machine-checked proof for calculi with a single name sort, followed by manual proof based on the translation of a multi-sorted psi calculus instance to a corresponding single-sorted instance.
We demonstrate the expressiveness of our generalization in Section~\ref{sec:process-calculi-examples} where we directly represent standard calculi, and in Section~\ref{sec:advanc-data-struct} where we give examples of calculi with advanced data structures and computations on them, even nondeterministic reductions.

\section{Definitions}\label{sec:definitions}
Psi-calculi are based on nominal data types. 
A nominal data type is similar to a traditional data type, 
but can also contain binders and identify alpha-variants of terms. 
Formally, %
the only requirements are related to the treatment of the atomic symbols called names as explained below.
In this paper, we consider sorted nominal datatypes, where names and
members of the data type may have different sorts.

We assume a set of sorts $\sortset$.
Given a countable set of sorts for names
$\namesorts\subseteq\sortset$, we assume countably infinite
pair-wise disjoint sets of atomic \emph{names} $\nameset[s]$, where $s\in \namesorts$. 
The set of all names, $\nameset=\cup_s\nameset[s]$, is ranged over by $a,b,\ldots,x,y,z$.
We write $\ve{x}$ for a tuple of names $x_1,\dots,x_n$ and similarly for other tuples,
and $\ve{x}$ also stands for the set of names $\Set{x_1,\dots,x_n}$ if used where a set is expected.
We let $\pi$ range over permutations of tuples of names: 
$\pi\cdot\ve{x}$ is a tuple of names of the same length as $\ve{x}$,
containing the same names with the same multiplicities.

A sorted \emph{nominal set}~\cite{PittsAM:nomlfo-jv,Gabbay01anew} is a set equipped with \emph{name swapping} functions written $(a\;b)$, for any sort $s$ and names $a,b\in \nameset[s]$, 
i.e.~name swappings must respect sorting.
An intuition is that for any member~$T$ of a nominal set we have that $(a\;b)\cdot T$ is $T$ with  $a$ replaced by $b$ and  $b$ replaced by $a$.
The support of a term, written $\n(T)$, is intuitively the set of names that can be be affected by name swappings on $T$.
This definition of support coincides with the usual definition of free names for abstract syntax trees that may contain binders.
We write $a\freshin T$ for $a\not\in\n(T)$, and extend this to finite sets and tuples by conjunction.
A function $f$ is \emph{equivariant} if $(a\;b)\cdot(f(T))=f((a\;b)\cdot T)$ always holds;
a relation $\mathcal R$ is equivariant if $x\;\mathcal{R}\;y$ implies that $(a\;b)\cdot x\;\mathcal{R}\;(a\;b)\cdot y$ holds;
and a constant symbol $C$ is equivariant if $(a\;b)\cdot C = C$.
In particular, we require that all sorts $s\in\sortset$ are equivariant.
A \emph{nominal data type} is a nominal set together with some equivariant functions on it, for instance a substitution function. 

\subsection{Original Psi-calculi Parameters}\label{sec:psi-defs}
Sorted psi-calculi is an extension of the original psi-calculi framework~\cite{LMCS11.PsiCalculi},
which are given by three nominal datatypes (data terms, conditions and
assertions) as discussed in the introduction.
\begin{defi}[Original psi-calculus parameters]
\label{def:parameters1}
The psi-calculus parameters from the original psi-calculus are the following nominal data types:
(data) terms $M,N\in\trms$,
conditions $\varphi\in \conditions$, and
assertions $\Psi\in\assertions$;
equipped with the following four equivariant operators:
channel equivalence ${\sch} :\trms \times \trms \to \conditions$,
assertion composition $\ftimes: \assertions \times \assertions \to \assertions$,
the unit assertion $\emptyframe\in\assertions$, and
the entailment relation ${\vdash}\subseteq \assertions \times \conditions$.
\end{defi}
The binary functions $\sch$ and $\ftimes$ and the relation $\vdash$ above will be used in infix form. 
Two assertions are said to be equivalent, written $\Psi \sequivalent \Psi'$, 
if they entail the same conditions, i.e.~for all $\varphi$ we have that $\Psi \vdash \varphi \Leftrightarrow\Psi' \vdash \varphi$.  

We impose certain requisites on the sets and operators. In brief,
channel equivalence must be symmetric and transitive modulo entailment, 
the assertions with $(\ftimes,\one)$ must form an abelian monoid modulo $\sequivalent$,
and $\ftimes$ must be compositional w.r.t.~$\sequivalent$ 
(i.e.~$\Psi_1\sequivalent\Psi_2\implies\Psi\otimes\Psi_1\sequivalent\Psi\otimes\Psi_2$).
(For details see~\cite{LMCS11.PsiCalculi}, 
and for examples of machine-checked valid instantiations of the parameters see~\cite{pohjola10:verifyingPsi}.)
In examples in this paper, we usually consider the trivial assertion
monoid $\assertions=\{\one\}$, and let channel equivalence be term equality (i.e.
$\one\vdash M\sch N$ iff $M=N$).

\subsection{New parameters for generalized pattern-matching}
\label{sec:new-parameters}
To the parameters of the original psi-calculi we add patterns $X,Y$, that are used in input prefixes; 
a function $\vars$ which yields the possible combinations of binding names in the pattern,
and a pattern-matching function $\Match$, which is used when the input takes place.
Intuitively, an input pattern $(\lambda\ve{x})X$ matches a message $N$ if 
there are $\ve{L}\in\Match(N,\ve{x},X)$; the receiving agent then continues after substituting $\ve{L}$ for $\ve{x}$.
If $\Match(N,\ve{x},X)=\emptyset$ then $(\lambda\ve{x})X$ does not match $N$;
if $\abs{\Match(N,\ve{x},X)}>1$ then one of the matches will be
non-deterministically chosen.  
Below, we use ``variable'' for names that can be bound in a pattern.
\begin{defi}[Psi-calculus parameters for pattern-matching]
\label{def:parameters2}
The psi-calculus parameters for pattern-matching include the nominal data type
$\pats$ of (input) patterns, ranged over by $X,Y$,
and the two equivariant operators
\[\begin{array}{rcll}
\Match&:&\trms\times\nameset^*\times\pats\to\PowFin(\trms^*)
 & \mbox{Pattern matching}\\
\vars&:&\pats\to\PowFin(\PowFin(\nameset))
 & \mbox{Pattern variables}\\
\end{array}
\]
\end{defi}

The $\vars$ operator gives the possible (finite) sets of names in a pattern which are bound by an input prefix. 
For example, we may want an input prefix with a pairing pattern
$\langle x,y\rangle$ to be able to bind both $x$ and $y$, 
only one of them, or none, and so we define $\vars(\langle x,y\rangle)=\Set{\Set{x,y},\Set{x},\Set{y},\Set{}}$. 
This way, we can let the input prefix $\lin{c}{x}\langle x,y\rangle$ only match pairs where the second argument is the name $y$.
To model a calculus where input patterns cannot be selective in this way,
we may instead define $\vars(\langle x,y\rangle)=\Set{\Set{x,y}}$.
This ensures that input prefixes that use the pattern $\langle x,y\rangle$ 
must be of the form $\lin{M}{x,y}\langle x,y\rangle$, where both $x$
and $y$ are bound.
Another use for $\vars$ is to exclude the binding of terms in certain
positions, such as the keys of cryptographic messages (cf.~Section~\ref{sec:pmSpi}).

Requisites on $\vars$ and $\Match$ are given below in Definition~\ref{def:pattern-match}.
Note that the four data types \trms, \conditions, \assertions\ and \pats\ are not required to be disjoint. 
In most of the examples in this paper the patterns \pats\ is a subset of the terms \trms.

\subsection{New parameters for sorting}
\label{sec:parameters-sorting}
To the parameters defined above we add a sorting function and four sort compatibility predicates.
\begin{defi}[Psi-calculus parameters for sorting]
\label{def:parameters3}
The psi-calculus parameters for sorting include the equivariant sorting function
$\SORT:\nameset\uplus\trms\uplus\pats\to\sortset$,
 and the four compatibility predicates
\[\begin{array}{rclll}
{\CanReceive}   &\subseteq& \sortset\times\sortset &\quad& \mbox{can be used to  receive,} \\
{\CanSend}      &\subseteq& \sortset\times\sortset && \mbox{can be used to send,} \\
{\CanSubstitute}&\subseteq& \sortset\times\sortset && \mbox{can be substituted by,} \\
\Restrictionsorts &\subseteq&\namesorts && \mbox{can be bound by name restriction.} 
\end{array}
\]
\end{defi}
The $\SORT$ operator gives the sort of a name, term or pattern; 
on names we require that $\sort a = s$ iff $a\in\nameset[s]$.  
This is similar to Church-style lambda-calculi, where each well-formed
term has a unique type. 

The sort compatibility predicates are used to restrict where terms and names of certain sorts may appear in processes.
Terms of sort $s$ can be used to send values of sort $t$ if $s\CanSend t$. 
Dually, a term of sort $s$ can be used to receive with a pattern of sort $t$ if $s\CanReceive t$. 
A name $a$ can be used in a restriction $(\nu{a})$ if $\sort{a}\in\Restrictionsorts$.
If $\sort{a}\CanSubstitute\sort{M}$ we can substitute the term $M$ for the name $a$.
In most of our examples, $\CanSubstitute$ is a subset of the equality relation.
These predicates can be chosen freely, although the set of well-formed substitutions
depends on $\CanSubstitute$, as detailed in Definition~\ref{def:subst} below.
\subsection{Substitution and Matching}
\label{sec:subst-match}
We require that each datatype is equipped with an equivariant substitution function, which intuitively substitutes terms for  names.  
The requisites on substitution differ from the original psi-calculi as indicated in the Introduction. %
Substitutions must preserve or refine sorts, 
and bound pattern variables must not be removed by substitutions.

We define two usage preorders $\le_{\trms}$ and $\le_{\pats}$ on $\sortset$.  
Intuitively, $s_1\le_{\trms} s_2$ 
  if terms of sort~$s_1$ can be used as a channel or message whenever~$s_2$ can be,
and $s_1\le_{\pats} s_2$ 
  if patterns of sort~$s_1$ can be used whenever~$s_2$ can be.
Formally~$s_1\le_{\trms} s_2$ iff  $\forall t\in\sortset. 
  (s_2\CanReceive t \Rightarrow s_1\CanReceive t) \land
  (s_2\CanSend t \Rightarrow s_1\CanSend t)\land 
  (t\CanSend s_2 \Rightarrow t\CanSend s_1)$.
Similarly, we define~$s_1\le_{\pats} s_2$ iff  $\forall t\in\sortset. 
  (t\CanReceive s_2 \Rightarrow t\CanReceive s_1)$.

Intuitively, substitutions must map every term of sort $s$ to a
term of some sort $s'$ with $s'\le_{\trms} s$ and similarly for patterns,
or else a sort compatibility predicate may be violated.
The usage preorders compare the sorts of terms (resp.~patterns), 
and so do not have any formal relationship to $\CanSubstitute$ 
(which relates the sort of a name to the sort of a term). 
In particular, $\CanSubstitute$ is not used in the definition of usage
preorders.
\begin{defi}[Requisites on substitution]\label{def:subst}
   If $\ve{a}$ is a sequence of distinct names 
   and~$\ve{N}$ is an equally long sequence of terms 
   such that $\sort{a_i}\CanSubstitute\sort{N_i}$ for all $i$, 
   we say that $\lsubst{\ve{N}}{\ve{a}}$ is a \emph{sub\-sti\-tution}. 
   Substitutions are ranged over by $\sigma$.

   For each data type among $\trms,\assertions,\conditions$ 
   we define an equivariant substitution operation on members $T$ of that data type as follows: 
   we require 
    that $T\sigma$ is an member of the same data type, and 
    that if $(\ve a\ \ve b)$ is a (bijective) name swapping 
    such that $\ve{b}\freshin T,\ve{a}$ 
    then $T\lsubst{\ve{N}}{\ve{a}}= ((\ve a\ \ve b)\cdot T)\lsubst{\ve{N}}{\ve{b}}$ 
    (alpha-renaming of substituted variables).
    For terms %
    we additionally require that 
    $\sort{M\sigma}\le_{\trms}\sort{M}$.

    For patterns $X\in\pats$, 
    we require that substitution is equivariant, that $X\sigma\in\pats$,
    and that if $\ve{x}\in\vars(X)$ and~$\ve{x}\freshin\sigma$ 
    then $\sort{X\sigma}\le_{\pats}\sort{X}$ and $\ve{x}\in\vars(X\sigma)$
    and alpha-renaming of substituted variables (as above) holds for $\sigma$ and $X$.

\end{defi}

Intuitively, the requirements on substitutions on patterns ensure that a
substitution on a pattern with binders 
$((\lambda\ve{x})X)\sigma$ with $\ve{x}\in\vars(X)$ and $\ve{x}\freshin\sigma$
yields a pattern $(\lambda\ve{x})Y$ with $\ve{x}\in\vars(Y)$.
As an example, consider the pair patterns discussed above with 
$\pats=\{\langle x,y\rangle\,:\,x\neq y\}$ and 
$\vars(\langle x,y\rangle)=\Set{\Set{x,y}}$.
We can let $\langle x,y\rangle\sigma=\langle x,y\rangle$ 
when $x,y\freshin\sigma$.
Since $\vars(\langle x,y\rangle)=\Set{\Set{x,y}}$ the pattern $\langle x,y\rangle$ in a well-formed agent will always occur directly under the binder $(\lambda x,y)$, i.e. as $(\lambda x,y) \langle x,y\rangle$, and here a substitution for $x$ or $y$ will have no effect. It therefore does not matter what e.g. $\langle x,y\rangle[x:=M]$ is, since it will never occur in derivations of transitions of well-formed agents.
We could think of substitutions as partial functions which are undefined in such cases; formally,
since substitutions are total, the result of this substitution can be assigned an arbitrary value.

In the original psi-calculi 
there is no requirement that substitution preserves names that are
used as input variables (i.e., $\n(N\sigma)\supseteq\n(N)\setminus\n(\sigma)$).
As seen in the introduction, this means that the original psi semantics does not always preserve the well-formedness of agents
(an input prefix $\lin{M}{\ve{x}}N\sdot P$ is well-formed when $\ve{x}\subseteq\n(N)$)
although this is assumed by the operational semantics~\cite{LMCS11.PsiCalculi}.
In pattern-matching psi-calculi, substitution on patterns is required to
preserve variables, and the operational semantics does
preserve well-formedness as shown below in Theorem~\ref{thm:wfPres}.

Matching must be invariant under renaming of pattern variables, 
and the substitution resulting from a match can only mention names 
that are from the matched term or the pattern.
\begin{defi}[Requisites on pattern matching]
\label{def:pattern-match}
  For the function $\Match$ we require that 
  if $\ve{x}\in\vars(X)$ are distinct 
  and $\ve{N}\in\Match(M,\ve{x},X)$ 
  then it must hold that 
    $\lsubst{\ve{N}}{\ve{x}}$ is a substitution,  that
    $\n(\ve{N})\subseteq\n(M)\cup(\n(X)\setminus\ve{x})$, and that
    for all name swappings $(\ve x\ \ve y)$ with $\ve y \freshin X$
        we have $\ve{N}\in\Match(M,\ve{y},(\ve x\ \ve y)\cdot X)$ 
      (alpha-renaming of matching).
\end{defi}

In many process calculi, and also in the symbolic semantics of psi~\cite{JLAP12.SymbolicPsi}, the input construct binds a single variable. 
This is a trivial instance of pattern matching where the pattern is a single bound variable, matching any term.

\begin{exa} \label{ex:symsempats}
Given values for the other requisites, 
we can take $\pats = \nameset$ with $\vars(a)=\Set{a}$, 
meaning that the pattern variable must always occur bound, 
and $\Match(M,a,a)=\Set{M}$ if $\sort a\CanSubstitute\sort M$.
On patterns we define substitution as $a\sigma=a$.
\end{exa}

When all substitutions on terms preserve names, 
we can recover the pattern matching of the original psi-calculi. 
Such psi-calculi also enjoy well-formedness preservation (Theorem~\ref{thm:wfPres}).
\begin{thm}\label{thm:origpats}
  Suppose $(\trms, \conditions, \assertions)$ is an original psi-calculus~\cite{LMCS11.PsiCalculi}
  where $\n(N\sigma)\supseteq\n(N)\setminus\n(\sigma)$ for all $N$, $\sigma$.
  Let $\pats=\trms$ and
  $\vars(X)=\Pow(\n(X))$ and
  $\Match(M,\ve{x},X)=\Set{\ve{L}: M=X\lsubst{\ve{L}}{\ve{x}}}$ and
  $\sortset=\namesorts=\Restrictionsorts = \Set{s}$ and
  ${\CanReceive} = {\CanSend} = {\CanSubstitute} = \Set{(s,s)}$
  and $\SORT:\nameset\uplus\trms\uplus\pats\to\Set s$;
  then $(\trms, \pats, \conditions, \assertions)$ is a sorted psi-calculus.
\end{thm}
\proof Straightforward; this result has been checked in Isabelle. \qed

\subsection{Agents}
\newcommand{\discard}[1]{}
\label{sec:agents}
\begin{defi}[Agents]\label{def:agents}
The \emph{agents}, ranged over by $P,Q,\ldots$,  are of the following forms.
{\rm
\[
\begin{array}{ll}

\out{M}N .P                   & \mbox{Output} \\
\lin{M}{\ve{x}}X.P          & \mbox{Input}\\
\caseonly\;{\ci{\varphi_1}{P_1}\casesep\cdots\casesep\ci{\varphi_n}{P_n}}
&\mbox{Case} \\
(\nu a)P                      & \mbox{Restriction}\\
P \pll Q                      & \mbox{Parallel}\\
! P                           & \mbox{Replication} \\
\pass{\Psi}                        & \mbox{Assertion}
\end{array}\]
}

In the Input all names in $\ve{x}$ bind their occurrences in both $X$ and~$P$, and in the Restriction $a$ binds in P.
Substitution on agents is defined inductively on their structure, using the substitution function of each datatype based on syntactic position, avoiding name capture.
\end{defi}
The output prefix $\out{M}N .P$ sends $N$ on a channel that is equivalent to $M$.
Dually, $\lin{M}{\ve{x}}X.P$ receives a message matching the pattern $X$ from a channel equivalent to $M$.
A non-deterministic case statement  $\caseonly\;{\ci{\varphi_1}{P_1}\casesep\cdots\casesep\ci{\varphi_n}{P_n}}$
executes one of the branches $P_i$ where the corresponding condition $\varphi_i$ holds, discarding the other branches.
Restriction $(\nu a)P$ scopes the name $a$ in $P$; the scope of $a$ may be extruded if $P$ communicates a data term containing $a$.
A parallel composition $P \pll Q$ denotes $P$ and $Q$ running in parallel; they may proceed independently or communicate.
A replication $! P$ models an unbounded number of copies of the process~$P$.
The assertion $\pass{\Psi}$ contributes $\Psi$ to its environment.
We often write $\ifthen{\varphi}{P}$ for $\caseonly\;\varphi:P$, and nothing or $\nil$ for the empty case statement $\caseonly$.

In comparison to~\cite{LMCS11.PsiCalculi}  we additionally restrict the syntax of well-formed agents by imposing requirements on sorts: 
the subjects and objects of prefixes must have compatible sorts, and restrictions may only bind names of a sort in $\Restrictionsorts$.
\begin{defi}
\label{def:sorted-psi-well-formedness}
An occurrence of an assertion is \emph{unguarded} if it is not a subterm of an Input or
Output. An agent is \emph{well-formed} if, for all its subterms, 
\begin{enumerate}
\item in a replication $!P$ there are no unguarded assertions in $P$; and
\item in $\caseonly\;{\ci{\varphi_1}{P_1}\casesep\cdots\casesep\ci{\varphi_n}{P_n}}$
  there is no unguarded assertion in any $P_i$; and
\item\label{item:out} in an Output $\out{M}N.P$
  we require that $\sort M\CanSend \sort N$; and
\item\label{item:in} in an Input
  $\lin{M}{\ve{x}}X.P$ we require that
  \begin{enumerate}
  \item\label{item:invars} $\ve{x}\in\vars(X)$ is a tuple of distinct names and
  \item\label{item:insort} $\sort M\CanReceive\sort X$; and
  \end{enumerate}
\item\label{item:res} in a Restriction
  $(\nu a)P$ we require that $\sort{a}\in\Restrictionsorts$.
\end{enumerate}
\end{defi}
Requirements~\ref{item:out},~\ref{item:insort} and~\ref{item:res} are new for sorted psi-calculi.

\subsection{Frames and transitions}
\label{sec:frames-transitions}
Each agent affects other agents that are in parallel with it via its frame, 
which may be thought of as the collection of all top-level assertions of the agent.
A \emph{frame}  $F$ is an assertion with local names, written $\framepair{\frnames{}}{\Psi}$ where $\frnames{}$ is a sequence of names that bind into the assertion~$\Psi$. We use $F,G$ to range over frames, and identify alpha-equivalent frames.
We overload 
$\ftimes$ to frame composition defined by $\framepair{\frnames{1}}{\Psi_1} \ftimes \framepair{\frnames{2}}{\Psi_2} = 
\framepair{\frnames{1} \frnames{2}}{(\Psi_1 \ftimes \Psi_2)}$ where 
$\frnames{1}\freshin\frnames{2},\Psi_2$ and vice versa. 
We write
$\Psi \ftimes F$ to mean $\framepair{\epsilon}{\Psi} \ftimes F$, and
$(\nu c)(\framepair{\frnames{}}{\Psi})$ for~$\framepair{c\frnames{}}{\Psi}$.

Intuitively a condition is entailed by a frame if it is entailed by
the assertion and does not contain any names bound by the frame, and
two frames are equivalent if they entail the same conditions.
Formally, we define $F \vdash \varphi$ to mean that there exists an alpha variant  $\framepair{\frnames{}}{\Psi}$ of $F$ such that  $\frnames{} \freshin \varphi$ and $\Psi \vdash \varphi$. We also define~\mbox{$F\sequivalent G$} to mean that for all $\varphi$ it holds that $ F \vdash \varphi$ iff $ G \vdash \varphi$.

\begin{defi}[Frames and Transitions]
\label{def:transitions}
 The \emph{frame $\frameof{P}$ of an agent} P is defined inductively as follows:
\[\begin{array}{c}
\frameof{\pass{\Psi}} = \framepair{\epsilon}{\Psi} \qquad\qquad
\frameof{P \pll Q} = \frameof{P} \ftimes \frameof{Q}\qquad\qquad
\frameof{\res{b}P} = \res{b}\frameof{P}\\[.3em]
\frameof{\lin{M}{\ve{x}}{N} \sdot P} = \frameof{\out{M}{N} \sdot P} = 
\frameof{\caseonly\;{\ci{\ve{\varphi}}{\ve{P}}}} = 
\frameof{!P} = \one
 \end{array}\]

\noindent The \emph{actions} ranged over by $\alpha, \beta$ are of the following three kinds:
Output $\boutlabel{M}{\ve{a}}{N}$ where $\tilde{a} \subseteq \names{N}$,
Input $\inlabel{M}{N}$, and Silent $\tau$. Here we refer to $M$ as the {\em
subject} and $N$ as the \emph{object}. We define 
$\bn{\boutlabel{M}{\tilde{a}}{N}} = \tilde{a}$, and $\bn{\alpha}=\emptyset$
if $\alpha$ is an input or $\tau$.
We also define $\names{\tau} = \emptyset$ and $\names{\alpha} = \names{M} \cup \names{N}$ for the input and output actions.
We write $\noboutlabel MN$ for $\boutlabel{M}{\varepsilon}N$.

A \emph{transition} is written \mbox{$\framedtransempty{\Psi}{P}{\alpha}{P'}$},
meaning that in the environment $\Psi$ the well-formed agent $P$
can do an $\alpha$ to become $P'$.  The transitions are defined inductively in 
Table~\ref{table:struct-free-labeled-operational-semantics}. We write $\trans{P}{\alpha}{P'}$ without an assertion to mean $\emptyframe \frames \trans{P}{\alpha}{P'}$.
\end{defi}
\begin{table*}[tb]

\begin{mathpar}

\inferrule*[Left=\textsc{In}]
    {\Psi \vdash M \sch K \quad \ve{L}\in\Match(N,\ve y, X)}
    {\Psi\frames\trans{\lin{M}{\ve{y}}{X}.P}{\inn{K}{N}}{P\lsubst{\ve{L}}{\ve{y}}}}

\inferrule*[left=\textsc{Out}]
    {\Psi \vdash M \sch K }
    {\Psi\frames\trans{\out{M}{N}.P}{\noboutlabel{K}{N}}{P}}

\inferrule*[left=\textsc{Com}, right={$\ve{a} \freshin Q$}]
 {\frass{Q} \ftimes \Psi\frames\trans{P}{\boutlabel{M}{\ve{a}}{N}}{P'} \\
  \frass{P} \ftimes \Psi\frames\trans{Q}{\inn{K}{N}}{Q'} \\
  \Psi \ftimes \frass{P} \ftimes \frass{Q} \vdash M \sch K
  }
       {\Psi\frames\trans{P \pll Q}{\tau}{(\nu \ve{a})(P' \pll Q')}}

\inferrule*[left=\textsc{Par},  right={$\bn{\alpha} \freshin Q$
}]
{\frass{Q} \ftimes \Psi\frames\trans{P} {\alpha}{P'}}
{\Psi\frames\trans{P\mid Q}{\alpha}{P' \mid Q}}

\inferrule*[left={\textsc{Case}}]
    {\Psi\frames\trans{P_i}{\alpha}{P'} \\ \Psi \vdash \varphi_i}
    {\Psi\frames\trans{\caseonly\;{\ci{\ve{\varphi}}{\ve{P}}}}{\alpha}{P'}}

\inferrule*[left=\textsc{Rep}]
   {\Psi\frames\trans{P \pll !P}{\alpha}{P'}}
   {\Psi\frames\trans{!P}{\alpha}{P'}}

\inferrule*[left=\textsc{Scope}, right={$b \freshin \alpha,\Psi$}]
    {\Psi\frames\trans{P}{\alpha}{P'}}
    {\Psi\frames\trans{(\nu b)P}{\alpha}{(\nu b)P'}}

\inferrule*[left=\textsc{Open}, right={$\inferrule{}{b \freshin \ve{a},\Psi,M\\\\
b \in \n(N)}$}]
    {\Psi\frames\trans{P}{\boutlabel{M}{\ve{a}}{N}}{P'}}
    {\Psi\frames\trans{(\nu b)P}{\boutlabel{M}{\ve{a} \cup \{b\}}{N}}{P'}}

\end{mathpar}\smallskip
{Symmetric versions of \textsc{Com}
and \textsc{Par} are elided. In the rule $\textsc{Com}$ we assume that $\fr{P} =
\framepair{\frnames{P}}{\frass{P}}$ and   $\fr{Q} =
\framepair{\frnames{Q}}{\frass{Q}}$ where $\frnames{P}$ is fresh for all of 
$\Psi, \frnames{Q}, Q, M$ and $P$, and that $\frnames{Q}$ is correspondingly
fresh. In the rule
\textsc{Par} we assume that $\fr{Q} = \framepair{\frnames{Q}}{\frass{Q}}$
where $\frnames{Q}$ is fresh for
$\Psi, P$ and $\alpha$. 
In $\textsc{Open}$ the expression $\nu\tilde{a} \cup \{b\}$ means the sequence $\tilde{a}$ with $b$ inserted anywhere.
}
\caption{Operational semantics.}  
\label{table:struct-free-labeled-operational-semantics}
\end{table*}

The operational semantics, defined in Table~\ref{table:struct-free-labeled-operational-semantics}, 
is the same as for the original psi-calculi, except for the use of $\Match$ in rule \textsc{In}.
We identify alpha-equivalent agents and transitions (see~\cite{LMCS11.PsiCalculi} for details). 
In a transition the names in $\bn{\alpha}$ bind into both the action object and the derivative, therefore $\bn{\alpha}$ is in the support of $\alpha$ but not in the support of the transition.
This means that the bound names can be chosen fresh, substituting each
occurrence in both the action and the derivative.

As shown in the introduction, well-formedness is not preserved by transitions in the original psi-calculi.
However, in sorted psi-calculi the usual well-formedness preservation result holds. 
\label{sec:wfPres}
\begin{thm}[Preservation of well-formedness]\label{thm:wfPres}
  If $P$ is well-formed, then
  \begin{enumerate}
  \item $P\sigma$ is well-formed; and
  \item if
    $\Psi\frames\trans{P}{\alpha}{P'}$ then $P'$ is well-formed.
  \end{enumerate}
\end{thm}
\proof
The first part is by induction on $P$.  The output prefix case uses the sort preservation property of substitution on terms (Definition~\ref{def:subst}). 
    The interesting case is input prefix $\lin{M}{\ve{x}}X.Q$: 
    assume that $Q$ is well-formed, that $\ve{x}\in\vars(X)$, 
    that $\sort M\CanReceive\sort X$ and that $\ve{x}\freshin\sigma$.
    By induction $Q\sigma$ is well-formed.
    By sort preservation we get $\sort{M\sigma}\le\sort M$, 
    so $\sort{M\sigma}\CanReceive\sort{X}$.
    By preservation of patterns by non-capturing substitutions we have that
    $\ve{x}\in\vars(X\sigma)$ and $\sort{X\sigma}\le\sort X$, so
    $\sort{M\sigma}\CanReceive\sort{X\sigma}$.

The second part is by induction on the transition rules, using part 1 in the \textsc{In} rule.\qed

\noindent Since well-formedness is preserved by transitions and substitutions, 
from this point on we only consider well-formed agents.

\section{Meta-theory}
\label{sec:meta-theory}
As usual, the labelled operational semantics gives rise to notions of
labelled bisimilarity.  Similarly to the applied
pi-calculus~\cite{abadi.fournet:mobile-values}, the standard
definition of bisimilarity needs to be adapted to take assertions into
account.  In this section, we show that both strong and weak
bisimilarity satisfy the expected structural congruence laws and the
standard congruence properties of name-passing process calculi.
We first prove these results for calculi with a single name sort
(Theorem~\ref{thm:allresults}) supported by Nominal Isabelle.
We then extend the results to all sorted psi-caluli
(Theorems~\ref{thm:sortstructcong}, \ref{thm:congruence-bisim}, and~\ref{thm:congruence-congruence}) by manual proofs.

\subsection{Recollection}
\label{sec:recollection}
We start by recollecting the required definitions, beginning with the
definition of strong labelled bisimulation on well-formed agents by
Bengtson et al.~\cite{LMCS11.PsiCalculi}, to which we refer for
examples and more intuitions.

\begin{defi}[Strong bisimulation]\label{def:strongbisim}
A \emph{strong bisimulation}
 $\mathcal R$ is a ternary relation on assertions and pairs of agents such that
 ${\mathcal R}(\Psi,P,Q)$ implies the following four statements.
 \begin{enumerate}%
 \item \label{item:static} Static equivalence:
  $\Psi \ftimes \frameof{P} \sequivalent \Psi \ftimes \frameof{Q}$.
 \item
   Symmetry: ${\mathcal R}(\Psi,Q,P)$.
 \item \label{item:extension}
 Extension with arbitrary assertion: for all $\Psi'$ it holds that ${\mathcal R}(\Psi \ftimes
\Psi',P,Q)$.
 \item   Simulation:
for all $\alpha, P'$ such that
$\bn{\alpha}\freshin \Psi,Q$
and $\Psi\frames \trans{P}{\alpha}{P'}$,\\ there exists $Q'$ such
that $\Psi\frames\trans{Q}{\alpha}{Q'}$ and ${\mathcal R}(\Psi , P', Q')$.
\end{enumerate}
 \label{def:bisim}
We define \emph{bisimilarity} $P \bisim_\Psi Q$ to mean that there is a bisimulation ${\mathcal R}$ such that
${\mathcal R}(\Psi,P,Q)$, and write $\bisim$ for $\bisim_\emptyframe$. %
\end{defi}
Above, (\ref{item:static}) corresponds to the capability of a parallel
observer to test the truth of a condition using $\caseonly$, 
while (\ref{item:extension}) models an observer taking a step and adding
a new assertion $\Psi'$ to the current environment.

We close strong bisimulation under substitutions to obtain a congruence.
\begin{defi}[Strong bisimulation congruence]\label{def:strongcongruence}
$P\sim_\Psi Q$ means that for all sequences $\ve\sigma$ of substitutions 
it holds that
$P\ve\sigma \bisim_\Psi Q\ve\sigma$.  We write $P \sim Q$ for $P \sim_\emptyframe Q$.
\end{defi}
To illustrate the definitions of bisimulation and bisimulation congruence, 
we here prove a result about the $\caseonly$ statement, 
to be used in Section~\ref{sec:process-calculi-examples}.

\begin{lem}[Flatten Case]
\label{lemma:flatten-case}
Suppose that there exists a condition
$\top\in \conditions$ such that $\Psi\vdash\top\ve{\sigma}$ for all $\Psi$ and substitution sequences $\ve\sigma$. 
Let $R =\caseonly\;\ci{\top}{(\caseonly\;\ci{\ve\varphi}{\ve{P}})}\casesep\ci{\ve{\phi}}{\ve{Q}}$ and $R' = \caseonly\;\ci{\ve\varphi}{\ve{P}}\casesep\ci{\ve\phi}{\ve{Q}}$; then $R\sim R'$.
\end{lem}
\proof
We let $\mathcal{I}:=\bigcup_{\Psi,P}\Set{(\Psi,P,P)}$ be the identity relation, and
\[
\mathcal{S} := \bigcup_{\Psi,\ve{P},\ve{Q},\ve{\phi},\ve{\varphi}}
\begin{array}[t]{rr}\{(\Psi,\caseonly\;\ci{\varphi_\top}{(\caseonly\;\ci{\ve\varphi}{\ve{P}})}\casesep\ci{\ve{\phi}}{\ve{Q}},\caseonly\;\ci{\ve\varphi}{\ve{P}}\casesep\ci{\ve\phi}{\ve{Q}}) :{} \\
\varphi_\top\in\conditions\land \forall \Psi'\in\assertions.\; \Psi' \vdash \varphi_\top\}.
\end{array}
\]
We prove that $\mathcal{T}:={\mathcal{S}}\, \cup\, {\mathcal{S}^{-1}}\, \cup\, \mathcal{I}$ is a bisimulation, 
where $\mathcal{S}^{-1} := \{(\Psi,Q,P) : (\Psi,P,Q) \in \mathcal{S}\}$.
Then, $\mathcal{T}(\emptyframe,R\ve\sigma,R'\ve\sigma)$ for all $\ve\sigma$, 
so $R\sim R'$ by the definition of $\sim$.
The proof that $\mathcal{T}$ is a bisimulation is straightforward:

\begin{description}
\item[Static equivalence]
The frame of a $\mathbf{case}$ agent is always $\emptyframe$, 
hence static equivalence follows by reflexivity of $\simeq$.
\item[Symmetry] Follows by definition of $\mathcal{T}$.
\item[Extension with arbitrary assertion]
Trivial by the choice of candidate relation, since the $\Psi$ in $\mathcal{S}$ and $\mathcal{I}$ are universally quantified.
\item[Simulation] Trivially, any process $P$ simulates itself.
  Fix $(\Psi,R,R') \in \mathcal{S}$, such that 
  $R=\caseonly\;\ci{\varphi_\top}{(\caseonly\;\ci{\ve\varphi}{\ve{P}})}\casesep\ci{\ve{\phi}}{\ve{Q}}$
  and $R'=\caseonly\;\ci{\ve\varphi}{\ve{P}}\casesep\ci{\ve\phi}{\ve{Q}}$. 
  Here $\Psi\vdash\varphi_\top$ follows by definition of $\mathcal{S}$. 
  Since $\mathcal{T}$ includes both $\mathcal{S}$ and $\mathcal{S}^{-1}$, 
  we must follow transitions from both $R$ and $R'$.
\begin{itemize}
\item 
A transition from $R$ via $P_i$ can be derived as follows:
\begin{mathpar}
\inferrule*[Left=Case]
    {\inferrule*[Left=Case]
        {\Psi\frames\trans{P_i}{\alpha}{P'_i} \and \Psi\vdash\varphi_i}
        {\Psi\frames\trans{\caseonly\;\ci{\ve\varphi}{\ve{P}}}{\alpha}{P'_i\and \Psi\vdash\varphi_\top}}}
    {\Psi\frames\trans{\caseonly\;\ci{\varphi_\top}{(\caseonly\;\ci{\ve\varphi}{\ve{P}})}\casesep\ci{\ve\phi}{\ve{Q}}}{\alpha}{P'_i}}
\end{mathpar}
Then $R'$ can simulate this with the following derivation:
\begin{mathpar}
\inferrule*[Left=Case]
    {\Psi\frames\trans{P_i}{\alpha}{P'_i} \and \Psi\vdash\varphi_i}
    {\Psi\frames\trans{\caseonly\;\ci{\ve\varphi}{\ve{P}}\casesep\ci{\ve\phi}{\ve{Q}}}{\alpha}{P'_i}}
\end{mathpar}
Since $\mathcal{I}(\Psi,P'_i,P'_i)$ and $\mathcal{I} \subseteq \mathcal{T}$ we have $\mathcal{T}(\Psi,P'_i,P'_i)$.

\item A transition from $R'$ via $Q_i$  can be derived as follows:
\begin{mathpar}
\inferrule*[Left=Case]
    {\Psi\frames\trans{Q_i}{\alpha}{Q'_i} \and \Psi\vdash\phi_i}
    {\Psi\frames\trans{\caseonly\;\ci{\ve\varphi}{\ve{P}}\casesep\ci{\ve\phi}{\ve{Q}}}{\alpha}{Q'_i}}
\end{mathpar}
The process $R$ can simulate this with the following derivation: 
\begin{mathpar}
\inferrule*[Left=Case]
    {\Psi\frames\trans{Q_i}{\alpha}{Q'_i} \and \Psi\vdash\phi_i}
    {\Psi\frames\trans{\caseonly\;\ci{\varphi_\top}{(\caseonly\;\ci{\ve\varphi}{\ve{P}})}\casesep\ci{\ve\phi}{\ve{Q}}}{\alpha}{Q'_i}}
\end{mathpar}
Since $\mathcal{I}(\Psi,Q'_i,Q'_i)$ and $\mathcal{I} \subseteq \mathcal{T}$ we have $\mathcal{T}(\Psi,Q'_i,Q'_i)$.

\item Symmetrically, $R'$ can simulate transitions derived from $R$ via $Q_i$, 
  and $R$ can simulate transitions derived from $R'$ via $P_i$. \qed
\end{itemize}
\end{description}

Psi-calculi are also equipped with a notion of weak bisimilarity
($\wbisim$) where $\tau$-transitions cannot be observed, introduced by
Bengtson et al.~\cite{LICS10.WeakPsi}. We here restate its definition,
but refer to the original publication for examples and more
motivation.

The definition of weak transitions is standard.
\begin{defi}[Weak transitions] \label{def:weakTrans}
$\Psi \frames \transw{P}{}{P'}$ is defined inductively by the rules:
\begin{enumerate}
\item $\Psi \frames \transw{P}{}{P}$
\item If $\Psi \frames \trans{P}{\tau}{P''}$ and $\Psi \frames \transw{P''}{}{P'}$, then $\Psi \frames \transw{P}{}{P'}$
\end{enumerate}
\end{defi}
For weak bisimulation we use static implication (rather than static equivalence) to compare the frames of the process pair under consideration.
\begin{defi}[Static implication]\label{def:staticimpl}
$P$ \emph{statically implies} $Q$ in the environmental assertion $\Psi$, written $P \simplies_\Psi Q$, if
\[\forall \varphi. \; \Psi \ftimes \frameof{P} \vdash \varphi \;\Rightarrow \;
\Psi \ftimes \frameof{Q} \vdash \varphi\]
\end{defi}

\begin{defi}[Weak bisimulation]\label{def:weak-bisim}
A \emph{weak bisimulation}
 $\mathcal{R}$ is a ternary relation between assertions and pairs of agents such that
 ${\mathcal{R}}(\Psi,P,Q)$ implies all of
 \begin{enumerate}
 \item \label{item:wstatic}
   Weak static implication: for all $\Psi'$ there exist $Q', Q''$ such that
   \[ \Psi \frames \transw{Q}{}{Q'}
   \quad\wedge\quad
   \Psi \ftimes \Psi' \frames \transw{Q'}{}{Q''}  
   \quad\wedge\quad 
   P \simplies_\Psi Q' 
   \quad\wedge\quad 
   {\mathcal R}(\Psi\ftimes\Psi',P,Q'')
   \]
  \item
   Symmetry: ${\mathcal{R}}(\Psi,Q,P)$
 \item
 Extension of arbitrary assertion: for all $\Psi'$ it holds that
 ${\mathcal{R}}(\Psi \ftimes \Psi',P,Q)$
 \item \label{item:wsim}  Weak simulation: for all $P'$, if $\framedtransempty{\Psi}{P}{\alpha}{P'}$ then
   \begin{enumerate}
   \item \label{item:wsimtau} if $\alpha = \tau$ then $\exists Q' .
       \;\Psi \frames \transw{Q}{}{Q'} \wedge {\mathcal
         R}(\Psi,P',Q')$; and
     \item \label{item:wsimvis} if $\alpha\neq\tau$ and $\bn{\alpha}\freshin \Psi,Q$, then there exists $Q',Q'', Q'''$ such that
     \[\begin{array}{rlcl}
       & 
       \Psi \frames  \transw{Q}{}{Q'}
       &\land&
       \Psi \frames \trans{Q'}{\alpha}{Q''} 
       \quad\wedge\quad
       \Psi \ftimes \Psi' \frames \transw{Q''}{}{Q'''} 
       \\ &
       {}\land\quad
       P \simplies_\Psi Q' 
       &\wedge&
       {\mathcal{R}}(\Psi\ftimes\Psi',P',Q''') 
     \end{array}\]
   \end{enumerate}
\end{enumerate}
\label{def:wbisim}
We define $P \wbisim Q$ to mean that there exists a weak bisimulation ${\mathcal{R}}$
such that ${\mathcal{R}}(\emptyframe,P,Q)$ and we write $P \wbisim_\Psi Q$
when 
there exists a weak bisimulation $\R$ such that $\R(\Psi,P,Q)$.
\end{defi}
Above, (\ref{item:wstatic}) allows $Q$ to take $\tau$-transitions before
and after enabling at least those conditions that hold in the frame of $P$, as per Definition~\ref{def:staticimpl}.
Moreover, when testing these conditions, the observer may also add an 
assertion $\Psi'$ to the environment.
In (\ref{item:wsimvis}), the observer may test the validity of conditions
when matching a visible transition, and may also add an assertion as above.

To obtain a congruence from weak bisimulation, we must require that
every $\tau$-transition is simulated by a weak transition containing at least one $\tau$-transition.

\begin{defi}\label{def:weakcongruence}
  A \emph{weak $\tau$-bisimulation}
  $\mathcal{R}$ is a ternary relation between assertions and pairs of agents such that
  ${\mathcal{R}}(\Psi,P,Q)$ implies all conditions of
  a weak bisimulation (Definition~\ref{def:wbisim}) with \ref{item:wsimtau} replaced
  by
  \[ \text{(\ref{item:wsimtau}${}'$) if }\alpha=\tau\text{ then }\exists Q',Q'' .
       \;\Psi \frames \trans{Q}{\tau}{Q'} \land \Psi \frames
       \transw{Q'}{}{Q''} \wedge P' \wbisim_\Psi Q''. \]
  We then let $P\approx_\Psi Q$ mean that for all sequences $\ve\sigma$ of substitutions 
  there is a weak $\tau$-bisimulation
  $\mathcal{R}$ such that $\mathcal{R}(\Psi,P\ve\sigma, Q\ve\sigma)$.
  We write $P \approx Q$ for $P \approx_\emptyframe Q$.
\end{defi}

\begin{lem}[Comparing bisimulations]\label{lem:comp-bisim}
  For all relations $\mathcal{R} \subseteq \assertions \times
  \processes \times \processes$,
  \begin{itemize}
  \item if $\mathcal{R}$ is a strong bisimulation then $\mathcal{R}$ is a weak $\tau$-bisimulation.
  \item if $\mathcal{R}$ is a weak $\tau$-bisimulation then $\mathcal{R}$ is a weak bisimulation.
  \end{itemize}
\end{lem}
\begin{cor}[Comparing congruences]\label{cor:comp-cong}
  If $P\sim_\Psi Q$ then $P\approx_\Psi Q$.
\end{cor}

We seek to establish the following standard congruence and structural properties properties of strong and weak bisimulation:
\begin{defi}[Congruence relation] \label{defi:congr}
A relation $\mathcal{R} \subseteq \assertions \times \processes \times \processes$, 
where $(\Psi,P,Q) \in \mathcal{R}$ is written $P\;\mathcal{R}_\Psi\;Q$,
is a \emph{congruence} iff for all $\Psi$,
$\mathcal{R}_\Psi$ is an equivalence relation,
and the following implications hold.
\[
\begin{array}{lrcll}
   \textsc{CPar}&P \;\mathcal{R}_\Psi \; Q &\;\Longrightarrow\;& (P \pll R) \;\mathcal{R}_\Psi \; (Q \pll R)
   \\ \textsc{CRes}&a\freshin\Psi \land P \;\mathcal{R}_\Psi \; Q &\;\Longrightarrow\;& \res{a}P \;\mathcal{R}_\Psi \; \res{a}Q
   \\ \textsc{CBang}&P \;\mathcal{R}_\Psi \; Q &\;\Longrightarrow\;& !P \;\mathcal{R}_\Psi\, !Q &
   \\ \textsc{CCase}&\forall i.P_i \;\mathcal{R}_\Psi \; Q_i  &\Longrightarrow&
   \casesprod{\ci{\ve{\varphi}}{\ve{P}}} \;{\;\mathcal{R}_\Psi}\;
   \casesprod{\ci{\ve{\varphi}}{\ve{Q}}}
   \\ \textsc{COut}&P \;\mathcal{R}_\Psi \; Q  &\Longrightarrow& \out{M}{N} \sdot P \;\mathcal{R}_\Psi\;
   \out{M}{N} \sdot Q
   \\\textsc{CIn}&P \;\mathcal{R}_\Psi \; Q &\Longrightarrow &\lin{M}{\ve{x}}{X} \sdot P  \;\mathcal{R}_\Psi \; \lin{M}{\ve{x}}{X} \sdot Q
\end{array}
\]
A \emph{\textsc{CCase}-pseudo-congruence} is defined like a
congruence, except that \textsc{CIn} is substituted by the following rule \textsc{CIn-2}.
\[
\begin{array}{lrcll}
  \textsc{CIn-2} & (\forall \ve{L}.\; P\lsubst{\ve{L}}{\ve{x}} \; \mathcal{R}_\Psi \; Q\lsubst{\ve{L}}{\ve{x}})
   &\Longrightarrow &
\lin{M}{\ve{x}}{X} \sdot P  \; \mathcal{R}_\Psi \; \lin{M}{\ve{x}}{X} \sdot Q
\end{array}
\]
A \emph{pseudo-congruence} is defined like a \textsc{CCase}-pseudo-congruence, but without rule \textsc{CCase}.
\end{defi}

\begin{defi}[Structural congruence]\label{def:structcong}
  \emph{Structural congruence}, denoted ${\equiv} \in \processes \times \processes$,
  is the smallest relation such that $\Set{(\one,P,Q) \; : \; P \equiv Q}$ is a congruence relation,
  and that satisfies the following clauses whenever $a\freshin Q,\ve{x},M,N,X,\ve{\varphi}$.
  \begin{mathpar}
    \begin{array}{rclcrcl}
      \casesprod{\ci{\ve{\varphi}}{(\nu a)\ve{P}}} & \equiv & 
      (\nu a)\casesprod{\ci{\ve{\varphi}}{\ve{P}}}
      &\qquad&
      !P  &\equiv &P \pll !P\\
      \lin{M}{\ve{x}}{X} \sdot \res{a}P &\equiv&
      \res{a}\lin{M}{\ve{x}}{X} \sdot P
      &\quad&
      P \pll (Q \pll R)  &\equiv& (P \pll Q) \pll R 
      \\
      \out{M}{N} \sdot \res{a}P  &\equiv& \res{a}\out{M}{N} \sdot P 
      &\quad&
      P \pll Q  &\equiv& Q \pll P\\
      Q \pll (\nu a)P & \equiv& (\nu a)(Q \pll P)
      &\quad&
      P &\equiv& P \pll \nil \\
      \res{b}\res{a}P  &\equiv& \res{a}\res{b}P
      &\quad&
      (\nu a)\nil &\equiv& \nil
    \end{array}
  \end{mathpar}

  A relation $\mathcal{R} \subseteq \processes \times \processes$ is
  \emph{complete with respect to structual congruence} 
  if ${\equiv} \subseteq \mathcal{R}$.
\end{defi}

\noindent Our goal is to establish that for all $\Psi$ the relations
$\bisim_\Psi$, $\sim_\Psi$, $\wbisim_\Psi$ and $\approx_\Psi$ are complete with respect to structural congruence; 
that $\bisim$ is a \textsc{CCase}-pseudo-congruence; 
that $\sim$ is a congruence; 
that $\wbisim$ is a pseudo-congruence; and
that $\approx$ is a congruence. 
\subsection{Psi-calculi with a single name sort}
\label{sec:trivally-sorted-calculi}

To prove the desired algebraic properties of strong and weak
bisimilarity and their induced congruences,
we first adapt the Isabelle proofs for the original psi-calculi 
to sorted psi-calculi with a single name sort, 
and then manually lift the results to arbitrary sorted psi-calculi. 
The reason for this approach is the lack of support in Nominal Isabelle for data types that are parametric in the sorts of names.
\begin{thm}
\label{thm:allresults}
If $\lvert \namesorts\rvert=\lvert \Restrictionsorts\rvert =1$, 
then $\bisim_\Psi$, $\sim_\Psi$, $\wbisim_\Psi$ and $\approx_\Psi$ are complete wrt. structural congruence for all $\Psi$, 
$\bisim$ is a \textsc{CCase}-pseudo-congruence, $\sim$ is a congruence, $\wbisim$ is a
pseudo-congruence, and $\approx$ is a congruence.
\end{thm}

These results have all been machine-checked in Isabelle~\cite{sortedProofsurl}.
The proof scripts are adapted from Bengtson's formalisation of psi calculi~\cite{Bengtson10.PhD}. The same technical lemmas hold and the proof scripts are essentially identical, save for the input cases of inductive proofs, a more detailed treatment of structural congruence, and the addition of sorts and compatibility relations.
We have also machine-checked 
Theorem~\ref{thm:origpats} (relationship to original psi-calculi)  and 
Theorem~\ref{thm:wfPres} (preservation of well-formedness) in this setting.
These developments comprise 31909 lines of Isabelle code; Bengtson's code is 28414 lines.
This represents no more than four days of work, with the bulk of the effort going towards proving a crucial technical lemma stating that transitions do not invent new names with the new matching construct.

Isabelle is an LCF-style theorem prover, where the only trusted component
is a small kernel that implements the inference rules of the logic and 
checks that they are correctly applied. All proofs must be fed through the
kernel. Hence the results are highly trustworthy.

As indicated these proof scripts apply only to calculi with a single name sort.
This restriction is a consequence of technicalities in Nominal Isabelle: 
it requires every name sort to be declared individually, and
there are no facilities to reason parametrically over the set of name sorts.

Huffman and Urban have developed a new foundation for Nominal Isabelle that 
lifts the requirement to declare every name sort individually~\cite{Huffman10:NFN:2176728.2176735}.
Unfortunately, the proof automation for reasoning about syntax 
quotiented by alpha-equivalence still assumes individually declared
name sorts. Working around this with manually constructed quotients is possible
in principle, but in practice this approach does not scale well enough to make
the endeavour feasible given the size of our formalisation.
A further difficulty is that Huffman and Urban's new foundation
is still alpha-ware and is not backwards-compatible.

\subsection{Trivially name-sorted psi-calculi}
\label{sec:triv-name-sort}
 A \emph{trivially name-sorted} psi-calculus is one where
 $\Restrictionsorts=\namesorts$ and there is $S\subseteq\sortset$ such that
${\CanSubstitute}=\namesorts\times S$, i.e., the sorts of names do not
affect how they can be used for restriction and substitution.

When generalising the result for single name-sorted calculi above, the main discrepancy is that
the mechanisation works with a single sort of names and thus would allow for ill-sorted alpha-renamings in the case of multiple name sorts. 
This is only a technicality, since every use of alpha-renaming in the formal proofs 
is to ensure that the bound names in patterns and substitutions avoid other bound names---thus, 
whenever we may work with an ill-sorted renaming, there would be a well-sorted renaming that suffices for the task.
\begin{thm}
  In trivially name-sorted calculi,
$\bisim_\Psi$, $\sim_\Psi$, $\wbisim_\Psi$ and $\approx_\Psi$ are complete wrt. structural congruence for all $\Psi$, 
$\bisim$ is a \textsc{CCase}-pseudo-congruence, $\sim$ is a congruence, $\wbisim$ is a
pseudo-congruence, and $\approx$ is a congruence.
\end{thm}
\proof By manually checking that all uses of alpha-equivalence in the proof of Theorem~\ref{thm:allresults} admit a well-sorted alpha-renaming.
\qed

\subsection{Arbitrary sorted psi-calculi}
\label{sec:extending-bisim}
\newcommand{\NOK}{\mathtt{fail}}
We here extend the results of Theorem~\ref{thm:allresults} to arbitrary sorted psi-calculi. 
The idea is to encode arbitrary sorted psi-calculi in trivially name-sorted psi-calculi by
introducing an explicit error element $\bot$, resulting from application of ill-sorted substitutions. 
For technical reasons we must also include one extra condition $\NOK$ 
  (cf.~Example~\ref{ex:extracondition}) 
and in the patterns we need different error elements with different support 
  (cf.~Example~\ref{ex:differenterrors}).

Let $I$ be a sorted psi-calculus with datatype parameters $\trms_I, \pats_I, \conditions_I, \assertions_I$.
We construct a trivially name-sorted psi-calculus $U(I)$ with one extra sort, $\kw{error}$,
and constant symbols~$\bot$ and~$\NOK$ with empty support of sort $\kw{error}$, where
$\bot$ is not a channel, never entailed, matches nothing and entails nothing but $\NOK$.

The parameters of $U(I)$ are defined by
$U(I) = (\trms_I \cup \Set{\bot}, \pats_I \cup \Set{(\bot,A)\;:\;A\subset_{\mathrm{fin}}\nameset}$, $\conditions_I \cup \Set{\bot, \NOK}, \assertions_I \cup \Set{\bot})$. 
We define $\Psi \ftimes \bot = \bot \ftimes \Psi = \bot$ for all $\Psi$, 
and otherwise $\ftimes$ is as in~$I$.
$\Match$ is the same in $U(I)$ as in $I$, 
plus $\Match(M,\ve{x},(\bot,S))=\Match(\bot,\ve{x},X)=\emptyset$. 
Channel equivalence~$\sch$ is the same in $U(I)$ as in $I$, 
plus $M \sch \bot = \bot \sch M = \bot$.
For $\Psi\in\assertions_{I}$ we let~$\Psi\vdash\varphi$ in $U(I)$ iff $\varphi\in\conditions_{I}$ and $\Psi\vdash\varphi$ in $I$, 
and we let $\bot\vdash\varphi$ iff $\varphi=\NOK$.
Substitution is then defined in $U(I)$ as follows: 
\[\begin{array}{rl}
  T\lsubst{\ve N}{\ve a}_{U(I)} := 
  &\;\; \left\{
  \begin{array}{l l}
    T\lsubst{\ve N}{\ve a}_I & \quad 
        \begin{array}[t]{@{}l}
             \text{if $\sort{a_i}\CanSubstitute_I\sort{N_i}$ and}\\
               \;\text{ $N_i \neq \bot$ for all $i$, and $T \neq (\bot,A)$}
         \end{array}\\
    (\bot, S\setminus \ve{a}) & \quad \text{if $T = (\bot,S)$ is a pattern} \\
    (\bot,\bigcup\vars(T)) & \quad \text{otherwise, if $T$ is a pattern}\\
    \bot & \quad \text{otherwise}\\
  \end{array} \right.
  \end{array}
\]
We define ${\bowtie}=(\sortset\times\{\kw{error}\})\cup(\{\kw{error}\}\times \sortset)$,
and the compatibility predicates of $U(I)$ as
$\CanReceive{=}\CanReceive_{I}\cup\bowtie$ and $\CanSend{=}\CanReceive_{I}\cup\bowtie$ and
${\CanSubstitute}=\namesorts\times\{s\in\sortset\;:\;\exists s'\in\namesorts.s'\CanSubstitute_{I}s\}$ and $\Restrictionsorts=\namesorts$.

\begin{lem}
  $U(I)$ as defined above is a trivially name-sorted psi-calculus, and
  any well-formed process $P$ in $I$ is well-formed in $U(I)$.
\end{lem}
\proof A straight-forward application of the definitions. \qed

The addition of $\NOK$ is in order to ensure the compositionality of $\ftimes$.
\begin{exa}
\label{ex:extracondition}
Let $\assertions=\{1,0\}$ and $\conditions=\{\varphi\}$ such that 
${\vdash}=\{(1,\varphi)\}$ and $1\otimes0=1$. 
Now add an assertion $\bot$ such that $1\otimes\bot = \bot$, and keep $\vdash$ unchanged.
Compositionality no longer holds, since $0\simeq\bot$, but $1\otimes0=1\not\simeq\bot=1\otimes\bot$.
\end{exa}

No variables can bind into equivariant patterns, 
so we need different error patterns with different support 
  to ensure the preservation of pattern variables under substitution.
\begin{exa}
\label{ex:differenterrors}
Assume that the pattern $X$ is equivariant. Then $\vars(X)\subseteq\{\emptyset\}$.
\end{exa}

Processes in $I$ have the same transitions in $U(I)$.
\begin{lem}
  \label{lemma:sametransitions}
  If $P$ is well-formed in $I$ and $\Psi\neq\bot$, 
  then $\Psi\frames \trans{P}{\alpha}{P'}$ in $U(I)$ iff
  $\Psi\frames \trans{P}{\alpha}{P'}$ in $I$.
\end{lem}
\proof
  By induction on the derivation of the transitions.
  The cases \textsc{In}, \textsc{Out}, \textsc{Case} and \textsc{Com} use the fact that 
  $\Match$, $\vdash$ and $\sch$ are the same in $I$ and $U(I)$, 
  and that substitutions in $I$ 
  have the same effect when considered as substitutions in $U(I)$.
\qed

Bisimulation in $U(I)$ coincides with bisimulation in $I$ for processes in $I$.
\begin{lem}\label{lem:correspondence}
  Assume that %
  $P$ and $Q$ are well-formed processes in $I$. Then
  $P\bisim_\Psi Q$ in $I$ iff $P\bisim_\Psi Q$ in $U(I)$, 
  and $P\wbisim_\Psi Q$ in $I$ iff $P\wbisim_\Psi Q$ in $U(I)$.
\end{lem}
\proof
We show only the proof for the strong case; the weak case is similar.
Let $\mathcal{R}$ be a bisimulation in $U(I)$. Then $\Set{(\Psi,P',Q')
  \in \mathcal{R}\;:\;\Psi\neq\bot \wedge P',Q' \text{~well-formed in~}I}$ is a bisimulation in $I$: 
the proof is by coinduction, using Lemma~\ref{lemma:sametransitions} and Theorem~\ref{thm:wfPres} in the simulation case.

Symmetrically, let $\mathcal{R}'$ be a bisimulation in $I$, and let
$\mathcal{R}_{\bot}'=\Set{(\bot,P,Q)\;:\;\exists \Psi. (\Psi,P,Q)\in \mathcal{R}'}$. 
Then $\mathcal{R}' \cup \mathcal{R}_{\bot}'$ is a bisimulation in $U(I)$:
simulation steps from $\mathcal{R}'$ lead back to $\mathcal{R}'$ by Lemma~\ref{lemma:sametransitions}. 
From $\mathcal{R}_{\bot}'$ there are no transitions, since $\bot$ entails no channel equivalence clauses. 
The other parts of Definition~\ref{def:bisim} are straightforward; 
when applying clause 3 with $\Psi'=\bot$ the resulting triple is in $\mathcal{R}_{\bot}'$.
\qed

With Lemma \ref{lem:correspondence}, we can lift the structural congruence results for trivially name-sorted psi-calculi to arbitrary sorted calculi:

\begin{thm}
  \label{thm:sortstructcong}
    For all sorted psi-calculi, 
    $\bisim_\Psi$, $\sim_\Psi$, $\wbisim_\Psi$ and $\approx_\Psi$ are complete wrt. structural congruence for all $\Psi$.
\end{thm}

\proof 
  Fix a sorted psi-calculus $I$. 
  For strong and weak bisimilarity, we show only the proof for commutativity of the parallel operator. 
  The other cases are analogous.  

  Let $P$ and $Q$ be well-formed in $I$ and $\Psi\neq\bot$. 
    By Theorem~\ref{thm:allresults}, $P \pll Q \sim_\Psi Q \pll P$ holds in $U(I)$.
    By Definition \ref{def:bisim}, $(P \pll Q)\ve\sigma \bisim_\Psi (Q \pll P)\ve\sigma$ in $U(I)$ for all $\ve\sigma$. 
    By Theorem~\ref{thm:wfPres}, when $\ve\sigma$ is well-sorted then $(P \pll Q)\ve\sigma$ and $(Q \pll P)\ve\sigma$ are well-formed. 
    By Lemma~\ref{lem:correspondence}, $(P \pll Q)\ve\sigma \bisim_\Psi (Q \pll P)\ve\sigma$ in $I$ for all well-sorted $\ve\sigma$. 
    $P \pll Q \sim_\Psi Q \pll P$ in $I$ follows by definition.
    $P \pll Q \approx_\Psi Q \pll P$ in $I$ follows by Corollary~\ref{cor:comp-cong}.
\qed
Using Lemma~\ref{lem:correspondence},
we can also lift the congruence properties of strong and weak
bisimilarity. 
\begin{thm}
\label{thm:congruence-bisim}  
In all sorted psi-calculi, $\bisim$ is a \textsc{CCase}-pseudo-congruence and
$\wbisim$ is a pseudo-congruence.
\end{thm}
\proof
  Fix a sorted psi-calculus $I$. 
  We show only the proof that $\bisim$ is a congruence with respect to parallel operator, the other cases are analogous.  

  Assume $P \bisim_\Psi Q$ holds in $I$. 
    By Lemma~\ref{lem:correspondence}, $P \bisim_\Psi Q$ holds in $U(I)$. 
    Theorem~\ref{thm:allresults} thus yields $P \pll R \bisim_\Psi Q \pll R$ in $U(I)$, 
    and Lemma~\ref{lem:correspondence} yields the same in $I$.  
\qed

Unfortunately, the approach of Theorems~\ref{thm:sortstructcong} and
\ref{thm:congruence-bisim} does not work for proving congruence properties for $\sim$ or $\approx$,
since the closure of bisimilarity under well-sorted substitutions does not imply its closure under ill-sorted substitutions:
consider a sorted psi-calculus $I$ such that $\nil \sim
\pass{\emptyframe}$. Here $\emptyframe\sigma = \bot$ if $\sigma$ is ill-sorted, but $\nil \bisim \pass{\bot}$ does not hold since only~$\bot$ entails $\NOK$.
We have instead performed a direct hand proof.
\begin{thm}
\label{thm:congruence-congruence}  
In all sorted psi-calculi, $\sim$ is a congruence and $\approx$ is a congruence.
\end{thm}
\proof
The proofs are identical, line by line, to the proofs for trivially name-sorted psi-calculi.
Theorem~\ref{thm:congruence-bisim} is used in every case.
\qed

\section{Representing Standard Process Calculi}\label{sec:process-calculi-examples}
We here consider psi-calculi corresponding to some variants of popular process calculi. 
One main point of our work is that we can represent other calculi directly as psi-calculi, 
 without elaborate coding schemes. 
In the original psi-calculi we could in this way directly represent the monadic pi-calculus, 
  but for the other calculi presented below 
    a corresponding unsorted psi-calculus would contain terms with no counterpart in the represented calculus, 
      as explained in Section~\ref{sec:sorting-intro}. 
We establish that 
our formulations enjoy a strong operational correspondence with the original calculus,
under trivial mappings that merely specialise the original concrete syntax (e.g., 
the pi-calculus prefix $a(x)$ maps to $\lin{a}{{x}}{x}$ in psi).

Because of the simplicity of the mapping and the strength of the correspondence
we say that psi-calculi \emph{represent} other process calculi, 
in contrast to \emph{encoding} them.
A representation is significantly stronger than standard correspondences, 
such as the approach to encodability proposed by Gorla~\cite{Gorla:encoding}. 
Gorla's criteria aim to capture the property that one language can encode the behaviour of another using some (possibly elaborate) protocol, 
while our criteria aim to capture the property that a language for all practical purposes is a sub-language of another. 
\begin{defi}
  A \emph{context} $C$ of arity $k$ is a psi-calculus process term with $k$
  occurrences of $\nil$ replaced by a hole $[]$.  
  We consider contexts as raw terms, i.e., no name occurrences are binding. 
  The instantiation $C[P_{1},\dots,P_{k}]$ of a context $C$ of arity
  $k$ is the psi-calculus process resulting from the replacement of
  the leftmost occurrence of $[]$ with $P_{1}$, the second leftmost
  occurrence of $[]$ with $P_{2}$, and so on.

  A psi-calculus is a \emph{representation} of a process calculus with processes
  $P\in\mathcal{P}$ and labelled transition system 
  ${\to}\subseteq\mathcal{P}\times\mathcal{A}\times\mathcal{P}$,
  if there exist
    an equivariant map $\semb\cdot$ from $\mathcal{P}$ to psi-calculus processes
    and an equivariant relation ${\approxeq}$ between $\mathcal{A}$ and psi-calculus actions 
  such that
  \begin{enumerate}
  \item $\semb\cdot$ is a simple homomorphism, i.e., for each process
    constructor $f$ of $\mathcal{P}$ there is an equivariant psi-calculus context
    $C$ such that $\semb{f(P_1,\dots,P_n)}=C[\semb{P_1},\dots,\semb{P_n}]$.
  \item $\semb\cdot$ is a strong operational correspondence 
    (modulo structural equivalence), i.e.,
    \begin{enumerate}
    \item whenever $\trans P\beta P'$ then there exist $\alpha,Q$ such that $\trans{\semb{P}}\alpha{Q}$
      and $\semb{P'} \equiv Q$ and $\beta\approxeq\alpha$; and
    \item whenever $\trans{\semb P}\alpha{Q}$ then there exist
      $\beta,P'$ such that $\trans P\beta P'$
      and $\semb{P'} \equiv Q$ and $\beta\approxeq\alpha$.
    \end{enumerate}
  \end{enumerate}
  A representation is \emph{complete} if it additionally satisfies
  \begin{enumerate}
  \item[(3)] $\semb\cdot$ is surjective modulo strong bisimulation congruence, i.e.,
    for each psi process $P$ there is $Q\in\mathcal{P}$ such that $P\sim\semb Q$.
  \end{enumerate}
\end{defi}

Any representation is a valid encoding in the sense of Gorla, but the
converse is not necessarily true.
\begin{itemize}
\item In Gorla's approach, the contexts that process constructors are translated to may fix certain names, 
or translate one name into several names, in accordance with a renaming policy.
We require equivariance, which admits no such special treatment of names.
\item Gorla uses three criteria for semantic correspondence:
weak operational correspondence modulo some equivalence for silent transitions,
that the translation does not introduce divergence,
and that reducibility to a success process in the source and target processes coincides.
Clearly strong operational correspondence modulo structural equivalence implies all of these criteria.
\end{itemize}

\noindent Our use of structural equivalence in the operational correspondence
allows to admit representations of calculi that use a structural
congruence rule to define a labelled semantics (cf.~Section~\ref{sec:polySynchPi}).

\medskip

Below, we use the standard notion of simultaneous substitution.
Since the calculi we represent do not use environments, 
we let the assertions be the singleton $\{\emptyframe\}$ in all examples, 
with $\emptyframe\vdash\top$ and $\emptyframe \not\vdash \bot$.
Proofs of lemmas and theorems can be found in Appendix~\ref{sec:examples-proofs}.

\subsection{Unsorted Polyadic pi-calculus}\label{sec:polyPi}

In the polyadic pi-calculus~\cite{milner:polyadic-tutorial} the only values that can be transmitted between agents are tuples of names. Tuples cannot be nested.
The processes are defined as follows.
\[
\boxed{
\begin{array}{rcl}
P, Q &\BnfDef&
\nil \BnfOr
x(\vec{y}).P \BnfOr
\overline{x}\langle \vec{y} \rangle.P \BnfOr
[a=b]P \BnfOr
\nu x\, P \BnfOr
!P \BnfOr
P \pll Q \BnfOr
P + Q 
\end{array}
}
\]
  An input binds a tuple of distinct names and can only communicate with an output of equal length, resulting in a simultaneous substitution of all names. In the unsorted polyadic pi-calculus there are no further requirements on agents, in particular $a(x).P \parop \overline{a}\vect{ y,z}.Q$ is a valid agent. This agent has no communication action since the lengths of the tuples mismatch.

We now present the psi-calculus \textbf{PPI}, 
which we will show represents the polyadic pi-calculus.
\instanceTwo{PPI}{
\trms=\nameset\cup\{\vect{\ve{a}}:\ve{a}\in\nameset^*\} \\
\conditions=\{\top \} \cup \{a=b\mid a,b\in \nameset\} \\
\pats=\{\vect{\ve{a}}:\ve{a}\in\nameset^* \wedge \ve{a} \text{ distinct}\} \\
{\sch} = \mbox{identity on names} \\
\emptyframe \vdash a=a \\
\vars(\vect{\ve{a}})=\Set{\ve{a}}\\
\Match(\vect{\ve{a}},\ve{x},\vect{\ve{y}})=\Set{\ve{c}} \text{~if~}\{\ve{x}\}=\{\ve{y}\}\text{ and }\vect{\ve{y}}\lsubst{\ve{c}}{\ve{x}}=\vect{\ve{a}}\\
\Match(M,\ve{x},\vect{\ve{y}})=\emptyset\text{~otherwise}
}{%
\sortset=\Set{\kw{chan}, \kw{tup}}\\
\namesorts=\Set{\kw{chan}} \\
\sort{a} = \kw{chan} \\
\sort{\vect{\ve{a}}} = \kw{tup} \\
\Restrictionsorts =\{\kw{chan}\} \\
{\CanSubstitute} = \{(\kw{chan}, \kw{chan})\}\\
{\CanSend}=    {\CanReceive} =\Set{(\kw{chan},\kw{tup})}\quad
}
This being our first substantial example, we give a detailed explanation of the new instance parameters.
Patterns $\pats$ are finite vectors of distinct names.
The sorts $\sortset$ are $\kw{chan}$ for channels and $\kw{tup}$ for
tuples (of names); 
the only sort of names $\namesorts$ is channels, as is the sort of restricted names. 
The only sort of substitutions ($\CanSubstitute$) are channels for channels; 
the only sort of sending ($\CanSend$) and receiving ($\CanReceive$) is tuples over channels.
In an input prefix all names in the tuple must be bound ($\vars$) and 
a vector of names $\ve{a}$ matches a pattern $\ve{y}$ if the lengths
match and all names in the pattern are bound (in some arbitrary
order).
 
As an example the agent $\lin{a}{x,y}{\vect{ x,y}}\sdot \out{a}{\vect{ y}}\sdot\nil$ is well-formed, since $\kw{chan}\CanReceive \kw{tup}$ and $\kw{chan}\CanSend \kw{tup}$, with $\vars(\vect{ x,y})=\Set{\Set{x,y}}$. 
This demonstrates that \textbf{PPI} disallows anomalies such as nested tuples but does not enforce a sorting discipline to guarantee that names communicate tuples of the same length.

To prove that \textbf{PPI} is a psi-calculus, 
we need to check the requisites on the parameters (data types and operations) defined above. 
Clearly the parameters are all equivariant, 
  since no names appear free in their definitions.
For the original psi-calculus parameters (Definition~\ref{def:parameters1}), 
  the requisites are symmetry and transitivity of channel equivalence, 
    which hold because of the same properties of (entailment of) name equality, 
  and abelian monoid laws and compositionality for assertion composition, 
    which trivially hold since $\assertions=\{\mathbf{1}\}$. 
The standard notion of simultaneous substitution of names for names preserves sorts, 
  and also satisfies the other requirements of Definition~\ref{def:subst}.
To check the requisites on pattern matching (Definition~\ref{def:pattern-match}), 
  it is easy to see that $\Match$ generates only well-sorted substitutions (of names for names), 
  and that $\n(\ve{b})=\n(\vect{\ve{a}})$ whenever $\ve{b}\in\Match(\vect{\ve{a}},\ve{x},\vect{\ve{y}})$
  Finally, for all name swappings $(\ve x\ \ve y)$ 
        we have $\Match(\vect{\ve{a}},\ve{x},\vect{\ve{z}})=
        \Match(\vect{\ve{a}},\ve{y},(\ve x\ \ve y)\cdot \vect{\ve{z}})$.

\textbf{PPI} is a representation of the polyadic pi-calculus as presented by Sangiorgi~\cite{sangiorgi:expressing-mobility}
(with replication instead of process constants).
\begin{defi}[Polyadic Pi-Calculus to \textbf{PPI}]$ $\\
Let $\llbracket \cdot \rrbracket$ be the function that maps the polyadic pi-calculus to
\textbf{PPI} processes as follows. 
The function $\semb\cdot$ is homomorphic for $\nil$, restriction, replication and parallel composition, and is otherwise defined as follows:
\[
\begin{array}{rcl}
\semb{P + Q} &=& \caseonly\;\ci{\top}{\semb{P}}\casesep\ci{\top}{\semb{Q}} \\
\semb{[x=y]P} &=& \caseonly\;\ci{x = y}{\semb{P}} \\
\semb{x(\vec{y}).P} &=& \lin{x}{\ve{y}}\vect{\ve{y}}.\semb{P} \\
\semb{\overline{x}\vect{ \vec{y} }.P} &=& \overline{x}\vect{ \vec{y} }.\semb{P}\\
\end{array}
\]
Similarly, we also translate the actions of polyadic pi-calculus. 
Here each action corresponds to a set of psi actions, since in a pi-calculus output label 
``the order of the bound names is immaterial''~\cite[p.~129]{sangiorgi.walker:theory-mobile},
which is not the case in psi-calculi.
\[\begin{array}{rcl}
\semb{(\nu \vec{y})\overline{x}\vect{\vec{z}}} &=&
                                                   \Set{\boutlabel{x}{\ve{y}'}{\vect{\ve{z}}}\;:\;\ve{y}'\text{~is
                                                   a permutation of~} \ve{y}} \\
\semb{x\vect{\vec{z}}} &=& \Set{\inn{x}{\vect{\vec{z}}}}\\
\semb{\tau} &=& \Set{\tau}
\end{array} \]
\end{defi}\smallskip

\noindent Although the binders in bound output actions are ordered in
psi-calculi, they can be arbitrarily reordered.
\begin{lem}\label{lem:bnReordering}
  If $\Psi\frames \trans P{\boutlabel{M}{\ve{a}}{N}}Q$ and $\ve{c}$ is
  a permutation of $\ve{a}$ then $\Psi\frames \trans P{\boutlabel{M}{\ve{c}}{N}}Q$.
\end{lem}
\proof
By induction on the derivation of the transition. The base case is trivial.
In the \textsc{Open} rule, we use the induction hypothesis to reorder
the bound names in the premise as desired; 
we can then add the opened name at the appropriate position in the action in the conclusion of the rule.
The other induction cases are trivial. \qed
We can now show that $\semb\cdot$ is a strong operational correspondence.
\begin{thm}
\label{thm:ppi-strong-operational-corresp.}
If $P$ and $Q$ are polyadic pi-calculus processes, then:
\begin{enumerate}
\item\label{item:comp} If $\trans{P}{\beta}{P'}$ then for all
  $\alpha\in\semb\beta$ we have $\trans{\semb{P}}{\alpha}{\semb{P'}}$; and
\item If $\trans{\semb{P}}{\alpha}{P''}$ then there is $\beta$ such that $\trans{P}{\beta}{P'}$
  and $\alpha\in\semb{\beta }$ and $\semb{P' } = P''$.
\end{enumerate}
\end{thm}
\proof By induction on the derivation of the transitions, using
Lemma~\ref{lem:bnReordering} in the \textbf{OPEN} case of~(\ref{item:comp}).
\qed
We have now shown that the polyadic pi-calculus can be embedded in \textbf{PPI},
with an embedding $\semb\cdot$ that is a strong operational correspondence.

In order to investigate surjectivity properties of the embedding $\semb\cdot$, 
we also define a translation $\overline{P}$ in the other direction.
\begin{defi}[\textbf{PPi} to Polyadic Pi-Calculus]
The translation~$\bmes\cdot$ is homomorphic for $\nil$, restriction, replication and parallel composition, and is otherwise defined as follows:
\[\begin{array}{rcl}
\bmes{\pass{\emptyframe}} &=& \nil \\
\bmes{\caseonly\;\ci{\varphi_1}{P_1}\casesep\dots\casesep\ci{\varphi_n}{P_n}} &=& \bmes{\ci{\varphi_1}{P_1}} + \dots + \bmes{\ci{\varphi_n}{P_n}} \\
\bmes{\lin{x}{\ve{y}}\vect{\ve{z}}.P} &=& x(\vec{z}).\bmes{P} \\
\bmes{\overline{x}\vect{ \vec{y} }.P} &=& \overline{x}\vect{ \vec{y} }.\bmes{P}\\
\end{array}\]
where condition-guarded processes are translated as 
\[\begin{array}{rcll}
\bmes{\ci{x = y}P} &=& [x=y]\bmes{P} \\
\bmes{\ci\top P} &=& \bmes{P}.
\end{array}\]
\end{defi}

Above, note that the order of the binders in input prefixes is ignored.
To show that the reverse translation is an inverse of $\semb{\cdot}$ modulo bisimilarity, 
we need to prove that their order does not matter.
\begin{lem}\label{lem:inputreorder}
 In \textnormal{\textbf{PPI}}, ${\lin{x}{\ve{y}}\vect{\ve{z}}.P}\sim{\lin{x}{\ve{z}}\vect{\ve{z}}.P}$.
\end{lem}
\proof Straightforward from the definitions of $\Match$ and substitution on patterns.\qed
We now show that the embeddings $\bmes\cdot$ and $\semb\cdot$
are inverses, modulo bisimilarity.
\begin{thm}\label{thm:polypiabs} $ $
    If $P$ is a \textnormal{\textbf{PPI}} process, then 
    $P\sim\llbracket \overline{P}\rrbracket$.
\end{thm}
\proof By structural induction on $P$. The input case uses Lemma~\ref{lem:inputreorder}.
For $\caseonly$ agents, we use an inner induction on the number of branches, 
with Lemma~\ref{lemma:flatten-case} applied in the induction case.\qed

Let the relation $\sim_e^c$ be early congruence of polyadic pi-calculus
agents as defined in \cite{sangiorgi:expressing-mobility}. Then we have

\begin{cor}\label{cor:ppi-sembinvertible}
  If $P$ is a polyadic pi-calculus process, then $P\sim_e^c\overline{\semb{P}}$.
\end{cor}
We also have
\begin{cor}
    If $P$ and $Q$ are polyadic pi-calculus process, 
    then $P \sim_e^c Q$ %
    iff $\llbracket P\rrbracket \sim \llbracket Q\rrbracket$.
\end{cor}
\proof Follows from the strong operational correspondence of Theorem~\ref{thm:ppi-strong-operational-corresp.}, and $\semb\cdot$ commuting with substitutions.\qed
This shows that every \textbf{PPI} process corresponds to a polyadic pi-calculus process, 
modulo strong bisimulation congruence,
since $\bmes\cdot$ is surjective on the 
bisimulation classes of polyadic pi-calculus, and the inverse of $\semb\cdot$.
In other words, \textbf{PPI} is a \emph{complete representation}.
\begin{thm}
  \textnormal{\textbf{PPI}} is a complete representation of the polyadic pi-calculus.
\end{thm}
\begin{proof} We let $\beta\approxeq\alpha$ iff $\alpha\in\semb\beta$.
  \begin{enumerate}
  \item $\semb\cdot$ is a simple homomorphism by definition.
  \item $\semb\cdot$ is a strong operational correspondence by Theorem~\ref{thm:ppi-strong-operational-corresp.}. 
  \item $\semb\cdot$ is surjective modulo strong bisimulation congruence by Theorem~\ref{thm:polypiabs}.\qedhere
  \end{enumerate}
\end{proof}

\subsection{LINDA~\cite{Gelernter.1985.LINDA}}\label{sec:linda}
A process calculus with LINDA-like pattern matching can easily be obtained from the \textbf{PPI} calculus,
by modifying the possible binding names in patterns.
 \instanceFrom{LINDA}{PPI}{
\pats=\{\vect{\ve{a}}:\ve{a}\subset_{\mathrm{fin}}\nameset \}\\
\vars(\vect{\ve{a}})=\Pow(\ve{a})\\
\Match(\vect{\ve{a}},\ve{x},\vect{\ve{y}})=\Set{\ve{c}} \text{~if~}\{\ve{x}\}\subseteq\{\ve{y}\}\text{ and }\vect{\ve{y}}\lsubst{\ve{c}}{\ve{x}}=\vect{\ve{a}}\\
}
Here, any subset of the names occurring in a pattern may be bound in the input prefix; 
this allows to only receive messages with particular values at certain positions (sometimes called ``structured names''~\cite{Gelernter.1985.LINDA})
We also do not require patterns to be linear, i.e., 
the same variable may occur more than once in a pattern, 
and the pattern only matches a tuple if each occurrence of the variable corresponds to the same name in the tuple.

As an example,
$
\lin{a}{x}\vect{ x,x,z}.P \parop \overline{a}\vect{ c,c,z}.Q\gtt P\lsubst{c}{x}\parop Q
$
while the agent
$
\lin{a}{x}\vect{ x,x,z}.P \parop \overline{a}\vect{ c,d,z}.Q
$ has no $\tau$ transition.

To prove that \textbf{LINDA} is a psi-calculus, 
the interesting case is the preservation of variables of substitution on patterns in Definition~\ref{def:subst}, i.e., 
that~$\ve{x}\in\vars(\vect{\ve{y}})$ and~$\ve{x}\freshin\sigma$ implies $\ve{x}\in\vars(\vect{\ve{y}}\sigma)$.  
This holds because standard substitution preserves names and structure: 
there is $\ve{z}$ such that $\vect{\ve{y}}\sigma=\vect{\ve{z}}$, and if $x\in\ve{y}$ and~$x\freshin\sigma$, 
then $x\in\ve{z}$.  

\subsection{Sorted polyadic pi-calculus}\label{sec:sortedPolyPi}
\newcommand{\basesorts}{\mbox{S}}
Milner's classic sorting~\cite{milner:polyadic-tutorial} regime for the polyadic pi-calculus 
ensures that pattern matching in inputs always succeeds, 
by enforcing that the length of the pattern is the same as the length of the received tuple.
This is achieved as follows. Milner assumes a countable set of subject sorts $\basesorts$ ascribed to names, and 
 a partial function $\kw{ob}:\basesorts\pto \basesorts^*$, assigning a sequence of object sorts to each sort in its domain. 
The intuition is that if $a$ has sort $s$ then any communication along $a$ must be a tuple of sort $\kw{ob}(s)$. 
An agent is \emph{well-sorted} if for any input prefix $a(b_1,\ldots b_n)$ it holds that $a$ has some sort $s$ where $\kw{ob}(s)$ is the sequence of sorts of $b_1,\ldots,b_n$ and similarly for output prefixes.
 \instanceFrom{SORTEDPPI}{PPI}{

\namesorts=\Restrictionsorts =\basesorts \qquad\qquad \qquad\;\;
\sortset=\basesorts\cup\Set{\vect{\ve{s}}:\ve{s}\in\basesorts^*}\\
{\CanSubstitute} = \{(s,s):s \in \basesorts\}\qquad \qquad
{\CanSend}=    {\CanReceive} =\{(s,\vect{\kw{ob}(s)}): s \in \basesorts\}
\\
\sort{\vect{a_1,\dots,a_n}}=\vect{\sort{a_1},\dots,\sort{a_n}}\\
\Match(\vect{\ve{a}},\ve{x},\vect{\ve{y}})=\Set{\pi\cdot\ve{a}} \quad \mbox{if}\; 
  \ve{x}=\pi\cdot\ve{y}\text{ and }\sort{\vect{\ve{a}}}=\sort{\vect{\ve{y}}}
}
We need to show that $\Match$ always generates well-sorted substitutions:
this holds since whenever
 $\ve{c}\in \Match(\vect{\ve{a}},\ve{x},\vect{\ve{y}})$ we have that $\lsubst{\ve{c}}{\ve{x}}=\lsubst{\pi\cdot\ve{a}}{\pi\cdot\ve{y}}$ 
and $\sort{y_i}=\sort{a_i}$ for all $i$.
 
  \newcommand{\nat}{\mathbb{N}} 
As an example, let $\sort{a}=s$ with $\kw{ob}(s)=t_1,t_2$ and $\sort{x}=t_1$ with $\kw{ob}(t_1)=t_2$ and $\sort{y}=t_2$ then the agent $\lin{a}{x,y}(x,y)\sdot \out{x}{y}\sdot\nil$ is well-formed, since $s\CanReceive t_1,t_2$ and $t_1\CanSend t_2$, with $\vars(x,y)=\Set{\Set{x,y}}$.

A formal comparison with the system in~\cite{milner:polyadic-tutorial} is complicated by the fact that Milner uses so called concretions and abstractions as agents. Restricting attention to agents in the normal sense we have the following result, where $\llbracket \cdot \rrbracket$ is the function from the previous example.
\begin{thm}
$P$ is well-sorted iff $\llbracket P \rrbracket$ is well-formed.
\end{thm}
  \proof
    A trivial induction over the structure of $P$, observing that the
    requirements are identical.
  \qed

\begin{thm}
  \textnormal{\textbf{SORTEDPPI}} is a complete representation of the sorted polyadic pi-calculus.
\end{thm}
\proof
  The operational correspondence in Theorem~\ref{thm:ppi-strong-operational-corresp.} 
  still holds when restricted to well-formed agents.
  The inverse translation $\bmes\cdot$ maps well-formed agents to well-sorted processes, 
  so the surjectivity result in Theorem \ref{thm:polypiabs} still applies.
\qed

\subsection{Polyadic synchronisation pi-calculus}
\label{sec:polySynchPi}

Carbone and
Maffeis~\cite{carbone.maffeis:expressive-power} explore the so called pi-calculus
with polyadic synchronisation, ${}^e\pi$, which can be thought of as a dual to the polyadic pi-calculus.
Here action subjects are tuples of names, while the objects transmitted are just single names.
It is
demonstrated that this allows a gradual
enabling of communication by opening the scope of names in a subject,
results in simple encodings of localities and cryptography, and  gives a
strictly greater expressiveness than standard \pic{}.
The processes of ${}^e\pi$ are defined as follows.
\[
\boxed{
\begin{array}{rcl}
P,Q &\BnfDef& \nil  \BnfOr \Sigma_i \alpha_i.P_i \BnfOr P\pll Q \BnfOr (\nu a)P \BnfOr !P \\
\alpha &\BnfDef& \vec{a}(x) \BnfOr \vec{a}\langle b \rangle
\end{array}
}\smallskip
\]
In order to represent ${}^e\pi$,
only minor modifications to the representation of
the polyadic pi-calculus in Section~\ref{sec:polyPi} are necessary.
To allow tuples in subject position but not in object position,
we invert the relations ${\CanSend}$ and $\CanReceive$.
Moreover, ${}^e\pi$ does not have name matching conditions $a=b$, 
since they can be encoded (see~\cite{carbone.maffeis:expressive-power}).
 \instanceTwoFrom[\qquad]{PSPI}{PPI}{
  \conditions=\{\top,\bot\} \\
  \pats=\nameset \\
  {\CanSend}=    {\CanReceive} =\{(\kw{tup},\kw{chan})\}
}{
  \mbox{$\ve{a} \sch \ve{b}$ is $\top$ if $\ve{a}=\ve{b}$, and $\bot$ otherwise}\\
  \vars(x)=\Set{\Set{x}} 
  \\
  \Match(a, x, x) = \{a\}
  }\smallskip

\noindent To obtain a representation, we consider a dialect of ${}^e\pi$ without the $\tau$ prefix.
This has no cost in terms of expressiveness
since the $\tau$ prefix can be encoded within ${}^e\pi$ using a communication over a restricted fresh name.
However, the \textbf{PSPI} context $C[]=(\nu\,a)(\out{\vect a}{a}.\nil\mid \lin{\vect{a}}{a}a.[]])$ that encodes the prefix 
is not admissible as part of a representation since it depends on the name $a$ and so is not equivariant. 

The  ${}^e\pi$ calculus also uses an operational semantics with late
input, unlike psi-calculi. In order to yield a representation, 
we consider an early version $\mathrel{{\apitransarrow{}}{}^{e}} $ of the semantics, 
obtained by turning bound input actions into free input actions at top-level.
\[
\inferrule*[left=\textsc{eIn}]
   {\trans{P}{\vec{x}(y)}{P'}}
   {\transPrim{P}{\vec{x}\;z}{P'\{z/y\}}}\qquad
\inferrule*[left=\textsc{Out}]
   {\trans{P}{\vec{x}\vect{ c}}{P'}}
   {\transPrim{P}{\vec{x}\vect{ c}}{P'}}\qquad
\inferrule*[left=\textsc{BOut}]
   {\trans{P}{\vec{x}\vect{ \nu c}}{P'}}
   {\transPrim{P}{\vec{x}\vect{ \nu c}}{P'}}\qquad
\inferrule*[left=\textsc{Tau}]
   {\trans{P}{\tau}{P'}}
   {\transPrim{P}{\tau}{P'}}
\]

\begin{defi}[Polyadic synchronisation pi-calculus to \textbf{PSPI}]
$\semb\cdot$ is homomorphic for~$\nil$, restriction, replication and parallel composition, and is otherwise defined as follows:
\[
\begin{array}{rcl}
\semb{\Sigma_i\alpha_i . P_i} &=& \caseonly\;\ci{\top_i}\semb{\alpha_i.P_i} \\
\semb{\ve{x}\vect{ y}.P} &=& \out{\vect{\ve{x}}}{y}.\semb{P}\\
\semb{\ve{x}(y).P} &=& \lin{\vect{\ve{x}}}{y}y.\semb{P} \\[0.4em]
\end{array}
\]
We translate bound and free output, free input, and tau actions in the following way.
\[\begin{array}{rcl}
\semb{\vec{x}\vect{ \nu c}} &=& \boutlabel{\vect{\vec{x}}}{c}{c} \\
\semb{\vec{x}\vect{ c}} &=& \out{\vect{\vec{x}}}{c} \\
\semb{\vec{x}\;y} &=&  \inlabel{\vect{\ve{x}}}y   \\
\semb{\tau} &=& \tau 
\end{array}\]
\end{defi}

The transition system in ${}^e\pi$ is given up to
structural congruence, i.e.,
for all $\alpha$ we have ${\xrightarrow\alpha}=(\equiv\xrightarrow\alpha\equiv)$.
\begin{defi}
 $\equiv$ is the least congruence satisfying alpha conversion, 
 the commutative monoidal laws with respect to both ($|$,0) and (+,0)
 and the following axioms\footnote{The original definition of $\equiv$~\cite{carbone.maffeis:expressive-power} 
includes an additional axiom $[x=x]P\equiv P$ allowing to contract
successful matches, but this axiom is omitted here
since the ${}^e\pi$ calculus does not include the match construct.
Unusually, the definition of $\equiv$ does not admit commuting
restrictions, i.e., $(\nu x)(\nu y)P\not\equiv(\nu y)(\nu x)P$.}:
 \[
   (\nu x)P \mid Q \equiv (\nu x)(P \mid Q)\text{ if }x\freshin Q
   \qquad \qquad \qquad
   (\nu x)P \equiv P\text{ if }x\freshin P
 \]
\end{defi}

The proofs of operational correspondence are similar to the polyadic pi-calculus case. 
We have the following initial results for late input actions.
\begin{lem}\label{lem:polySynchTransition}$ $
\begin{enumerate}
\item If $\trans{P}{\vec{x}(y)}{P'}$ then for all z, $\trans{\semb{P}}{\rule[-0.5em]{0pt}{1em}\inn{\vect{\vec{x}}}{z}}{P''}$ where $P'' \equiv \semb{P'}\lsubst{z}{y}$.
\item If $\trans{\semb{P}}{\rule[-0.5em]{0pt}{1em}\inn{\vect{\vec{x}}}{z}}{P''}$ then for all $y \freshin P$, $\trans{P}{\vec{x}(y)}{P'}$ where $\semb{P'\{z/y\} } = P''$.
\end{enumerate}
\end{lem}
\proof
By induction on the derivation of the transitions.%
\qed

This in turn yields the desired operational correpondence.
\begin{thm}\label{thm:polySynchTransition}$ $
\begin{enumerate}
\item If $\transPrim{P}{\alpha}{P'}$, then $\trans{\semb{P}}{\semb{\alpha}}{P''}$ where $P'' \equiv \semb{P'}$. 
\item If $\trans{\semb{P}}{\alpha'}{P''}$, then $\transPrim{P}{\alpha}{P'}$ where $\semb{\alpha } = \alpha'$ and $\semb{P' } = P''$.
\end{enumerate}
\end{thm}
\proof
By induction on the derivation of the transitions.%
\qed
Again, these results lead us to say that the polyadic synchronization pi-calculus can be \emph{represented} as a psi-calculus.
\begin{thm}
  \textnormal{\textbf{PSPI}} is a representation of the polyadic synchronization pi-calculus.
\end{thm}
\begin{proof} We let $\beta\approxeq\alpha$ iff $\alpha=\semb\beta$.
  \begin{enumerate}
  \item $\semb\cdot$ is a simple homomorphism by definition.
  \item $\semb\cdot$ is a strong operational correspondence by Theorem~\ref{thm:ppi-strong-operational-corresp.}. \qedhere
  \end{enumerate}
\end{proof}

To investigate the surjectivity properties of $\semb\cdot$, 
we need to consider the fact that polyadic synchronization pi has only mixed (i.e., prefix-guarded) choice. 
\begin{defi}[Case-guarded]
  A \textbf{PSPI} process is case-guarded if in all its subterms of the
  form
  $\caseonly\;{\ci{\varphi_1}{P_1}\casesep\cdots\casesep\ci{\varphi_n}{P_n}}$,
  for all $i\in\Set{1,\dots,n}$, $\varphi_i=\top$ implies $P_i=\out{M}N .Q$ or
  $P_i=\lin{M}{\ve{x}}X.Q$.
\end{defi}
We define the translation $\overline{R}$ from case-guarded \textbf{PSPI} processes to ${}^e\pi$ as the translation with the same name from \textbf{PPI}, except that $\bot$-guarded branches of $\caseonly$ statements are discarded.
\begin{thm}
\label{thm:full-abstraction-of-PSPi}
For all case-guarded \textnormal{\textbf{PSPI}} processes $R$ we have $R\sim \semb{\overline{R}}$. 
\end{thm}
\proof By structural induction on $R$.
For $\caseonly$ agents, we use an inner induction on the number of branches, 
with Lemma~\ref{lemma:flatten-case} applied in the induction case.\qed
\begin{cor}
  If $P$ is a polyadic synchronization pi-calculus process, then $P\mathrel{\dot\sim}\overline{\semb{P}}$.
\end{cor}
\begin{cor}
  For all ${}^e\pi$ processes $P$, $Q$, $P \mathrel{\dot\sim} Q$ (i.e., $P$ and $Q$ are early labelled congruent) iff $\semb{P} \sim \semb{Q}$.
\end{cor}
\proof By strong operational correspondence~\ref{thm:polySynchTransition}, and $\semb\cdot$ commuting with substitutions.\qed

We thus have that polyadic synchronization pi corresponds
to the case-guarded \textbf{PSPI} processes, modulo strong bisimulation. 

\subsection{Value-passing CCS}\label{sec:vpCCS}

  Value-passing CCS~\cite{milner:cac} is an extension of pure CCS to
    admit arbitrary data from some set \textbf{V} to be sent along
    channels; there is no dynamic connectivity so channel names cannot
    be transmitted.  When a value is received in a communication it
    replaces the input variable everywhere, and where this results in
    a closed expression it is evaluated, so for example $a(x) \sdot
    \overline{c}(x+3)$ can receive 2 along $a$ and become
    $\out{c}5$. There are conditional \textbf{if} constructs that can
    test if a boolean expression evaluates to true, as in $a(x) \sdot
    \ifthen{x>3}{P}$.
Formally, the value-passing CCS processes are defined by the following grammar
with $x,y$ ranging over names, $v$ over values, $b$ over boolean expressions,
and $L$ over sets of names.
\[
\boxed{
\begin{array}{rcl}
P,Q \;\BnfDef\;
x(y).P \BnfOr 
\overline{x} (v) .P \BnfOr
\Sigma_i\;P_i \BnfOr 
\mathbf{if}\;b\;\mathbf{then}\;P \BnfOr 
P\;\backslash\;L \BnfOr 
P \pll Q  \BnfOr
!P  \BnfOr
\nil
\end{array}
}
\]\smallskip

\noindent     To represent this as a psi-calculus we assume an arbitrary set of
    expressions $e\in\textbf{E}$ including at least the values $\mathbf{V}$. 
    A subset of \textbf{E} is the boolean expressions $b\in\mathbf{E}_\mathbf{B}$.  
    Names are either used as channels (and then have the sort
    \kw{chan}) or expression variables (of sort \kw{exp}); only the
    latter can appear in expressions and be substituted by values. An
    expression is closed if it has no name of sort \kw{exp} in its
    support, otherwise it is open.  The values $v\in\textbf{V}$ 
    are closed and have sort \kw{value}; all other expressions have sort \kw{exp}.
    The boolean values are 
    $\textbf{V}_{\mathbf{B}}:=\mathbf{V}\cap\mathbf{E}_{\mathbf{B}}=\{\top,\bot\}$,
and $\emptyframe\vdash\top$ but $\neg(\emptyframe\vdash\bot)$.
    We let $E$ be an evaluation function on expressions, that takes
    each closed expression to a value and leaves open expressions
    unchanged. We write $e\{\ve{V}/\ve{x}\}$ for the result of
    syntactically replacing all $\ve{x}$ simultaneously by $\ve{V}$ in
    the (boolean) expression $e$, and assume that the result is a
    valid (boolean) expression. For example $(x+3)\{2/x\}$ = 2+3, and
    $E(2+3) = 5$.  We define substitution on expressions to use evaluation, i.e.
    $e\lsubst{\ve{V}}{\ve{x}}= E(e\{\ve{V}/\ve{x}\})$. As an example, 
    $(x+3)\lsubst{2}{x}=E((x+3)\{2/x\})=E(2+3)=5$.  We use the single-variable 
    patterns of Example~\ref{ex:symsempats}.

    \instanceTwo{VPCCS}{
      \trms = \nameset\cup\mathbf{E}\\
      \conditions = \mathbf{E}_{\mathbf{B}} \\
      \assertions = \{\emptyframe\} \\
      \pats = \nameset\\
      {a\sch a} = \top \\
      {e\sch e'} = \bot\text{ otherwise} \\
      \vars(a) = \{a\} \\
      \Match(v,a,a)=\Set{v}\;\mbox{if $v \in \mathbf{V}$}\\
      \Match(M,\ve{x},a)=\emptyset\text{~otherwise}}{
      \namesorts=\{\kw{chan},\kw{exp}\} \\
      \sortset=\namesorts \cup \{\kw{value}\} \\
      v \in \mathbf{V}  \Rightarrow \sort{v} = \kw{value} \\
      e \in \mathbf{E}\setminus\textbf{V}  \Rightarrow \sort{e} = \kw{exp} \\
      e \in \mathbf{E}  \Rightarrow e\lsubst{\ve{M}}{\ve{x}} = E(e\{\ve{M}/\ve{x}\}) \\
      {\CanSubstitute} = \{(\kw{exp}, \kw{value})\}\\
      \Restrictionsorts =\{\kw{chan}\} \\{\CanSend}= {\CanReceive}
      =\{(\kw{chan},\kw{exp}), (\kw{chan},\kw{value})\} }

    Closed value-passing CCS processes correspond to \textbf{VPCCS}
    agents $P$ where all free names are of sort $\kw{chan}$.
    To prove that \textbf{VPCCS} is a psi-calculus, 
    the interesting case is when the sort of a term is changed by substitution:
    let~$e$ be an open term, and $\sigma$ a substitution such that 
    $\n(e)\subseteq\dom\sigma$.  
    Here $\sort e=\kw{exp}$ and $\sort{e\sigma}=\kw{value}$;
    this satisfies Definition~\ref{def:subst} since $\kw{value}\le\kw{exp}$ 
    in the subsorting preorder (here $\kw{exp}\le\kw{value}$ also holds,
    but is immaterial since there are no names of sort \kw{value}).

    We show that \textbf{VPCCS} represents value-passing CCS as
    defined by Milner~\cite{milner:cac}, with the following
    modifications: 
    \begin{itemize}
    \item We use replication instead of process constants.
    \item We
      consider only finite sums. Milner allows for infinite sums
      without specifying exactly what infinite sets are allowed and
      how they are represented, making a fully formal comparison
      difficult. Introducing infinite sums naively in psi-calculi
      means that agents might exhibit cofinite support and exhaust the
      set of names, rendering crucial operations such as
      $\alpha$-converting all bound names to fresh names impossible.
    \item We do not consider the relabelling construct $P[f]$ of CCS
	at all. Injective relabelings are redundant in CCS~\cite{DBLP:conf/fossacs/GiambiagiSV04}, and 
      the construct is not included in the psi-calculi
      framework.
    \item We only allow finite sets $L$ in restrictions
      $P\;\backslash\;L$.  With finite sums, this results in no loss
      of expressivity since agents have finite support.
    \end{itemize}
    Milner's restrictions are of sets of names, which we represent as a
    sequence of $\nu$-binders.  To create a unique such sequence from $L$, 
    we assume an injective and support-preserving function 
    $\overrightarrow\cdot :
    {\PowFin}(\nameset[\kw{chan}]) \rightarrow (\nameset[\kw{chan}])^*$.
    For instance, $\overrightarrow L$ may be defined as sorting the names in $L$ 
    according to some total order on $\nameset[\kw{chan}]$, 
    which is always available since $\nameset[\kw{chan}]$ is countable.

    The mapping $\semb{\cdot}$ from value-passing CCS into \textbf{VPCCS} is 
    defined homomorphically on parallel composition, output and~$\nil$, and otherwise as follows.

  \[\begin{array}{rcl}
          \semb{x(y).P} &=& \lin{x}{y}{y}.\semb{P}  \\
  \semb{\Sigma_i\;P_i } &=& \caseonly\;{\ci{\top}{\semb{P_1}}\casesep\cdots\casesep\ci{\top}{{\semb{P_i}}}}\\
  \semb{\mathbf{if}\;b\;\mathbf{then}\;P} &=& \caseonly\;{\ci{b}{\semb{P}}}\\
  \semb{P\;\backslash\;L} &=& (\nu \overrightarrow L)\semb{P}
  \end{array}
  \]

  We translate the value-passing CCS actions as follows
  \[\begin{array}{rcl}
          \semb{x(v)} &=& \inlabel{x}{v} \\
          \semb{\overline{x}(v)} &=& \out{x}{v} \\
          \semb{\tau} &=& \tau
  \end{array}\]

As an example, in a version of \textbf{VPCCS} where the expressions $\mathbf{E}$ include natural numbers and operations on those,
\[
\begin{array}{rcl}
\multicolumn{3}{l}{
\lin{a}{x}{x} \sdot \caseonly\;{\ci{x>3}{\overline{c}(x+3)}}
}\\
\qquad &\gt{\inlabel{a}{4}}& (\caseonly\;{\ci{x>3}{\overline{c}(x+3)}})\lsubst{4}{x}\\
&= & \caseonly\;{\ci{E((x>3)\subst{4}{x})}{\overline{c}(E((x+3)\subst{4}{x}))}}\\
&= & \caseonly\;{\ci{E({4>3})}{\overline{c}(E({4+3}))}}\\
&= & \caseonly\;{\ci{\top}{\overline{c}7}}\\
&\gt{\out{c}{7}} & \nil
\end{array}
\]

In our psi semantics, expressions in processes are evaluated when they
are closed by reception of variables (e.g.~in the first transition
above), while Milner simply identifies closed expressions with their
values~\cite[p55f]{milner:cac}.

\begin{lem}\label{lemma:CCSclosed}
  If $P$ is a closed \textnormal{\textbf{VPCCS}} process and $\trans{P}{\alpha}{P'}$,
  then $P'$ is closed.
\end{lem}
\begin{thm}\label{thm:CCStrans}
  If $P$ and $Q$ are closed value-passing CCS processes, then
  \begin{enumerate}
  \item if $\trans{P}{\alpha}{P'}$ then
    $\trans{\semb{P}}{\semb{\alpha}}{\semb{P'}}$; and
  \item if $\trans{\semb{P}}{\alpha'}{P''}$ then
    $\trans{P}{\alpha}{P'}$ where $\semb{\alpha} = \alpha'$ and
    $\semb{P'} = P''$.
  \end{enumerate}
\end{thm}
\proof
  By induction on the derivations of $P'$ and $P''$, respectively.
 The full proof is given in Appendix~\ref{sec:value-passing-ccs-proofs}.
\qed

As before, this yields a representation theorem.
\begin{thm}
  \textnormal{\textbf{VPCCS}} is a representation of the closed agents of value-passing CCS
  (modulo the modifications described above).
\end{thm}
\proof We let $\beta\approxeq\alpha$ iff $\alpha=\semb\beta$.
  \begin{enumerate}
  \item $\semb\cdot$ is a simple homomorphism by definition.
  \item $\semb\cdot$ is a strong operational correspondence by Theorem~\ref{thm:CCStrans}. 
    \qed
  \end{enumerate}

To investigate the surjectivity of the encoding, we let
$\mathcal{P}=\{P\;:\;\sort{\n(P)}\subseteq\Set{\kw{chan}}\}$ 
be the \textbf{VPCCS} processes where all fre names are of channel sort.
\begin{lem} \label{lem:CCSsurj}
  If $P\in\mathcal{P}$, then there is a CCS process $Q$ such that
  $P\sim\semb{Q}$.
\end{lem}
\proof
As before, we define an inverse translation $\bmes\cdot$, that is
homomorphic except for 
\[\bmes{\caseonly\;{\ci{b_1}{{P_1}}\casesep\cdots\casesep\ci{b_i}{{{P_i}}}}}
= (\mathbf{if}\;b_1\;\mathbf{then}\;\bmes{P_1}) +\dots + (\mathbf{if}\;b_i\;\mathbf{then}\;\bmes{P_i})\]
Using Lemma~\ref{lemma:flatten-case}, we get $P\sim\semb{\bmes{P}}$.
\qed

\begin{exa}[Value-passing pi-calculus]\label{ex:vpPi}
  To demonstrate the modularity of psi-calculi, assume that we wish a
  variant of the pi-calculus enriched with values in the same way as
  value-passing CCS. This is achieved with only a minor change to {\bf
    VPCCS}: \instance{VPPI}{
    \text{Everything as in \textbf{VPCCS} except:} \\
    \Match(z,a,a)=\Set{z}\;\mbox{if $z \in \mathbf{V} \cup \nameset[ch]$}\\
    \CanSubstitute\; = \{(\kw{exp}, \kw{value}), (\kw{chan},\kw{chan})\}\\
    {\CanSend}= {\CanReceive} =\{(\kw{chan},\kw{exp}),
    (\kw{chan},\kw{value}), (\kw{chan},\kw{chan})\} } 
  Here also channel names can be substituted for other channel names, 
  and they can be sent and received along channel names.
\end{exa}

\section{Advanced Data Structures}\label{sec:advanc-data-struct}
 We here demonstrate that we can accommodate a variety of term structures for data and communication channels; 
in general these can be any kind of data, and substitution can include any kind of computation on these structures. 
This indicates that the word ``substitution'' may be a misnomer --- a better word may be ``effect'' --- though we keep it to conform with our earlier work.
We focus on our new contribution in the patterns and sorts, and
therefore make the following definitions that are common to all the
examples (unless explicitly otherwise defined).
\[\boxed{
     \begin{array}{lll}
       \assertions = \{\emptyframe\} & \emptyframe\otimes\emptyframe=\emptyframe & \conditions = \{\top,\bot\} \\
       {\vdash} = \{(\emptyframe,\top)\} & M\sch M =\top & M\sch N = \bot\text{~if~}M\neq N\\
           {\CanSubstitute} = \{(s,s)\;:\; s\in\sortset\} &
           {\CanReceive}={\CanSend}=\sortset\times\sortset &
           \Restrictionsorts=\namesorts=\sortset
      \end{array}
     }
   \]
If $t$ and $u$ are from some term algebra, we write $t\preceq u$ when $t$ is a (non-strict) subterm of $u$.

\subsection{Convergent rewrite systems on terms}\label{sec:nform}
In Example~\ref{ex:vpPi}, the value language consisted of closed
terms, with an opaque notion of evaluation.  We can instead work with
terms containing names and consider deterministic computations
specified by a convergent rewrite system.
The interesting difference is in which terms are admissible as patterns,
and which choices of 
$\vars(X)$ are valid.
We first give a general definition and then give a concrete instance in Example~\ref{ex:Peano}.

Let $\Sigma$ be a sorted signature with sorts $\sortset$,
and $\cdot\Downarrow$ be normalization with respect to a convergent sort-preserving
rewrite system on the nominal term algebra over $\nameset$ generated
by the signature~$\Sigma$.  We let terms $M$ range over the range of $\Downarrow$, i.e., the normal forms.
  We write $\rho$ for sort-preserving capture-avoiding
  simultaneous substitutions~$\subst{\ve{M}}{\ve{a}}$ where every
  $M_i$ is in normal form; here $n(\rho)=\n(\ve{M},\ve{a})$.
  A term $M$ is stable if for all $\rho$, $M\rho{\Downarrow}=M\rho$.  
  The patterns are all instances of stable terms, i.e., 
  $X=M\rho$ where $M$ is stable. 
  Such a pattern $X$ can bind any combination of names occurring in $M$ but not in $\rho$.
  As an example, any term $M$ is a pattern (since any name $x$ is
  stable and $M=x\subst Mx$) that can be used to match the term $M$
  itself (since $\emptyset\subseteq\n(x)\setminus\n(M,x)=\emptyset$).

\instanceTwo{REWRITE($\Downarrow$)}{
  \trms=\pats=\operatorname{range}(\Downarrow)\\
  M\lsubst{\ve{L}}{\ve{y}}=M\subst{\ve{L}}{\ve{y}}{\Downarrow}}{
  \Match(M,\ve{x},X)=\Set{\ve{L}: M=X\subst{\ve{L}}{\ve{x}}}\\
  \vars(X)=\bigcup\Set{\Pow(\n(M)\setminus\n(\rho)): M\ \text{stable}\land X=M\rho}
}
We need to show that the patterns are closed under substitution,
  including preservation of \textsc{vars} (cf.~Definition~\ref{def:subst}), 
and that matching satisfies the criteria of Definition~\ref{def:pattern-match}.
Since any term is a pattern, the patterns are closed under substitution. 
Since term substitution $\subst{\cdot}{\cdot}$ and normalization $\Downarrow$ are both sort-preserving, 
  term and pattern substitution $\lsubst\cdot\cdot$ is also sort-preserving.

To show preservation of pattern variables, 
  assume that $\ve{x}\in\vars(X)$ is a tuple of distinct names.  
By definition there are $M$ and $\rho$ such that 
  $X=M\rho$ with $M$ stable and $\ve{x}\subseteq\n(M)\setminus \n(\rho)$.  
Assume that $\ve{x}\freshin\sigma$; 
  then $X\sigma=(M\rho)\sigma=M(\sigma\circ\rho)$ 
  with $\ve{x}\freshin\sigma\circ\rho$, so $\ve{x}\in\vars(X\sigma)$.

For the criteria of Definition~\ref{def:pattern-match}, 
  additionally assume that $\ve{L}\in\Match(N,\ve{x},X)$ and let $\sigma= \lsubst{\ve{L}}{\ve{x}}$. 
Since $\subst{\ve{L}}{\ve{x}}$ is well-sorted, so is $\lsubst{\ve{L}}{\ve{x}}$. 
We also immediately have $\n(\ve{L})=\n(N) \cup(\n(X)\setminus\ve{x})$, 
  and alpha-renaming of matching follows from the same property for term substitution.

\begin{exa}[Peano arithmetic]\label{ex:Peano}
As a simple instance of \textbf{REWRITE}($\Downarrow$), 
we may consider Peano arithmetic.
  The rewrite rules for addition (below) induce
  a convergent rewrite system $\Downarrow^{\mathrm{Peano}}$, 
  where the stable terms are those that do not contain any occurrence of $\kw{plus}$.  

\instanceFrom{PEANO}{REWRITE($\Downarrow$)} {
  \sortset=\Set{\kw{nat},\kw{chan}}\\
  \Sigma=\{\kw{zero}:\kw{nat}, \qquad 
  \kw{succ}:\kw{nat}\to\kw{nat}\qquad
  \kw{plus}:\kw{nat}\times\kw{nat}\to\kw{nat}\}\\
  \kw{plus}(K,\kw{zero})\to K\qquad
  \kw{plus}(K,\kw{succ}(M))\to\kw{plus}(\kw{succ}(K),M) \\
  \vars(\kw{succ}^n(a))=\Set{\emptyset,\Set{a}}\qquad
  \vars(M)=\Set{\emptyset}\text{ otherwise}
}
  Writing $i$ for $\kw{succ}^i(\kw{zero})$, the agent 
  $(\nu a)(\out{a}{2}\parop\lin{a}{y}{\kw{succ}(y)}\sdot\out{c}{\kw{plus}(3,y)})$
  of\linebreak \textbf{REWRITE}($\Downarrow^{\mathrm{Peano}}$) has one visible transition, with the label $\out{c}{4}$.
  In particular, the object of the label is $\kw{plus}(3,y)[y:=1]=\kw{plus}(3,y)\subst{1}y{\Downarrow^{\mathrm{Peano}}}=4$.
\end{exa}

\subsection{Symmetric cryptography}\label{sec:DYalg}
We can also consider variants of \textbf{REWRITE}($\Downarrow$), 
such as a simple Dolev-Yao style~\cite{DY83} cryptographic message algebra for symmetric cryptography, 
where we ensure that the encryption keys of received encryptions can not be bound in input patterns,
in agreement with cryptographic intuition.

  The rewrite rule describing decryption $\kw{dec}(\kw{enc}(M,K),K)\to M$ induces a
  convergent rewrite system $\Downarrow^{\mathrm{enc}}$, 
  where the terms not containing $\kw{dec}$ are stable.
  The construction of \textbf{REWRITE}($\Downarrow$)
  yields that $\ve{x}\in\vars(X)$ if $\ve{x}\subseteq\n(X)$ are pair-wise different
  and no $x_i$ occurs as a subterm of a $\kw{dec}$ in $X$. 
  This construction would still permit to bind the keys of an encrypted message upon reception, 
e.g.~$\lin{a}{m,k}{\kw{enc}(m,k)}\sdot P$ would be allowed although it does not make cryptographic sense.
Therefore we further restrict $\vars(X)$ to those sets not containing names that occur in 
  key position in $X$, thus disallowing the binding of $k$ above. 
Below we give the formal definition (recall that $\preceq$ is the subterm preorder).
\instanceFrom{SYMSPI}{\textbf{REWRITE}($\Downarrow^{\mathrm{enc}}$)}{
  \sortset=\{\kw{message}, \kw{key}\}\\
  \Sigma=\Set{\kw{enc}:\kw{message}\times\kw{key}\to\kw{message},\quad
    \kw{dec}:\kw{message}\times\kw{key}\to\kw{message}}\\
  \kw{dec}(\kw{enc}(M,K),K)\to M\\
  \vars(X)=\mathcal{P}(n(X)\setminus\{a:a\preceq\kw{dec}(Y_1,Y_2)\preceq X \lor
(a\preceq Y_2\land\kw{enc}(Y_1,Y_2)\preceq X)\})
}
The proof of the conditions of Definition~\ref{def:subst} and
Definition~\ref{def:pattern-match} for patterns is the same as for
\textbf{REWRITE}($\cdot$) in Section~\ref{sec:nform} above.

 As an example, the agent 
\[
(\nu a,k)(\out{a}{\kw{enc}(\kw{enc}(M,l),k)}\parop\lin{a}{y}{\kw{enc}(y,k)}\sdot\out{c}{\kw{dec}(y,l)})
\]
has a visible transition with label $\out{c}{M}$,
where one of the leaf nodes of the derivation is
\[\lin{a}{y}{\kw{enc}(y,k)}\sdot\out{c}{\kw{dec}(y,l)}\gt{\inlabel{a}{\kw{enc}(\kw{enc}(M,l),k)}} \out{c}{\kw{dec}(y,l)}\lsubst{\kw{enc}(M,l)}{y}\]
 since $\kw{enc}(M,l)\in\Match(\kw{enc}(\kw{enc}(M,l),k),y,\kw{enc}(y,k))$. 
The resulting process is
\[\out{c}{\kw{dec}(y,l)}\lsubst{\kw{enc}(M,l)}{y}
={\out{c}{\kw{dec}(y,l)\subst{\kw{enc}(M,l)}{y}\Downarrow}}
{}={\out{c}{\kw{dec}(\kw{enc}(M,l),l)\Downarrow}}=\out{c}{M}.\]

\subsection{Asymmetric cryptography}\label{sec:pmSpi}
A more advanced version of Section~\ref{sec:DYalg} is the
treatment of data in the
pattern-matching spi-calculus~\cite{haack.jeffrey:pattern-matching-spi},
to which we refer for more examples and motivations of the definitions below.
The calculus uses asymmetric encryption, and 
 includes a non-homomorphic definition of substitution that does not preserve sorts, and a sophisticated way of computing permitted pattern variables.
This example highlights the flexibility of sorted psi-calculi
in that such specialized modelling features can be presented
in a form that is very close to the original.

  We start from the term algebra $T_\Sigma$ over the unsorted signature 
  \[\Sigma=\Set{(),\,(\cdot,\cdot),\,\kw{eKey}(\cdot),\,\kw{dKey}(\cdot),\,\kw{enc}(\cdot,\cdot),\,\kw{enc}^{-1}(\cdot,\cdot)}\]
The $\kw{eKey}(M)$ and $\kw{dKey}(M)$ constructions represent the encryption and decryption parts of the key pair $M$, respectively.
  The operation $\kw{enc}^{-1}(M,N)$ is encryption of $M$ with the inverse of the decryption key $N$,
  which is not an implementable operation but only permitted to occur in patterns.
  We add a sort system on $T_\Sigma$ with sorts $\sortset=\{\kw{impl},
  \kw{pat}, \bot\}$,
  where $\kw{impl}$ denotes implementable terms not containing $\kw{enc}^{-1}$, 
  and $\kw{pat}$ those that may only be used in patterns. 
  The sort $\bot$ denotes ill-formed terms, 
  which do not occur in well-formed processes.
  Names stand for implementable terms, so we let $\namesorts=\{\kw{impl}\}$.
  Substitution is defined homomorphically on the term algebra, 
  except to avoid unimplementable subterms on the form 
  $\kw{enc}^{-1}(M,\kw{dKey}(N))$.

\begin{table*}[tb]
  \begin{mathpar}    
    \inferrule[\textsc{DY True}]
    { \mbox{ }}
    { \ve M \Vdash {} }

\inferrule[\textsc{DY Id}]
    { \ve M, N \Vdash \ve L }
    { \ve M, N \Vdash N, \ve L }

\inferrule[\textsc{DY Copy}]
    { \ve M \Vdash  N, \ve L }
    { \ve M \Vdash  N,  N, \ve L }

\inferrule[\textsc{DY Nil}]
    { \ve M \Vdash \ve L }
    { \ve M \Vdash  (), \ve L }

\inferrule[\textsc{DY Pair}]
    { \ve M \Vdash  N, N', \ve L }
    { \ve M \Vdash  (N,  N'), \ve L }

\inferrule[\textsc{DY Split}]
    { \ve M, N, N'  \Vdash \ve L }
    { \ve M,  (N,  N') \Vdash \ve L }

\inferrule[\textsc{DY Key}]
    { \ve M \Vdash  N, \ve L \quad f\in \{\kw{eKey}, \kw{dKey}\}}
    { \ve M \Vdash  f(N), \ve L }

\inferrule[\textsc{DY Encrypt}]
    { \ve M \Vdash  N, N', \ve L }
    { \ve M \Vdash  \kw{enc}(N,  N'), \ve L }

\inferrule[\textsc{DY Decrypt}]
    { \ve M \Vdash  N' \qquad \ve M,N \Vdash \ve L }
    { \ve M,  \kw{enc}^{-1}(N,  N')  \Vdash \ve L }

\inferrule[\textsc{DY Unencrypt}]
    { \ve M \Vdash  N' \qquad \ve M,N \Vdash \ve L }
    { \ve M, \kw{enc}(N,  \kw{eKey}(N'))  \Vdash \ve L }
  \end{mathpar}

\caption{Dolev-Yao derivability~\cite{haack.jeffrey:pattern-matching-spi}.}  
\label{table:DY-derivability}
\end{table*}

  In order to define $\vars(X)$, 
  we write $\ve{M}\Vdash\ve{N}$ if all $N_i\in\ve{N}$ can be deduced from $\ve{M}$
  in the Dolev-Yao message algebra 
  (i.e., using cryptographic operations such as encryption and decryption).
  For the precise definition, see Table~\ref{table:DY-derivability}.
  The definition of $\vars(X)$ below allows to bind a set $S$ of names 
  only if all names in $S$ can be deduced from the message term $X$ 
  using the other names occurring in $X$.
  This excludes binding an unknown key (cf. Section~\ref{sec:DYalg}).

\instance{PMSPI}{  %
  \trms=\pats=T_\Sigma \qquad \qquad
  \sortset=\Set{\kw{impl},\kw{pat},\bot}\qquad \qquad
  \namesorts=\{\kw{impl}\}  \\
  {\CanSubstitute}={\CanSend}=\Set{(\kw{impl},\kw{impl})}\qquad \qquad
  {\CanReceive}=\Set{(\kw{impl},\kw{impl}),(\kw{impl},\kw{pat})}\\
  \sort{M}=\kw{impl}\text{~if~}\forall N_1,N_2.\;\kw{enc}^{-1}(N_1,N_2)\not\preceq M\\
  \sort{M}=\bot \text{~if~}\exists N_1,N_2.\;\kw{enc}^{-1} (N_1,\kw{dKey}(N_2))\preceq M\\
  \sort{M}=\kw{pat}\text{ otherwise}\\
  \Match(M,\ve{x},X)=\Set{\ve{L}\;:\;M=X\lsubst{\ve{L}}{\ve{x}}}\\
  \vars(X)=\Set{S\subseteq\names{X}\;:\; (\names{X}\setminus S),{X}\Vdash S} \\[\GAP]
{\begin{array}{@{}r@{}l@{}l@{}}  
x\lsubst{\ve{L}}{\ve{y}} ={}& L_i &\text{ if }y_i=x\\
x\lsubst{\ve{L}}{\ve{y}} ={}& x &\text{ otherwise.}\\
\kw{enc}^{-1} (M_1,M_2)\lsubst{\ve{L}}{\ve{y}}={}&\kw{enc}(M_1\lsubst{\ve{L}}{\ve{y}},\kw{eKey}(N)) 
&\text{ when }M_2\lsubst{\ve{L}}{\ve{y}}=\kw{dKey}(N)\\ 
 f(M_1,\dots,M_n)\lsubst{\ve{L}}{\ve{y}}={}&f(M_1\lsubst{\ve{L}}{\ve{y}},\dots,M_n\lsubst{\ve{L}}{\ve{y}}) &\text{ otherwise.}
\end{array}}
} 

As an example, consider the following transitions in \textbf{PMSPI}:

\begin{math}
\begin{array}{l@{}l}
(\nu
a,k,l)&\begin{array}[t]{@{}l@{}l@{}}
	    (&\out{a}{\kw{enc}(\kw{dKey}(l),\kw{eKey}(k))}.\out{a}{\kw{enc}(M,\kw{eKey}(l))}\\
	    \parop&\lin{a}{y}{\kw{enc}(y,\kw{eKey}(k))}\sdot\lin{a}{z}{\kw{enc}^{-1}(z,y)}\sdot\out{c}{z})
	     \end{array}\\
&\gtt
(\nu a,k,l)(
\out{a}{\kw{enc}(M,\kw{eKey}(l))}\parop\lin{a}{z}{\kw{enc}(z,\kw{eKey}(l))}\sdot\out{c}{z})\\
&\gtt
(\nu
a,k,l)\out{c}{M}.
\end{array}
\end{math}

\noindent Note that $\sigma=\lsubst{\kw{dKey}(l)}{y}$ resulting from the first input changed the sort of the second input pattern: 
$\sort{\kw{enc}^{-1}(z,y)}=\kw{pat}$, but $\sort{\kw{enc}^{-1}(z,y)\sigma}= \sort{\kw{enc}(z,\kw{eKey}(l))}=\kw{impl}$.  
However, this is permitted by Definition~\ref{def:subst} (Substitution), since $\kw{impl}\le\kw{pat}$ (implementable terms can be used as channels or messages whenever patterns can be).

Terms (and patterns) are trivially closed under
substitution. All terms in the domain of a well-sorted substitution have sort $\kw{impl}$, so well-sorted substitutions cannot introduce subterms of
the forms $\kw{enc}^{-1} (N_1,N_2)$ or  $\kw{enc}^{-1} (N_1,\kw{dKey}(N_2))$ where none existed; thus $\sort{M\sigma}\le\sort{M}$ as required by Definition~\ref{def:subst}.

To show preservation of pattern variables, we first need some technical results about Dolev-Yao derivability.
\begin{lem}\label{lem:dytech}$ $
  \begin{enumerate}
  \item \label{item:dytech.weak}
    If $\ve{M}\Vdash \ve{N}$, then $\ve{M'}\ve{M}\Vdash \ve{N}$.
  \item \label{item:dytech.sigma}
    If $\ve{M}\Vdash \ve{N}$, then $\ve{M}\sigma\Vdash \ve{N}\sigma$.
  \item \label{item:dytech.impl} 
    If $\sort{N}=\kw{impl}$, then $\names{N}\Vdash N$.
  \item \label{item:dytech.transfer}
    If $\ve{M},N\Vdash \ve{L}$ and $\sort{N}=\kw{impl}$ and $\ve{M}\Vdash N$, then $\ve{M}\Vdash \ve{L}$.
  \end{enumerate}
\end{lem}

\begin{lem}[Preservation of pattern variables]\label{lem:dypres}$ $ \\
 If $\ve x\freshin\sigma$ and $(\names{X}\setminus \ve{x}),{X}\Vdash \ve{x}$ 
then
$(\names{X\sigma}\setminus \ve{x}),X\sigma\Vdash \ve{x}$.
\end{lem}
\proof
Let $\ve{M}=(\names{X}\setminus\ve{x})\sigma$.
By Lemma \ref{lem:dytech}(\ref{item:dytech.sigma}) we get
  $\ve{M},X\sigma\Vdash \ve{x}$, so 
$(\names{X\sigma}\setminus \ve{x}),\ve{M},X\sigma\Vdash \ve{x}$
by Lemma \ref{lem:dytech}(\ref{item:dytech.weak}).
Since $\names{\ve{M}}=(\names{X\sigma}\setminus \ve{x})$, 
 Lemma \ref{lem:dytech}(\ref{item:dytech.impl}) yields that
  $(\names{X\sigma}\setminus\ve{x})\Vdash\ve{M}$.
Finally, by Lemma \ref{lem:dytech}(\ref{item:dytech.transfer}) we get
  $(\names{X\sigma}\setminus \ve{x}),X\sigma\Vdash \ve{x}$.
\qed

The requisites on matching (Definition~\ref{def:pattern-match}) follow from those on substitution.
Lemma \ref{lem:dypres} implies that the set of (well-sorted) processes~\cite{haack.jeffrey:pattern-matching-spi} 
is closed under (well-sorted) substitution, a result which appears not to have been published previously.

\subsection{Nondeterministic computation}\label{sec:lambdaAmb}
The previous examples considered total deterministic notions of computation on the term language.
Here we consider a data term language equipped with
partial non-deterministic evaluation: a lambda calculus extended with the erratic choice operator $\cdot\ErrChoice\cdot$ and the reduction rule $M_1 \ErrChoice M_2\to M_i$ if $i\in\Set{1,2}$. %
Due to non-determinism and partiality, evaluation cannot be part of the substitution function. 
Instead, we define the $\Match$ function to collect all evaluations of the received term, 
which are non-deterministically selected from by the \textsc{In} rule.
This example also highlights the use of object languages with binders, a common application of nominal logic.

We let substitution on terms be the usual capture-avoiding syntactic replacement, 
  and define reduction contexts $\mathcal{R}::=[\,]\mid\mathcal{R}\;M\mid (\mathbold{\lambda} x.M)\;\mathcal{R}$ (we here use the boldface $\mathbold{\lambda}$ rather than the $\lambda$ used in input prefixes).
  Reduction $\to$ is the smallest pre-congruence for reduction contexts 
  that contain the rules for $\beta$-reduction ($\mathbold{\lambda} x.M\;N\to M\lsubst{N}x$) and $\cdot\ErrChoice\cdot$ (see above).
  We use the single-name patterns of Example~\ref{ex:symsempats}, but
  include evaluation in matching. 
\instance{NDLAM}{
\sortset=\Set{s}\qquad\qquad \pats = \nameset\\ %
M ::= a \mid M\;M \mid \mathbold{\lambda} x.M\mid M\ErrChoice M\qquad
\text{where $x$ binds into $M$ in }\mathbold{\lambda} x.M\\
\Match(M,x,x)=\Set{N\;:\;M\to^*N\not\to}\qquad\\
\Match(M,\ve{y},x)=\emptyset\text{~otherwise}}
To avoid confusing the $\lambda$ of the input prefix and the
$\mathbold{\lambda}$ of the term language, we write $\innn{a}{x}$ for $\lin{a}{{x}}{x}$.
As an example, the agent
$P\isdef (\nu a)(\innn{a}{y}\sdot\out{c}{y}\sdot\nil\parop\out{a}{%
((\mathbold{\lambda} x.x\;x) \ErrChoice (\mathbold{\lambda} x.x))}\sdot\nil)$
has the following transitions:
\[
\begin{array}{l}
P\gtt (\nu a)(\out{c}{\mathbold{\lambda}x.x x}\sdot\nil\parop \nil) \gt{\out{c}{\mathbold{\lambda}x.x x}}\nil\\
P\gtt (\nu a)(\out{c}{\mathbold{\lambda}x.x}\sdot\nil\parop \nil) \gt{\out{c}{\mathbold{\lambda}x.x}}\nil.\\
\end{array}
\]
\section{Conclusions and further work}
\label{sec:conclusions}
We have described two features that taken together significantly
improve the precision of applied process calculi:
generalised pattern matching and substitution, 
which allow us to model computations on an arbitrary data term language,
and a sort system which allows us %
   to remove spurious data terms from consideration and %
   to ensure that channels carry data of the appropriate sort.
The well-formedness of processes is thereby guaranteed to be preserved by transitions.
Using these features we have provided representations of other process calculi,
ranging from the simple polyadic pi-calculus to 
the spi-calculus and non-deterministic computations,
in the psi-calculi framework. 
The critera for representation (rather than encoding) are stronger than standard correspondences e.g.~by Gorla, 
and mean that the psi-calculus and the process calculus that it represents are for all practical purposes one and the same.
The meta-theoretic results carry over from the original psi formulations, 
and have been machine-checked in Isabelle for the case of a single name sort
(e.g.~the calculi \textbf{PPI}, \textbf{LINDA} and \textbf{PSPI} in Section~\ref{sec:process-calculi-examples},
and the calculi \textbf{PMSPI} and \textbf{NDLAM} in Section~\ref{sec:advanc-data-struct}).
We have also added sorts to an existing tool for psi-calculi \cite{PWB15:TECS}, 
the Psi-calculi Workbench (\Pwb), which provides an interactive simulator and automatic
bisimulation checker. Users of the tool need only implement the
parameters of their psi-calculus instances, supported by a core library.
In the tool we currently support only tuple patterns, similarly to the
\textbf{PPI} calculus of Section~\ref{sec:polyPi}. 

Future work includes developing a symbolic semantics with more
elaborate pattern matching. For this, a reformulation of the operational semantics 
of Table~\ref{table:struct-free-labeled-operational-semantics}
 in the late style, where input objects are not instantiated until
communication takes place, is necessary.

A comparison of expressiveness to calculi with 
non-binary (e.g., join-calculus~\cite{fournet96.join} or Kell calculus) or 
bidirectional (e.g., dyadic interaction terms~\cite{Honda:1993} 
or the concurrent pattern calculus~\cite{Jay10.ConcurrentPattern}) 
communication primitives would be interesting. 
We here inherit positive results from the pi calculus, such as the encoding of the join-calculus. 

We aim to extend the use of sorts and generalized pattern
matching to other variants of psi-calculi, 
including higher-order psi calculi~\cite{HOPSI} and reliable broadcast
psi-calculi~\cite{pohjola13:Reliable}.
Although assertions and conditions are unsorted, we intend to
investigate adding sorts and
pattern-matching to psi-calculi with non-trivial
assertions~\cite{LMCS11.PsiCalculi}. 

As discussed in Section~\ref{sec:trivally-sorted-calculi}, further 
work is needed for scalable mechanised reasoning about theories that are 
parametric in an arbitrary but fixed name sorting.

\bigskip
\subsubsection*{Acknowledgments. } 
We thank the anonymous reviewers for their helpful comments.


\appendix

\newenvironment{proofcasesdescription}{%
    \begin{description}
}{\end{description}}
\newcommand\proofcase[1]{\vspace{5pt}\item[#1]\hfill\\}
\newcommand{\IH}{\text{IH}}

\section{Full proofs for Section~\ref{sec:process-calculi-examples}}
\label{sec:examples-proofs}

We will assume that the reader is acquainted with the relevant psi-calculi
presented in Section~\ref{sec:process-calculi-examples},
as well as the definitions, notation and terminology of
Sangiorgi~\cite{sangiorgi:expressing-mobility} for polyadic pi-calculus,
Carbone and Maffeis~\cite{carbone.maffeis:expressive-power} for polyadic synchronisation pi-calculus,
and Milner~\cite{milner:cac} for CCS and VPCCS.
We will use their notation except for bound names,
where we will adopt the notation of nominal sets, e.g., we will write $\bn\alpha \freshin Q$
instead of $\bn\alpha \cap \fn Q = \emptyset$.

\subsection{Polyadic Pi-Calculus}

This section contains full proofs of Section~\ref{sec:polyPi} for the polyadic
pi-calculus example.  We use definitions and results of
Sangiorgi~\cite{sangiorgi:expressing-mobility}.  However, we opted to replace
process constants with replication.

For convenience, we repeat definition of the encoding function given in
Example~\ref{sec:polyPi}.

\begin{defi}[Polyadic Pi-Calculus to \textbf{PPi}]$ $\\
Agents:
\[
\begin{array}{rcl}
\semb{P + Q} &=& \caseonly\;\ci{\top}{\semb{P}}\casesep\ci{\top}{\semb{Q}} \\
\semb{[x=y]P} &=& \caseonly\;\ci{x = y}{\semb{P}} \\
\semb{x(\vec{y}).P} &=& \lin{x}{\ve{y}}\langle\ve{y}\rangle.\semb{P} \\
\semb{\overline{x}\langle \vec{y} \rangle.P} &=& \overline{x}\langle \vec{y} \rangle.\semb{P}\\
\semb{0} &=& 0 \\
\semb{P\pll Q} &=& \semb{P}\pll\semb{Q} \\
\semb{\nu x P} &=& (\nu x)\semb{P} \\
\semb{!P} &=& !\semb{P} \\
\end{array}
\]
Actions:
\[\begin{array}{rcl}
\semb{(\nu \vec{y}')\overline{z}\langle\vec{y}\rangle} &=& \boutlabel{z}{\ve{y}'}{\langle\ve{y}\rangle} \\
\semb{x\langle\vec{z}\rangle} &=& \inn{x}{\langle\vec{z}\rangle}\\
\semb{\tau} &=& \tau
\end{array} \]
In the output action $\vec{y}'$ bind into $\ve{y}$ and the residual process, but not into $z$.
\end{defi}

\begin{defi}[{\bf PPi} to Polyadic Pi-Calculus]$ $\\
Process:
\[\begin{array}{rcl}
\bmes{\pass{\emptyframe}} &=& \nil \\
\bmes{\nil} = \bmes{\caseonly} &=& \nil \\
\bmes{\caseonly\;\ci{\varphi_1}{P_1}\casesep\dots\casesep\ci{\varphi_n}{P_n}} &=& \bmes{\ci{\varphi_1}{P_1}} + \dots + \bmes{\ci{\varphi_n}{P_n}} \\
\bmes{!P} &=& !\bmes{P} \\
\bmes{(\nu x)P} &=& \nu x\bmes{P}\\
\bmes{P\pll Q} &=& \bmes{P}\pll\bmes{Q}\\
\bmes{\lin{x}{\ve{y}}\langle\ve{y}\rangle.P} &=& x(\vec{y}).\bmes{P} \\
\bmes{\overline{x}\langle \vec{y} \rangle.P} &=& \overline{x}\langle \vec{y} \rangle.\bmes{P}\\
\end{array}\]
Case clause:
\[\begin{array}{rcl}
\bmes{\ci{x = y}{P}} &=& [x=y]\bmes{P} \\
\bmes{\ci{\top}{P}} &=& \bmes{P}
\end{array}\]
\end{defi}

We prove that the substitution function distributes over the encoding function.

\begin{lem}\label{lem:ppi-substitution}
$\semb{P}\lsubst{\ve z}{\ve y} = \semb{P\{\vec{z}/\vec{y}\} }$
\end{lem}
\proof
By induction on $P$. We consider only the agents where $\bn P \freshin P\{\vec{z}/\vec{y}\}$
\cite[Definition~2.1.1]{sangiorgi:expressing-mobility}. 
We demonstrate the non-trivial cases of the proof in the following.
\begin{itemize}
    \item case $P = P' + Q$. 
        \[\begin{array}{rcll}
        \semb{P' + Q}\lsubst{\vec{z}}{\vec{y}} 
        &=& \caseonly\;\ci{\top\lsubst{\vec{z}}{\vec{y}}}{\semb{P'}\lsubst{\vec{z}}{\vec{y}}}\casesep\ci{\top\lsubst{\vec{z}}{\vec{y}}}{\semb{Q}\lsubst{\vec{z}}{\vec{y}}} \\
         &=&   \caseonly\;\ci{\top}{\semb{P'}\lsubst{\vec{z}}{\vec{y}}}\casesep\ci{\top}{\semb{Q}\lsubst{\vec{z}}{\vec{y}}} \\
         &=&   \caseonly\;\ci{\top}{\semb{P'\{\vec{z}/\vec{y}\}}}\casesep\ci{\top}{\semb{Q\{\vec{z}/\vec{y}\}}} & \text{(IH)}\\
         &=&   \semb{P'\{\vec{z}/\vec{y}\} + Q\{\vec{z}/\vec{y}\}} \\
         &=&   \semb{(P' + Q)\{\vec{z}/\vec{y}\}} \\
        \end{array}\]

    \item case $P = [x = y]Q$. 
        \[\begin{array}{rcll}
                \semb{[x = y]Q}\lsubst{\vec{z}}{\vec{y}} 
                &=& \caseonly\;\ci{x\lsubst{\vec{z}}{\vec{y}} = y\lsubst{\vec{z}}{\vec{y}}}{\semb{Q}\lsubst{\vec{z}}{\vec{y}}} \\
                &=& \caseonly\;\ci{x\lsubst{\vec{z}}{\vec{y}} = y\lsubst{\vec{z}}{\vec{y}}}{\semb{Q\{\vec{z}/\vec{y}\}}} & \text{(IH)}\\
                &=& [x\{\vec{z}/\vec{y}\} = y\{\vec{z}/\vec{y}\}]{\semb{Q\{\vec{z}/\vec{y}\}}} \\
                &=&   \semb{([x = y]Q)\{\vec{z}/\vec{y}\}} \\
        \end{array}\]

    \item case $P = a(\vec{x}).Q$
        \[\begin{array}{rcll}
                \semb{a(\vec{x}).Q}\lsubst{\vec{z}}{\vec{y}} 
                &=&\lin{a\lsubst{\vec{z}}{\vec{y}}}{\ve x}{\langle\ve{x}\rangle}.\semb{Q}\lsubst{\vec{z}}{\vec{y}} 
                & \text{(From assumption $\vec{x}\freshin\lsubst{\vec{z}}{\vec{y}} $)} \\
                &=&\lin{a\lsubst{\vec{z}}{\vec{y}}}{\ve x}{\langle\ve{x}\rangle}.{\semb{Q\{\vec{z}/\vec{y}\}}} & \text{(IH)}\\
                &=& a\{\vec{z}/\vec{y}\}(\vec{x}).{\semb{Q\{\vec{z}/\vec{y}\}}} \\
                &=&   \semb{(a(\vec{x}).Q)\{\vec{z}/\vec{y}\}} \\
        \end{array}\]
\end{itemize}
\endproof

\noindent The following is the proof of the strong operational correspondence with
respect to the labeled semantics of polyadic pi-calculus \cite[page 30]{sangiorgi:expressing-mobility}.
\proof[Proof of Theorem \ref{thm:ppi-strong-operational-corresp.}]$ $
\begin{enumerate}
\item We show that if $\trans{P}{\beta}{P'}$ then for all
  $\alpha\in\semb\beta$ we have $\trans{\semb{P}}{\alpha}{\semb{P'}}$
  by induction on the derivation of the transition.

\begin{proofcasesdescription}
  \proofcase{\textsc{ALP}} Trivial, since psi-calculi processes are identified up
            to alpha equivalence.

  \proofcase{\textsc{OUT}} Assume
        $\trans{\overline{x}\langle \vec{y}\rangle.P}
               {\overline{x}\langle\vec{y}\rangle}{P}$ and
         $\alpha \in \Set{\out{x}{\langle\vec{y}\rangle}} = \semb{\overline{x}\langle\vec{y}\rangle}$.
    Since $\emptyframe\vdash x \sch x$ and $\semb{\out{x}{\langle\vec y\rangle}.P} = \out{x}{\langle\vec y\rangle}.\semb{P}$
    and $\alpha = {\out{x}{\langle\vec y\rangle}}$, we can derive
    $\trans{\out{x}{\langle\vec y\rangle}.\semb{P}}{\out{x}{\langle\vec y\rangle}}{\semb{P}}$.

    \proofcase{\textsc{INP}} 
    Assume $\trans{x(\vec{y}).P}{x\langle\vec{z}\rangle}{P\{\vec{z}/\vec{y}\}}$,
    and $\vec{z}$ and $\vec{y}$ are of the same arity (in the
    trminology of Sangiorgi, $\vec{z} : \vec{y}$), and also 
    $\alpha \in \semb{\beta} = \Set{\inn{x}{\langle\vec{z}\rangle}}$.
    Note that $\semb{x(\vec{y}).P} = {\lin{x}{\ve y}{\langle\ve{y}\rangle}.\semb{P}}$
    and $\ve{z}\in\Match(\langle \ve{z}\rangle,\ve y, \langle\ve y\rangle)$.
    By using $\emptyframe \vdash x\sch x$, we can derive 
    $\trans{\lin{x}{\ve y}{\langle\ve{y}\rangle}.\semb{P}}{\inn{x}{\langle\vec{z}\rangle}}{\semb{P}\lsubst{\ve z}{\ve y}}$
    with the \textsc{In} rule. By applying Lemma~\ref{lem:ppi-substitution}, we complete this proof case.

    \proofcase{\textsc{SUM}} 
    Assume $\trans{P + Q}{\beta}{P'}$ and $\alpha \in \semb{\beta}$, and also $\trans{P}{\beta}{P'}$. 
    The induction hypothesis is that for every $\alpha\in\semb{\beta}$,
     $\trans{\semb{P}}{\alpha}{\semb{P'}}$. 
     We can then derive $\trans{\caseonly\;{\ci{\top}{\semb{P}}\casesep\ci{\top}{\semb{Q}}}}{\alpha}{\semb{P'}}$
    with the \textsc{Case} rule for every $\alpha\in\semb{\beta}$.

    \proofcase{\textsc{PAR}} 
    Assume $\trans{P \parop Q }{\beta}{P' \parop Q}$ and $\alpha\in\semb\beta$, and
    $\trans{P}{\beta}{P'}$ with $\bn \beta \cap \fn Q = \emptyset$.
    The induction hypothesis is that for every $\alpha\in\semb\beta$, $\trans{\semb{P}}{\alpha}{\semb{P'}}$.
    From the definition of $\semb{\beta}$ we get that $\bn \alpha\freshin \semb{Q}$ for any $\alpha\in\semb\beta$.
    By applying the \textsc{Par} rule, we obtain the required transitions
    $\trans{\semb{P}\parop\semb{Q}}{\alpha}{\semb{P'}\parop\semb{Q}}$.
    
    \proofcase{\textsc{COM}}
    Assume $\trans{P \parop Q}{\tau}{\nu \vec{y}' (P' \parop Q')}$
    with $\vec{y}' \cap \fn Q = \emptyset$.
    Also assume $\trans{P}{(\nu \vec{y'})\overline{x}\langle\vec{y}\rangle}{P'}$ and
    $\trans{Q}{x\langle\vec{y}\rangle}{Q'}$. The induction hypothesis is that for every
    $\alpha' \in \semb{(\nu \vec{y'})\overline{x}\langle\vec{y}\rangle}$ and
    $\alpha'' \in \semb{x\langle\vec{y}\rangle}$, 
    $\trans{\semb{P}}{\alpha'}{\semb{P'}}$ and
    $\trans{\semb{Q}}{\alpha''}{\semb{Q'}}$
    Moreover, we note that $\emptyframe\vdash x \sch x$ and $\vec{y}' \freshin \semb{Q}$.
    We then choose $\alpha'$ and $\alpha''$ and alpha-variants of the
    frames of $\semb{P}$ and $\semb{Q}$ that are sufficiently fresh to allow
    the derivation $\trans{\semb{P}\parop\semb{Q}}{\tau}{(\nu \ve{y'})(\semb{P'}\parop\semb{Q'})}$
    with the \textsc{Com} rule.

    \proofcase{\textsc{MATCH}} 
    Assume $\trans{[x = x]P}{\beta}{P'}$ and $\alpha \in \semb\beta$, as well as
    $\trans{P}{\beta}{P'}$. The induction hypothesis is that 
    $\trans{\semb{P}}{\alpha}{\semb{P'}}$. Since $\emptyframe\vdash x = x$ and 
    $\caseonly\;{\ci{x=x}{\semb{P}}} = \semb{[x=x]P}$,
    we derive $\trans{\caseonly\;{\ci{x=x}{\semb{P}}}}{\alpha}{\semb{P'}}$
    with the \textsc{Case} rule.
    
    \proofcase{\textsc{REP}} 
    Assume $\trans{!P}{\beta}{P'}$ and $\alpha\in\semb\beta$.
    Moreover, assume $\trans{P\parop !P}{\beta}{P'}$ and hence by the induction hypothesis
    $\trans{\semb{P \parop !P}}{\alpha}{\semb{P'}}$. We compute
    $\semb{P} \parop !\semb{P} = \semb{P \parop !P}$ and apply the \textsc{Rep} rule
    to obtain $\trans{!\semb{P}}{\alpha}{\semb{P'}}$.

    \proofcase{\textsc{RES}} 
    Assume $\trans{\nu x P}{\beta}{\nu x P'}$ where $x\not\in n(\beta)$ and
    $\alpha\in\semb\beta$. Also assume $\trans{P}{\beta}{P'}$. The
    induction hypothesis is $\trans{\semb{P}}{\alpha}{\semb{P'}}$.
    Now by obtaining $x\freshin \alpha$ from assumptions and computing
    $\semb{\nu x P} = (\nu x)\semb{P}$, we derive
    $\trans{(\nu x)\semb{P}}{\alpha}{(\nu x)\semb{P'}}$
    with the \textsc{Scope} rule.

    \proofcase{\textsc{OPEN}} 
    Let $\beta = {(\nu x, \vec{y'})\overline{z}\langle \vec{y}\rangle}$.
    Assume $\trans{\nu x P}{\beta}{P'}$
    and $x\not= z, x \in \vec{y} - \vec{y'}$ and
    $\alpha \in \semb{\beta} = \Set{\boutlabel{z}{\ve{y}''}{\langle\vec{y}\rangle}\;:\;\vec{y}''=\pi \cdot x,\vec{y}'} $.
    The induction hypothesis is that for every
    $\alpha' \in \semb{(\nu \vec{y}')\overline{z}\langle\vec{y}\rangle} =
    \Set{\boutlabel{z}{\ve{y}''}{\langle\vec{y}\rangle}\;:\;\vec{y}''=\pi \cdot \vec{y'}}$
    we have $\trans{\semb{P}}{\alpha'}{\semb{P'}}$.
    We choose $\alpha' = \boutlabel{z}{\ve{y}'}{\ve{y}}$ and, by having
    $\semb{\nu x P} = (\nu x)\semb{P}$, we derive
    $\trans{(\nu x)\semb{P}}{\boutlabel{z}{x,\vec{y'}}
    {\langle\vec{y}\rangle}}{\semb{P'}}$
    with the \textsc{Open} rule.
    The side conditions of \textsc{Open} ($x \freshin \vec{y}', z$ and $x \in n(\vec{y})$)
    follow from assumptions.

    From the assumption $\alpha \in \semb{\beta}$, it follows that, for any permutation $\pi$, $\alpha$ is of the form 
    $\boutlabel{z}{\pi\cdot x,\ve{y}'}{\langle\vec{y}\rangle}$.
    By applying Lemma~\ref{lem:bnReordering}, we get the required $\alpha$ and transition
    $\trans{(\nu x)\semb{P}}{\alpha}{\semb{P'}}$. And this concludes this proof case.

\end{proofcasesdescription}

\item We now show that if $\trans{\semb{P}}{\alpha}{P''}$ then
    $\trans{P}{\beta}{P'}$ where $\alpha \in \semb{\beta }$ and $\semb{P' } = P''$.
    We proceed by by induction on the derivation of the transition.
    We show the interesting cases:

\begin{proofcasesdescription}
    \proofcase{\textsc{Case}} 
    Assume $\trans{\semb{P}}{\alpha}{P''}$. 
    By inversion of the \textsc{Case} rule,
    $\semb{P}$ is of the form $\caseonly\;\ci{\vec{\varphi}}{\vec{P}}$.
    Since $P_C = \caseonly\;{\ci{\ve{\varphi}}{\ve{P}}}$ is in the range of
    $\semb{\cdot}$, either $P_C =
    \ci{\top}{\semb{P}}\casesep\ci{\top}{\semb{Q}}$,  $P_C =
    \ci{\top}{\semb{Q}}\casesep\ci{\top}{\semb{P}}$ or $P_C =
    \caseonly\;{\ci{x=y}{\semb{P}}}$. We proceed by case analysis:
    \begin{enumerate}
      \item When $P_C = \ci{\top}{\semb{P}}\casesep\ci{\top}{\semb{Q}}$, we
          note that $\semb{P + Q} = P_C$ and imitate the derivation of $P''$
          from $P_C$ with the derivation $\trans{P + Q}{\beta}{P'}$, using the
          \textbf{SUM} rule and the fact obtained from induction hypothesis $\alpha\in\semb\beta$.
      \item The case when $P_C = \ci{\top}{\semb{Q}}\casesep\ci{\top}{\semb{P}}$ is symmetric to the previous case.
      \item When $P_C = \caseonly\;{\ci{x=y}{\semb{P}}}$, since $\emptyframe \vdash x=y$
          by the induction hypothesis, $x=y$. We note that $\semb{[x=x]P } =
          P_C$ and imitate the derivation of $P''$ from $P_C$ with the
          derivation $\trans{[x=x]P}{\beta}{P'}$, using the \textbf{MATCH}
          rule and the fact obtained from induction hypothesis $\alpha\in\semb\beta$.
    \end{enumerate}

    \proofcase{\textsc{Open}}
    Assume $\trans{\semb{P}}{\boutlabel{z}{\vec{y}\cup\Set{x}}{\langle\vec{y}'\rangle}}{P''}$. 
    Because $P''$ is derived with the \textsc{Open} rule,
    $\semb{P}$ is of the form $(\nu x)R$. Since $(\nu x) R$ is in the range of $\semb{\cdot}$,
    $P = \nu x R'$, where $R = \semb{R'}$.
    From induction hypothesis, we have that 
    $\trans{R}{\boutlabel{z}{\vec{y}}{\langle\vec{y}'\rangle}}{P''}$ and
    ${\boutlabel{z}{\vec{y}}{\langle\vec{y}'\rangle}}\in\semb{\beta'}$ and
    $\trans{R'}{\beta'}{P'}$ and lastly $\semb{P'} = P''$. 
    Thus, we use $\beta' = (\nu \vec{y})\overline{z}\langle\vec{y}'\rangle$ 
    as it gives us ${\boutlabel{z}{\vec{y}}{\langle\vec{y}'\rangle}}\in\semb{\beta'}$
    to derive, by using the rule \textbf{OPEN}, $\trans{\nu x R'}{(\nu x,\vec{y})\overline{z}\langle\vec{y}'\rangle}{P'}$.
    Clearly, ${\boutlabel{z}{\vec{y}\cup\Set{x}}{\langle\vec{y}'\rangle}} \in \semb{(\nu x,\vec{y})\overline{z}\langle\vec{y}'\rangle}$
    for every insertion of $x$. \qed

\end{proofcasesdescription}
\end{enumerate}

\noindent From the strong operational correspondence, we obtain full abstraction. 
We use Sangiorgi's definition of bisimulation and congruence for the polyadic
pi-calculus~\cite[page~42]{sangiorgi:expressing-mobility}.

\begin{thm}
    For polyadic-pi calculus agents $P$ and $Q$ we have
    $P \sim_e^c Q$ iff $\semb{P} \sim \semb{Q}$.
\end{thm}

\proof

For direction $\Leftarrow$, assume $\semb{P}\sim\semb{Q}$. We claim that the
relation $\mathcal{R} = \Set{(P, Q): \semb{P} \sim \semb{Q}}$ is an \emph{early
congruence} in the polyadic pi-calculus. 

Firs let us consider the simulation case.  Assume $\trans{P}{\beta}{P'}$. Then,
we need to show that there exists $Q'$ such that $\trans{Q}{\beta}{Q'}$ and $(P',Q')
\in \mathcal{R}$.  By Theorem~\ref{thm:ppi-strong-operational-corresp.} (1), we
get $\trans{\semb{P}}{\alpha}{\semb{P'}}$ for any $\alpha\in\semb{\beta}$.  By
Theorem~\ref{thm:ppi-strong-operational-corresp.} (2) and using the assumption
$\alpha\in\semb{\beta}$ as well as the fact $\semb{P}\sim\semb{Q}$, we derive
$\trans{\semb{Q}}{\alpha}{\semb{Q'}}$. From the simulation clause and that
$\semb{P}$ and $\semb{Q}$ are congruent we get that $\semb{P'}\sim\semb{Q'}$.
Hence, $(P', Q') \in \mathcal{R}$.  The symmetry case follows from the symmetry
of $\sim$. Thus, $\mathcal{R}$ is an early bisimulation.  Since $\mathcal{R}$
is closed under all substitutions by Lemma~\ref{lem:ppi-substitution}, it is
also an early congruence.

Now let us consider the other direction $\Rightarrow$.
First, assume $P\sim_e^cQ$. We claim the relation
$\{(\emptyframe,\semb{P},\semb{Q}): P \sim_e^c Q\}$ is a congruence in \textbf{PPI}.
The static equivalence and extension of arbitrary assertion cases are trivial
since there is unit assertion only. Symmetry follows from symmetry of $\sim_e^c$, and
simulation follows by Theorem~\ref{thm:ppi-strong-operational-corresp.} and the
fact that $\sim_e^c$ is an early congruence.
\qed

\proof[Proof of Theorem \ref{thm:polypiabs}]

By structural induction on $P$. We only consider the $\caseonly$ agent
since the other cases are trivial.

\begin{proofcasesdescription}
\proofcase{$P=\caseonly\;\ci{\varphi_1}{P_1}\casesep\dots\casesep\ci{\varphi_n}{P_n}$}
We have one induction hypothesis $\IH_i$ for every $i\in\{1..n\}$, namely that $P_i \sim \semb{\bmes{P_i}}$.

We proceed by induction on $n$.
    \begin{proofcasesdescription}
    \proofcase{Base case $n=0$} $\semb{\bmes{\caseonly}} = \semb{\nil} = \nil$.
        By reflexivity of $\sim$, $\nil\sim\nil$.

    \proofcase{Induction step $n+1$}
        The $\IH$ for this case is
        \[\semb{\bmes{\caseonly\;\ci{\varphi_1}{P_1}\casesep\dots\casesep\ci{\varphi_n}{P_n}}}
        \sim \caseonly\;\ci{\varphi_1}{P_1}\casesep\dots\casesep\ci{\varphi_n}{P_n} = P'\]

        We need to show that $Q \sim \semb{\bmes{Q}}$ for $Q =
        \caseonly\;\ci{\varphi_1}{P_1}\casesep\dots\casesep\ci{\varphi_n}{P_n}
                      \casesep\ci{\varphi_{n+1}}{P_{n+1}}$.

        We thus compute
        \[\begin{array}{rcl}
        \semb{\bmes{Q}} &=&
            \semb{\bmes{\ci{\varphi_1}{P_1}}+\dots+\bmes{\ci{\varphi_n}{P_n}}
                      + \bmes{\ci{\varphi_{n+1}}{P_{n+1}}}} \\
                        &=& 
            \caseonly\;\ci{\top}{\semb{\bmes{\ci{\varphi_1}{P_1}}}}\casesep\dots
                      \casesep\ci{\top}{\semb{\bmes{\ci{\varphi_n}{P_n}}}}
                      \casesep \ci{\top}{\semb{\bmes{\ci{\varphi_{n+1}}{P_{n+1}}}}} \\

                        &\sim&  \text{(by Lemma~\ref{lemma:flatten-case})} \\
                        &    &
            \caseonly\;
                \ci{\top}{
                    (\caseonly\;\ci{\top}{\semb{\bmes{\ci{\varphi_1}{P_1}}}}\casesep\dots
                      \casesep\ci{\top}{\semb{\bmes{\ci{\varphi_n}{P_n}}}}) }
                      \casesep \ci{\top}{\semb{\bmes{\ci{\varphi_{n+1}}{P_{n+1}}}}} \\

                        &\sim&  \text{(by $\IH$)} \\
                        &    &
            \caseonly\;
                \ci{\top}{
                    (\caseonly\;\ci{\varphi_1}{P_1}\casesep\dots\casesep\ci{\varphi_n}{P_n})}
                      \casesep \ci{\top}{\semb{\bmes{\ci{\varphi_{n+1}}{P_{n+1}}}}} \\
                        &  = &
            \caseonly\;
                \ci{\top}{P'}
                      \casesep \ci{\top}{\semb{\bmes{\ci{\varphi_{n+1}}{P_{n+1}}}}} \\
                        & = & Q'
        \end{array}\]

        We distinguish two cases of $\varphi_{n+1}$:
        \begin{proofcasesdescription}
        \proofcase{Case $\varphi_{n+1} = \top$}
            \[\begin{array}{rcl}
                Q' &=&
                \caseonly\;\ci{\top}{P'}
                      \casesep \ci{\top}{\semb{\bmes{\ci{\top}{P_{n+1}}}}} \\
                   &=&
                \caseonly\;\ci{\top}{P'}
                      \casesep \ci{\top}{\semb{\bmes{P_{n+1}}}} \\
                   &\sim& \text{(by $\IH_{n+1}$)} \\
                & & \caseonly\;\ci{\top}{P'}
                      \casesep \ci{\top}{P_{n+1}} \\
                &\sim& \text{(by Lemma~\ref{lemma:flatten-case})} \\
                &    & \caseonly\;\ci{\varphi_1}{P_1}\casesep\dots\casesep\ci{\varphi_n}{P_n}
                        \casesep\ci{\top}{P_{n+1}} = Q \\
            \end{array}\]
            We conclude this case.

        \proofcase{Case $\varphi_{n+1} = x = y$}
            \[\begin{array}{rcl}
                Q' &=&
                \caseonly\;\ci{\top}{P'}
                      \casesep \ci{\top}{\semb{\bmes{\ci{x = y}{P_{n+1}}}}} \\
                   &=&
                \caseonly\;\ci{\top}{P'}
                      \casesep \ci{\top}{(\caseonly\;\ci{x = y}{\semb{\bmes{P_{n+1}}}})} \\
                   &\sim& \text{(by $\IH_{n+1}$)} \\
                &  &\caseonly\;\ci{\top}{P'}
                      \casesep \ci{\top}{(\caseonly\;\ci{x = y}{P_{n+1}})} \\
                &\sim& \text{(by Lemma~\ref{lemma:flatten-case})} \\
                &  &\caseonly\;\ci{\varphi_1}{P_1}\casesep\dots\casesep\ci{\varphi_n}{P_n}
                      \casesep \ci{\top}{(\caseonly\;\ci{x = y}{P_{n+1}})} \\
                &\sim& \text{(by Lemma~\ref{lemma:flatten-case})} \\
                &    & \caseonly\;\ci{\varphi_1}{P_1}\casesep\dots\casesep\ci{\varphi_n}{P_n}
                        \casesep\ci{x = y}{P_{n+1}} = Q \\
            \end{array}\]
            By concluding this case, we conclude the proof.\qed
        \end{proofcasesdescription}

    \end{proofcasesdescription}

\end{proofcasesdescription}

\begin{lem}
\label{lemma:ppi-to-psi-injective}
$\semb{\cdot}$ is injective, that is, for all $P, Q$,
if $\semb{P} = \semb{Q}$ then $P = Q$.
\end{lem}
\proof
By induction on $P$ and $Q$ while inspecting all possible cases.
\qed

\begin{lem}
\label{lemma:ppi-to-psi-surjective}
$\semb{\cdot}$ is surjective up to $\sim$, that is, for every $P$
there is a $Q$ such that $\semb{Q} \sim P$.
\end{lem}

\proof
By induction on the well-formed agent $P$.

\begin{proofcasesdescription}

\proofcase{Case $\lin{x}{\ve{y}}\langle\ve{y}\rangle.P'$} By induction there is
$Q'$ such that $\semb{Q'} \sim P'$. Let $Q = x(\vec{y}).Q'$. Then
$\semb{Q} = \semb{x(\vec{y}).Q'} = \lin{x}{\ve{y}}\langle\ve{y}\rangle.\semb{Q'}
\sim \lin{x}{\ve{y}}\langle\ve{y}\rangle.P'=P$.

\proofcase{Case $\overline{x}\langle \vec{y} \rangle.P'$} By induction there is $Q'$ such that
$\semb{Q'} \sim P'$. Let $Q = \overline{x}\langle \vec{y} \rangle.Q'$. Now
$\semb{Q} = \overline{x}\langle \vec{y} \rangle.\semb{Q'} \sim
\overline{x}\langle \vec{y} \rangle.P'=P$.

\proofcase{Case $P\pll P'$} By induction there are $Q',Q''$ such that
$\semb{Q'}\sim P$ and $\semb{Q''}\sim P'$. Then let $Q = Q'\pll Q''$, obtaining
$\semb{Q} = \semb{Q'}\pll\semb{Q''} \sim P\pll P'=P$.

\proofcase{Case $(\nu x)P$} By induction there is $Q'$ such that $\semb{Q'}\sim P$.
Let $Q = \nu x Q'$. Then $\semb{Q} = (\nu x)\semb{Q'} \sim (\nu x)P$.

\proofcase{Case $!P$} By induction there is $Q'$ such that $\semb{Q'}\sim P$.
Let $Q = {!Q'}$. Then $\semb{Q} = {!\semb{Q'}} \sim {!P}$.

\proofcase{Case $\pass{\emptyframe}$} Let $Q = \nil$. Then $\semb{Q} = \nil \sim
\pass{\emptyframe}$.

\proofcase{Case $\caseonly\;\ci{\ve{\varphi}}{\ve{P'}}$} 
The induction hypothesis IH${}_{\caseonly}$ is that for every $P'_i$ there is $Q'_i$ such that
$\semb{Q'_i} \sim P'_i$. The proof proceeds by induction on the length of
$\ve\varphi$.
    \begin{proofcasesdescription}
    \proofcase{Base case} Let $Q = \nil$, then $\semb{Q} = \nil \sim \caseonly$.
    \proofcase{Induction step} At this step, we get the following IH
        \[\semb{Q''}\sim
            \caseonly\;\ci{\varphi_1}{P_1}\casesep\dots\casesep\ci{\varphi_n}{P_n}\]
        We need to find $\semb{Q}$ such that
        \[\semb{Q}\sim
            \caseonly\;\ci{\varphi_1}{P_1}\casesep\dots\casesep\ci{\varphi_n}{P_n}
            \casesep\ci{\varphi_{n+1}}{P_{n+1}}
        \]
        By IH${}_{\caseonly}$ for $P'_{n+1}$ we get $\semb{Q'_{n+1}} \sim P_{n+1}$.
        We proceed by case analysis on $\varphi_{n+1}$.
        \begin{proofcasesdescription}
        \proofcase{Case $\varphi_{n+1} = \top$} Let $Q = Q'' + Q'_{n+1}$. Then
            \[\begin{array}{rcl}
            \semb{Q} &=& \caseonly\;\ci{\top}{\semb{Q''}}\casesep\ci{\top}{\semb{Q'_{n+1}}} \\
                &\sim& \caseonly\;\ci{\top}{(\caseonly\;\ci{\varphi_1}{P_1}\casesep\dots\casesep\ci{\varphi_n}{P_n})}\\
                &    & \casesep\ci{\top}{\semb{Q'_{n+1}}} \\
                &\sim& \caseonly\;\ci{\top}{(\caseonly\;\ci{\varphi_1}{P_1}\casesep\dots\casesep\ci{\varphi_n}{P_n})}\\
                &    & \casesep\ci{\top}{P_{n+1}} \\
                &\sim& \text{(by Lemma \ref{lemma:flatten-case})} \\
                &    & \caseonly\;\ci{\varphi_1}{P_1}\casesep\dots\casesep\ci{\varphi_n}{P_n}\\
                &    & \casesep\ci{\top}{P_{n+1}} \\
            \end{array} \]

        \proofcase{Case $\varphi_{n+1} = x = y$} Let $Q = Q'' + [x=y]Q'_{n+1}$. Then
            \[\begin{array}{rclr}
            \semb{Q} &=& \caseonly\;\ci{\top}{\semb{Q''}}\casesep\ci{\top}{\semb{[x=y]Q'_{n+1}}} \\
                &\sim& \caseonly\;\ci{\top}{(\caseonly\;\ci{\varphi_1}{P_1}\casesep\dots\casesep\ci{\varphi_n}{P_n})}\\
                &    & \casesep\ci{\top}{(\caseonly\;\ci{x = y}{\semb{Q'_{n+1}}})} \\
                &\sim& \caseonly\;\ci{\top}{(\caseonly\;\ci{\varphi_1}{P_1}\casesep\dots\casesep\ci{\varphi_n}{P_n})}\\
                &    & \casesep\ci{\top}{(\caseonly\;\ci{x = y}{P_{n+1}})} \\
                &\sim& \text{(by Lemma \ref{lemma:flatten-case})} \\
                &    & \caseonly\;\ci{\varphi_1}{P_1}\casesep\dots\casesep\ci{\varphi_n}{P_n}\\
                &    & \casesep\ci{\top}{(\caseonly\;\ci{x = y}{P_{n+1}})} \\
                &\sim& \text{(by permuting and applying Lemma \ref{lemma:flatten-case})} \\
                &    & \caseonly\;\ci{\varphi_1}{P_1}\casesep\dots\casesep\ci{\varphi_n}{P_n}\casesep\ci{x=y}{P_{n+1}}
            \end{array} \]
            This case concludes the proof. \qed
        \end{proofcasesdescription}
    \end{proofcasesdescription}
\end{proofcasesdescription}

\subsection{Polyadic Synchronisation Pi-Calculus}

In this section, we include the full proofs of Section~\ref{sec:polySynchPi}.
We use definitions and results for polyadic synchronisation pi-calculus, ${}^{e}\pi$,
by Carbone and Maffeis~\cite{carbone.maffeis:expressive-power}.

We give an explicit definition of encoding function defined in
Example~\ref{sec:polySynchPi}.

\begin{defi}[Polyadic synchronisation pi-calculus to {\bf PSPi}]$ $\\
Agents:
\[
\begin{array}{rcl}
\semb{\ve{x}(y).P} &=& \lin{\langle\ve{x}\rangle}{y}y.\semb{P} \\
\semb{\ve{x}\langle y\rangle.P} &=& \out{\langle\ve{x}\rangle}{y}.\semb{P} \\
\semb{P\pll Q} &=& \semb{P}\pll\semb{Q} \\
\semb{(\nu x) P} &=& (\nu x)\semb{P} \\
\semb{!P} &=& !\semb{P} \\
\semb{0} &=& 0 \\

\semb{\Sigma_i\alpha_i . P_i} &=& \caseonly\;\ci{\top_i}\semb{\alpha_i.P_i} \\
\end{array}
\]
Actions:
\[\begin{array}{rcl}
\semb{\vec{x}\langle \nu c\rangle} &=& \boutlabel{\langle\vec{x}\rangle}{c}{c} \\
\semb{\vec{x}\langle c\rangle} &=& \out{\langle\vec{x}\rangle}{c} \\
\semb{\tau} &=& \tau \\
\semb{\vec{x}(y)} &=&  \text{undefined}  \\
\end{array}\]
\end{defi}

\begin{defi}[{\bf PSPi} to Polyadic synchronisation pi-calculus]
\[\begin{array}{rcl}
\bmes{\pass{\emptyframe}} &=& \nil \\
\bmes{\nil} &=& \nil \\
\bmes{!P} &=& !\bmes{P} \\
\bmes{(\nu x)P} &=& (\nu x)\bmes{P}\\
\bmes{P\pll Q} &=& \bmes{P}\pll\bmes{Q}\\
\bmes{\langle\ve{a}\rangle y.P} &=& \overline{a}\langle y\rangle . \bmes{P}\\
\bmes{\lin{\ve{x}}{y}{y}.P}  &=& \overline{x}(y).\bmes{P}\\
\bmes{\tau.P} &=& \tau.\bmes{P} \\
\bmes{\caseonly\;\ci{\top}{\alpha_i.P_i}} &=& \Sigma_i \bmes{\alpha_i.P_i} \\
\end{array}\]
\end{defi}

\begin{lem}
\label{lemma:polySynchStructCong}
If $P \equiv Q$ then $\semb{P} \sim \semb{Q}$
\end{lem}
\proof
  The relation $\mathcal{R} = \{(P,Q): \semb{P}\sim\semb{Q}\}$ satisfies 
the axioms defining $\equiv$ and is also a process congruence. Since $\equiv$
is the least such congruence, ${\equiv} \subseteq \mathcal{R}$.
\qed

\proof[Proof of Lemma~\ref{lem:polySynchTransition}] $ $
\begin{enumerate}

\item By induction on the derivation of $P'$, avoiding $z$.
  \begin{proofcasesdescription}
    \proofcase{\textsc{Prefix}} Here $\trans{\Sigma_i \vec{x}_i(y_i). P_i}{\vec{x}_i(y_i)}{P_i}$. We have that
      \[\begin{array}{rcl}
      \semb{\Sigma_i \vec{x}_i(y_i). P_i} & = &
      \caseonly\;\ci{\top}{\lin{\langle \vec{x}\rangle}{y_1}{y_1}.\semb{P_1}}\casesep \\
            & & \cdots \casesep\ci{\top}{\lin{\langle \vec{x}\rangle}{y_i}{y_i}.\semb{P_i}}
        \end{array}\]

       Since $\Match(z,\langle y_i\rangle,y_i)=\{z\}$, we can use the
        \textsc{Case} and \textsc{In} rules to derive the transition
      \[\begin{array}{l}
          \caseonly\;\ci{\top}{\lin{\langle \vec{x}_1\rangle}{y_1}{y_1}.\semb{P_1}}\casesep \\
                    \,\cdots\casesep\ci{\top}{\lin{\langle \vec{x}_i\rangle}{y_i}{y_i}.\semb{P_i}} %
            \quad\goesto{\inn{\langle\vec{x}\rangle}{z}}\quad %
          \semb{P_i}\lsubst{z}{y_i}
        \end{array} \]
      Finally, we have $P'' = \semb{P_i}\lsubst{z}{y_i}$ and use reflexivity of $\sim$ to conclude this case.

    \proofcase{\textsc{Bang}} Here $\trans{P\parop !P}{\vec{x}(y)}{P'}$ and by
    induction,
    $\trans{\semb{P}\parop!\semb{P}}{\inn{\langle\vec{x}\rangle}{z}}{P''}$ with
    $P'' \sim \semb{P'}\lsubst{z}{y}$. By rule \textsc{Rep}, we also have
    that $\trans{!\semb{P}}{\inn{\langle\vec{x}\rangle}{z}}{P''}$.

    \proofcase{\textsc{Par}} Here $\trans{P}{\vec{x}(y)}{P'}$, $y \freshin Q$
    and by induction, $\trans{\semb{P}}{\inn{\langle\vec{x}\rangle}{z}}{P''}$
    with $P'' \sim \semb{P'}\lsubst{z}{y}$. Using the \textsc{Par} rule we
    derive
    $\trans{\semb{P}\parop\semb{Q}}{\inn{\langle\vec{x}\rangle}{z}}{P'\parop\semb{Q}}$.
    Since $\sim$ is closed under $|$, $P''\parop\semb{Q} \sim
    \semb{P'}\lsubst{z}{y}\parop\semb{Q}$. Finally, since $y \freshin Q$,
    $\semb{P'}\lsubst{z}{y}\parop\semb{Q} = \semb{P'\parop Q}\lsubst{z}{y}$.

    \proofcase{\textsc{Struct}} Here $P \equiv Q$, $\trans{Q}{\vec{x}(y)}{Q'}$
    and $Q' \equiv P'$. By induction we obtain $Q''$ such that
    $\trans{\semb{Q}}{\inn{\langle\vec{x}\rangle}{z}}{Q''}$ where $Q'' \sim
    \semb{Q'}\lsubst{z}{y}$. By Lemma~\ref{lemma:polySynchStructCong},
    $\semb{P} \sim \semb{Q}$ and $\semb{Q'} \sim \semb{P'}$, and by expanding the definition
    of $\sim$, we obtain $\semb{Q'}\lsubst{z}{y} \sim \semb{P'}\lsubst{z}{y}$. Since
    $\semb{P} \sim \semb{Q}$ and
    $\trans{\semb{Q}}{\inn{\langle\vec{x}\rangle}{z}}{Q''}$, there exists $P''$
    such that $\trans{\semb{P}}{\inn{\langle\vec{x}\rangle}{z}}{P''}$ and $Q''
    \sim P''$. By using the transitivity of $\sim$, we conclude $P'' \sim \semb{P'}\lsubst{z}{y}$.

    \proofcase{\textsc{Res}} Here $\trans{P}{\vec{x}(y)}{P'}$, $a \neq y$, $a
    \neq z$ and $a \freshin \vec{x}$. By induction,
    $\trans{\semb{P}}{\inn{\langle\vec{x}\rangle}{z}}{P''}$ with $P'' \sim
    \semb{P'}\lsubst{z}{y}$. We can then
    derive $\trans{(\nu a)\semb{P}}{\inn{\langle\vec{x}\rangle}{z}}{(\nu
    a)P''}$. Since $\sim$ is closed under restriction, $(\nu a)P'' \sim
    (\nu a)(\semb{P'}\lsubst{z}{y})$. Finally, $a$ is sufficiently fresh to show
    that $(\nu a)(\semb{P'}\lsubst{z}{y}) = ((\nu a)\semb{P'})\lsubst{z}{y}$

  \end{proofcasesdescription}

  \item By induction on the derivation of $P''$, avoiding $y$.
  \begin{proofcasesdescription}
    \proofcase{\textsc{Par}} 
        Here $\trans{\semb{P}}{\inlabel{\vect{\ve{x}}}z}{P''}$, $y
        \freshin P,Q$, and by induction $\trans{P}{\vec{x}(y)}{P'}$ where
        $\semb{P'\{z/y\}} = P''$. By \textsc{Par} using $y \freshin Q$, we derive
        $\trans{P\parop Q}{\vec{x}(y)}{P'\parop Q}$. Finally, we note that since $y
        \freshin Q$, $\semb{(P'\parop Q)\{z/y\}} = P'' \parop \semb{Q}$.

    \proofcase{\textsc{Case}} Here $\trans{P_C}{\inlabel{\vect{\ve{x}}}z}{P''}$, where
        $P_C = \caseonly\;{\ci{\ve{\varphi}}{\ve{Q}}}$ is in the range of
        $\semb{\cdot}$. Hence $P_C$ must be the encoding of some prefix-guarded sum,
        i.e., $P_C = \semb{\Sigma_i\alpha_i.P_i} =
        \caseonly\;{\ci{\top}{\semb{\alpha_1}.\semb{P_1}}\casesep\dots\casesep\ci{\top}{\semb{\alpha_i}.\semb{P_i}}}$.
        By transition inversion, we can deduce that for some $j$, $\alpha_j =
        \vec{x}(y)$ and $\semb{P_j}\lsubst{z}{y} = P''$. By the \textsc{Prefix} rule,
        $\trans{\Sigma_i\alpha_i.P_i}{\vec{x}(y)}{P_j}$.

    \proofcase{\textsc{Out}} A special case of \textsc{Case}.

    \proofcase{\textsc{Rep}}  Here $\trans{\semb{P}\parop
        !\semb{P}}{\inlabel{\vect{\ve{x}}}z}{P''}$. By induction
        $\trans{P\parop !P}{\vec{x}(y)}{P'}$ where $\semb{P'\{z/y\}} = P''$. 
        Using the \textsc{Bang} rule, we derive $\trans{!P}{\vec{x}(y)}{P'}$.

    \proofcase{\textsc{Scope}} Here
        $\trans{\semb{P}}{\inn{x}{\langle\vec{z}\rangle}}{P''}$, $y \freshin P,Q$ and $a
        \freshin \vec{x},y,z$. By induction $\trans{P}{\vec{x}(y)}{P'}$ with
        $\semb{P'\{z/y\}} = P''$. Since $a\freshin \vec{x},y,z$, we obtain
        $\trans{(\nu a)P}{\vec{x}(y)}{(\nu a)P'}$ by the \textsc{Res} rule. Finally,
        $\semb{((\nu a)P')\{z/y\}} = (\nu a)P''$.\qed
  \end{proofcasesdescription}
\end{enumerate}

We give a proof for the strong operational correspondence.
\proof[Proof of Theorem~\ref{thm:polySynchTransition}] $ $
\begin{enumerate}

\item By induction on the derivation of $P'$. 
In case of input rule \textsc{eIn}, we apply Lemma~\ref{lem:polySynchTransition}~(1).
The other interesting cases are:

\begin{proofcasesdescription}
    \proofcase{\textsc{Comm}} Here $\trans{P}{\overline{\vec{x}}\langle
    y\rangle}{P'}$ and $\trans{Q}{\vec{x}(z)}{Q'}$. By induction, $\trans{\semb
    P}{\out{\langle\vec{x}\rangle}{y}}{P''}$ where $P'' \sim \semb{P'}$ and
    by Lemma~\ref{lem:polySynchTransition}~(1), $\trans{\semb
    Q}{\inn{\langle\vec{x}\rangle}{y}}{Q''}$ such that $\semb{Q'}\lsubst{y}{z}
    \sim Q''$. The \textsc{Com} rule lets us derive the transition
      \[
      \trans{\semb{P}\parop\semb{Q}}
            {\tau}
            {P'' \parop Q''}
      \]
      To complete the induction case, we note that $(\nu y)(P'' \parop Q'')
      \sim \semb{(\nu y)(P' \parop Q'\{y/z\})}$

    \proofcase{\textsc{Close}} Here $\trans{P}{\overline{\vec{x}}\langle\nu
    y\rangle}{P'}$ and $\trans{Q}{\vec{x}(y)}{Q'}$. We assume $y \freshin Q$;
    if not, $y$ can be $\alpha$-converted so that this holds. By induction,
    $\trans{\semb P}{\boutlabel{\langle\vec{x}\rangle}{y}{y}}{P''}$ where $P''
    \sim \semb{P'}$ and by Lemma~\ref{lem:polySynchTransition}~(1), $\trans{\semb
    Q}{\inn{\langle\vec{x}\rangle}{y}}{Q''}$ such that $\semb{Q'}\lsubst{y}{y}
    = \semb{Q'} \sim Q''$. The \textsc{Com} rule lets us derive the
    transition
      \[
      \trans{\semb{P}\parop\semb{Q}}
            {\tau}
            {(\nu y)(P'' \parop Q'')}
      \]
      To complete the induction case, we note that $(\nu y)(P'' \parop Q'')
      \sim \semb{(\nu y)(P' \parop Q')}$

    \proofcase{\textsc{Open}} Here $\trans{P}{\overline{\vec{x}}\langle
    y\rangle}{P'}$ with $y \neq x$, and by induction, $\trans{\semb
    P}{\out{\langle\vec{x}\rangle}{y}}{P''}$ where $P'' \sim \semb{P'}$. By
    \textsc{Open}, we derive $\trans{(\nu y)\semb P}{\boutlabel{\langle\vec{x}\rangle}{y}{y}}{P''}$.

  \end{proofcasesdescription}
\item By induction on the derivation of $P''$. The cases not shown are similar to Lemma~\ref{lem:polySynchTransition}~(2).
  \begin{proofcasesdescription}
    \proofcase{\textsc{Com}} Here $\trans{\semb P}{\boutlabel{\langle\vec{x}\rangle}{\vec{y'}}{y}}{P''}$, $\trans{\semb Q}{\inn{\langle\vec{x}\rangle}{y}}{Q''}$ and $y' \freshin Q$. Either $\vec{y'} = \epsilon$ or $\vec{y'} = y$; we proceed by case analysis.
      \begin{enumerate}
        \item If $\vec{y'} = \epsilon$, we have
$\trans{P}{\overline{\vec{x}}\langle y\rangle}{P'}$ where $\semb{P'} = P''$ by
induction and, by Lemma~\ref{lem:polySynchTransition}~(2),
$\trans{Q}{\vec{x}(z)}{Q'}$ where $\semb{Q'\{y/z\}} = Q''$. The \textsc{Comm}
rule then lets us derive $\trans{P\parop Q}{\tau}{P'\parop Q'\{y/z\}}$.
        \item If $\vec{y'} = y$, we have $\trans{P}{\overline{\vec{x}}\langle
\nu y\rangle}{P'}$ where $\semb{P'} = P''$ by induction and, by Lemma~\ref{lem:polySynchTransition}~(2), 
$\trans{Q}{\vec{x}(y)}{Q'}$ where $\semb{Q'\{y/y\}} =
\semb{Q'} = Q''$. The \textsc{Close} rule then lets us derive $\trans{P\parop
Q}{\tau}{(\nu y)(P'\parop Q')}$.
      \end{enumerate}
    \proofcase{\textsc{Open}} Here $\trans{\semb P}{\out{\langle\vec{x}\rangle}{y}}{P''}$ with $y \neq x$. By induction, $\trans{P}{\overline{\vec{x}}\langle y\rangle}{P'}$ where $\semb{P'} = P''$. By rule \textsc{Open}, $\trans{(\nu y)P}{\overline{\vec{x}}\langle\nu y\rangle}{P'}$.\qed
  \end{proofcasesdescription}
\end{enumerate}

We give the full abstraction result for this calculus. The definition of
congruence for polyadic synchronisation pi-calculus can be found in
\cite{carbone.maffeis:expressive-power} on page 6.

\begin{thm}
For all ${}^e\pi$ processes $P$ and $Q$, $P \sim Q$ iff $\semb{P} \sim \semb{Q}$
\end{thm}

\proof
$\mathcal{R} = \Set{(P,Q): \semb{P} \sim \semb{Q}}$ is an early congruence in
the polyadic synchronisation pi-calculus; if $P\RelR Q$ then
\begin{enumerate}
  \item If $\trans{P}{\vec{x}(y)}{P'}$ and $\semb{P} \sim \semb{Q}$, since
      $\mathcal{R}$ is equivariant, we can assume that $y\freshin P,Q$ without
      loss of generality. Fix $z$. By Lemma~\ref{lem:polySynchTransition} (1),
      $\trans{\semb{P}}{\inn{\langle\vec{x}\rangle}{z}}{P''}$ where $P'' \sim
      \semb{P'}\lsubst{z}{y} = \semb{P'\{z/y\}}$. Hence, since $\semb{P} \sim
      \semb{Q}$, $\trans{\semb{Q}}{\inn{\langle\vec{x}\rangle}{z}}{Q''}$ where
      $P'' \sim Q''$.  Hence, by Lemma~\ref{lem:polySynchTransition} (2)
      using $y \freshin Q$, $\trans{Q}{\vec{x}(y)}{Q'}$ where $\semb{Q'\{z/y\}}
      = Q''$. By transitivity, $\semb{P'\{z/y\}} \sim \semb{Q'\{z/y\}}$.

  \item If $\trans{P}{\alpha}{P'}$ and $\semb{P} \sim \semb{Q}$, since
      $\mathcal R$ is equivariant, we can assume that $\bn{\alpha}\freshin P,Q$
      without loss of generality. By Theorem~\ref{thm:polySynchTransition} (1),
      we have that $\trans{\semb P}{\semb\alpha}{P''}$ with $P'' \sim
      \semb{P'}$. Hence, since $\semb{P}\sim\semb{Q}$ and
      $\bn{\alpha}\freshin Q$, there is a $Q''$ such that $\trans{\semb
      Q}{\semb\alpha}{Q''}$ and $Q'' \sim P''$. By
      Theorem~\ref{thm:polySynchTransition} (2), there is $Q'$ such that
      $\trans{Q}{\alpha}{Q'}$ and $\semb{Q'} = Q''$. By transitivity,
      $\semb{P'} \sim \semb{Q'}$.
\end{enumerate}

Symmetrically, we show that 
$\mathcal{R} = \{(\emptyframe,\semb{P},\semb{Q}): P \sim Q\}$ 
is a congruence in \textbf{PSPI}:
\begin{proofcasesdescription}
  \proofcase{Static equivalence} Trivial since there is only a unit assertion.
  \proofcase{Symmetry} By symmetry of $\sim$
  \proofcase{Simulation} Here $\trans{\semb{P}}{\alpha'}{P''}$ and $P \sim Q$. We proceed by case analysis on $\alpha'$:
    \begin{enumerate}
        \item If $\alpha' = \inn{\langle\vec{x}\rangle}{z}$, then by Lemma~\ref{lem:polySynchTransition} (2) and a sufficiently fresh $y$, $\trans{P}{\vec{x}(y)}{P'}$ where $\semb{P'\{z/y\} } = P''$. Since $P \sim Q$, there exists $Q'$ such that $\trans{Q}{\vec{x}(y)}{Q'}$ and $P'\{z/y\} \sim Q'\{z/y\}$.
          Hence, by Lemma~\ref{lem:polySynchTransition} (1), $\trans{\semb{Q}}{\inn{\langle\vec{x}\rangle}{z}}{Q''}$ where $Q'' \sim \semb{Q'}\lsubst{z}{y} = \semb{Q'\{z/y\}}$. We have that $P'' = \semb{P'\{z/y\}}\RelR\semb{Q'\{z/y\}} \sim Q''$, which suffices.
      \item If $\alpha'$ is not an input, since $\mathcal{R}$ is equivariant, we can assume that $\bn{\alpha'}\freshin P,Q$ without loss of generality. Since $\trans{\semb{P}}{\alpha'}{P''}$, by Theorem~\ref{thm:polySynchTransition} (2) we have that $\trans{P}{\alpha}{P'}$ where $\semb{\alpha} = \alpha'$ and $\semb{P'} = P''$. Since $P \sim Q$, there is $Q'$ such that $\trans{Q}{\alpha}{Q'}$ and $P' \sim Q'$. By Theorem~\ref{thm:polySynchTransition} (1), $\trans{\semb Q}{\semb \alpha}{Q''}$, where $Q'' \sim \semb{Q'}$. Hence $P'' = \semb{P'}\RelR\semb{Q'} \sim Q''$, which suffices.
    \end{enumerate}
  \proofcase{Extension of arbitrary assertion} Trivial since there is only a unit assertion.
  \qed
\end{proofcasesdescription}

\begin{lem}
\label{lemma:pspi-to-psi-surjective}
$\semb{\cdot}$ is surjective up to $\sim$ on the set of case-guarded processes, that is, for every case-guarded $P$
there is a $Q$ such that $\semb{Q} \sim P$.
\end{lem}
\proof
By induction on the well-formed agent $P$.

\begin{proofcasesdescription}
\proofcase{Case $\lin{\langle\ve{x}\rangle}{y}y.P'$} It is valid to consider only
this form, since $\{y\}\in\vars(y)$. The IH is for some $Q'$, $\semb{Q'}\sim P'$.
Let $Q = \ve{x}(y).Q'$. Then $\semb{Q} = \lin{\langle\ve{x}\rangle}{y}y.\semb{Q'}
    \sim \lin{\langle\ve{x}\rangle}{y}y.P'$.

\proofcase{Case $\out{\langle\ve{x}\rangle}{y}.P'$} From IH, we get for some $Q'$,
$\semb{Q'}\sim P'$. Let $Q = \ve{x}\langle y\rangle.Q'$. Then $\semb{Q} =
\out{\langle\ve{x}\rangle}{y}.\semb{Q'} \sim \out{\langle\ve{x}\rangle}{y}.P'$.

\proofcase{Case $P' \pll P''$}  From IH, for some $Q',Q''$, we have $\semb{Q'}\sim
P'$ and $\semb{Q''}\sim P''$. Let $Q = Q'\pll Q''$. Then $\semb{Q} = \semb{Q'}\pll\semb{Q''}
\sim P'\pll P''$.

\proofcase{Case $(\nu x)P'$} Let $Q = \nu x Q'$, then by the induction hypothesis
$\semb{Q} = (\nu x)\semb{Q'} \sim (\nu x)P'$.

\proofcase{Case $!P'$} Let $Q = !Q'$ ($Q'$ from IH). $\semb{Q} ={} !\semb{Q'} \sim {}!P'$.

\proofcase{Case $\nil$} Then $\semb{\nil} = \nil \sim \nil$.

\proofcase{Case $\pass{\emptyframe}$} Then $\semb{\nil} = \nil \sim \pass{\emptyframe}$.

\proofcase{Case $\caseonly\;\ci{\ve{\varphi}}{\ve{P'}}$} For induction hypothesis
IH${}_{\caseonly}$, we have for every $i$ there is $Q'_i$ such that
$\semb{Q'_i} \sim P'_i$. The proof proceeds by induction on the length of
$\ve\varphi$.
    \begin{proofcasesdescription}
    \proofcase{Base case} Let $Q = \nil$, then $\semb{Q} = \nil \sim \caseonly$.
    \proofcase{Induction step} In this case, we get the following IH
        \[\semb{Q''}\sim
            \caseonly\;\ci{\varphi_1}{P_1}\casesep\dots\casesep\ci{\varphi_n}{P_n}\]
        We need to show that there is some $\semb{Q}$ such that
        \[\semb{Q}\sim
            \caseonly\;\ci{\varphi_1}{P_1}\casesep\dots\casesep\ci{\varphi_n}{P_n}
            \casesep\ci{\varphi_{n+1}}{P_{n+1}} = P
        \]
        First, we note that IH${}_{\caseonly}$ holds for every $i$ and in particular
        $i = n+1$, thus we get $\semb{Q'_{n+1}} \sim P_{n+1}$.
        Second, we note that $\varphi_{n+1}$ has two forms, thus we proceed by case
        analysis on $\varphi_{n+1}$.
        \begin{proofcasesdescription}
        \proofcase{Case $\varphi_{n+1} = \bot$} Let $Q = Q''$. Then
            \[\begin{array}{rcl}
            \semb{Q} &=& \semb{Q''} \\
                &\sim& \caseonly\;\ci{\varphi_1}{P_1}\casesep\dots\casesep\ci{\varphi_n}{P_n} \\
                &\sim&   %
                       \caseonly\;\ci{\varphi_1}{P_1}\casesep\dots\casesep\ci{\varphi_n}{P_n}\casesep\ci{\bot}{P_{n+1}}\\
            \end{array} \]
            We conclude the case.

        \proofcase{Case $\varphi_{n+1} = \top$}
        From the assumption, we know that $P_{n+1}$ is of form $\alpha.P'_{n+1}$ and that
        $\semb{Q'_{n+1}} \sim \alpha.P'_{n+1}$. By investigating the
        construction of $Q'_{n+1}$ we can conclude that $Q'_{n+1} = \alpha.Q''_{n+1}$
        where $\semb{Q''_{n+1}} \sim P'_{n+1}$.
        The agent from IH $Q''$ is either $\nil$, or prefixed agent, or a mixed sum.

        In case $Q'' = \nil$, let $Q = Q'_{n+1}$, then $\semb{Q} = \semb{Q'_{n+1}} \sim P$.

        In case $Q''$ is prefixed agent, let $Q = Q'' + Q'_{n+1}$. Since $Q''$ and $Q'_{n+1}$ are
        prefixed, $Q$ is well formed.
        Then $\semb{Q} = \caseonly\;\ci{\top}{\semb{Q''}}\casesep\ci{\top}{\semb{Q'_{n+1}}}
             \sim \caseonly\;\ci{\varphi_1}{P_1}\casesep\dots\casesep\ci{\varphi_n}{P_n}\casesep\ci{\top}{P_{n+1}}$.

        In case $Q''$ is a sum, let $Q = Q'' + Q'_{n+1}$. Since $Q'_{n+1}$ is
        guarded, $Q$ is well formed. Then
            \[\begin{array}{rclr}
            \semb{Q} &=& \caseonly\;\ci{\top}{\semb{Q''}}\casesep\ci{\top}{\semb{Q'_{n+1}}} \\
                &\sim& \caseonly\;\ci{\top}{(\caseonly\;\ci{\varphi_1}{P_1}\casesep\dots\casesep\ci{\varphi_n}{P_n})}\\
                &    & \casesep\ci{\top}{\semb{Q'_{n+1}}} \\
                &\sim& \text{(by Lemma \ref{lemma:flatten-case})} \\
                &    & \caseonly\;\ci{\varphi_1}{P_1}\casesep\dots\casesep\ci{\varphi_n}{P_n}\\
                &    & \casesep\ci{\top}{\semb{Q'_{n+1}}} \\
                &\sim& \caseonly\;\ci{\varphi_1}{P_1}\casesep\dots\casesep\ci{\varphi_n}{P_n}\\
                &    & \casesep\ci{\top}{P'_{n+1}} \\
            \end{array} \]

            This concludes the proof.\qed
        \end{proofcasesdescription}
    \end{proofcasesdescription}

\end{proofcasesdescription}

\begin{lem}
\label{lemma:pspi-to-psi-injective}
$\semb{\cdot}$ is injective, that is, for all $P, Q$,
if $\semb{P} = \semb{Q}$ then $P = Q$.
\end{lem}
\proof
By induction on $P$ and $Q$ while inspecting all the possible cases.
\qed

\subsection{Value-passing CCS}
\label{sec:value-passing-ccs-proofs}

This section contains the full proofs of the results found in Section~\ref{sec:vpCCS}
for the value-passing CCS.

\begin{lem}\label{lemma:CCSnoBoundOut}
If $P$ is a {\bf VPCCS} process such that $\trans{P}{\boutlabel{M}{\ve{x}}N}{P''}$ then $\ve{x} = \epsilon$
\end{lem}
\proof
  By induction on the derivation of $P'$. Obvious in all cases except
\textsc{Open}, where we derive a contradiction since only values can be
transmitted and yet only channels can be restricted - hence the name $a$ is
both a name and a value.
\qed

We prove strong operational correspondence using the implicit translation from value-passing CCS to CCS
of Milner \cite[Section 2.6, p. 56]{milner:cac}.
If $L$ is a set of labels, we write $L \freshin \alpha$ to mean that
for every $\ell \in L$ there is no $v$ such that $\alpha = \ell_v$ or $\alpha = \overline\ell_v$.

\proof[Proof of Theorem \ref{thm:CCStrans}] $ $
  \begin{enumerate}
    \item By induction on the derivation of $P'$.
      \begin{proofcasesdescription}

        \proofcase{\textsc{Act}} We have that $\trans{\alpha.P}{\alpha}{P}$.
        Since $\alpha.P$ is a closed value-passing CCS agent, $\alpha$ cannot be a free input.
        Thus, $\alpha$ is an output action $\alpha = \overline{x}(v)$ for some
        $x$ and $v$.  The \textsc{Out} rule then admits the derivation
        $\semb{\overline{x}(v).P} = \trans{\out{x}{v}.\semb{P}}{\out{x}{v}}{\semb{P}}$.

        \proofcase{\textsc{Sum}} There are two cases to consider: either
            $\Sigma_i P_i$ is the encoding of an input, or a summation.
          \begin{enumerate}
            \item If it is an encoding of an input $\Sigma_i P_i = x(y).P = \Sigma_v x(v). P\{v/y\}$,
                    then the action $\alpha$ must be the free input action $x(v)$ for some
                    value $v$.  Thus, for each $v$, we can
                    derive $\semb{x(y).P} = \trans{\lin{x}{y}{y}.\semb{P}}{\inn{x}{v}}{\semb{P\{v/y\}}}$ using the
                    $\textsc{In}$ rule.

            \item Otherwise it is a summation. We assume $\trans{\Sigma_i P_i}{\alpha}{P'}$.
                From induction hypothesis, we have $\trans{P_i}\alpha{P'}$, and
                \[
                    \trans{\semb{P_i}}{\semb{\alpha}}{\semb{P'}}
                \]
                for any $i$. By using this and the \textsc{Case} rule, we derive
                \[
                    \semb{\Sigma_i P_i} = \trans{\caseonly\;{\ci{\top}{\semb{P_1}}\casesep\cdots
                                                 \casesep\ci{\top}{\semb{P_i}}}}{\alpha}{\semb{P'}}
                \] as required.
          \end{enumerate}

        \proofcase{\textsc{Com1}} Here $\trans{P}{\alpha}{P'}$, and by induction
        $\trans{\semb{P}}{\semb{\alpha}}{\semb{P'}}$. The \textsc{Par} rule admits
        derivation of the transition
        $\trans{\semb{P}\parop\semb{Q}}{\semb{\alpha}}{\semb{P'}\parop\semb{Q}}$, 
        as, by using Lemma~\ref{lemma:CCSnoBoundOut}, freshness side condition is vacuous.

        \proofcase{\textsc{Com2}} Symmetric to \textsc{Com1}.

        \proofcase{\textsc{Com3}} Here $\trans{P}{\alpha}{P'}$ and
        $\trans{Q}{\overline{\alpha}}{Q'}$. Since $\alpha$ is in the range of $\widehat
        \cdot$, there are $x$ and $v$ such that $\alpha = x(v)$ and $\overline{\alpha}
        = \overline{x}(v)$ (or vice versa, in which case read the next sentence
        symmetrically). By the induction hypotheses,
        $\trans{\semb{P}}{\inn{x}{v}}{\semb{P'}}$ and
        $\trans{\semb{Q}}{\out{x}{v}}{\semb{Q'}}$.
        Then
        $\trans{\semb{P}\parop\semb{Q}}{\tau}{\semb{P'}\parop\semb{Q'}}$ by the
        \textsc{Com} rule.

        \proofcase{\textsc{Res}} 
        Here $\trans{P \;\backslash\; L}{\alpha}{P'\;\backslash\; L}$ with 
        $L \freshin \alpha$.
        Hence $\overrightarrow{L} \freshin \semb{\alpha}$. 
        By induction $\trans{\semb{P}}{\semb{\alpha}}{\semb{P'}}$. 
        We use the \textsc{Res} rule $|L|$ times to derive 
        $\trans{(\nu \overrightarrow{L})\semb{P}}{\semb{\alpha}}{(\nu \overrightarrow{L})\semb{P'}}$.

        \proofcase{\textsc{Rep}} Here $\trans{P\parop !P}{\alpha}{P'}$. 
        By induction $\trans{\semb{P}\parop !\semb{P}}{\semb{\alpha}}{\semb{P'}}$. 
        By the \textsc{Rep} rule $\trans{!\semb{P}}{\semb{\alpha}}{\semb{P'}}$
      \end{proofcasesdescription}

    \item By induction on the derivation of $P'$.
      \begin{proofcasesdescription}
        \proofcase{\textsc{In}} Here
        $\trans{\lin{x}{y}{y}.\semb{P}}{\inn{x}{v}}{\semb{P\{v/y\}}}$. We match this by
        deriving $\trans{x(y).P}{x(v)}{P\{v/y\}}$ using the
        \textsc{Act} and \textsc{Sum} rules, where $\semb{x(y).P} = \lin{x}{y}{y}.\semb{P}$.

        \proofcase{\textsc{Out}} Here
        $\trans{\out{x}{v}.\semb{P}}{\out{x}{v}}{\semb{P}}$. We match this by deriving
        $\trans{\overline{x}(v).P}{\overline{x}(v)}{P}$ using the
        \textsc{Act} rule.

        \proofcase{\textsc{Com}} Here $\trans{\semb{P}}{\boutlabel{x}{\vec{y}}{v}}{P''}$, $\trans{\semb{Q}}{\inn{x}{v}}{Q''}$. By Lemma~\ref{lemma:CCSnoBoundOut}, $\vec{y} = \epsilon$, and by induction, $\trans{P}{\overline{x}(v)}{P'}$ and $\trans{Q}{x(v)}{Q'}$, where $\semb{P'} = P''$ and $\semb{Q'} = Q''$. Using the \textsc{Com3} rule we derive $\trans{P\parop Q}{\tau}{P'\parop Q'}$
        \proofcase{\textsc{Par}} Straightforward.
        \proofcase{\textsc{Case}} Our case statement can either be the encoding of either a summation or an $\mathbf{if}$ statement. We proceed by case analysis:
          \begin{enumerate}
            \item Here $\trans{\semb{P_j}}{\alpha'}{P''}$. By induction,
            $\trans{P_j}{\alpha}{P'}$ where $\semb{\alpha} = \alpha'$ and $P'' = \semb{P'}$. By \textsc{Sum},
            $\trans{\Sigma_i P_i}{\alpha}{P'}$.

            \item Here $\trans{\semb{P}}{\alpha'}{P''}$ and $\emptyframe\vdash
            b$. By induction, $\trans{P}{\alpha}{P'}$ where $\semb{\alpha} = \alpha'$ and
            $\semb{P'} = P''$. Since $b$ evaluates to $\true$,
            $\trans{\mathbf{if}\;b\;\mathbf{then}\;P}{\alpha}{P'}$.

          \end{enumerate}
        \proofcase{\textsc{Rep}} Straightforward.
        \proofcase{\textsc{Scope}} Here $\trans{\semb{P}}{\alpha'}{P''}$ with $x \sharp \alpha'$ and by induction, $\trans{P}{\alpha}{P'}$ where $\alpha' = \semb{\alpha}$ and $P'' = \semb{P'}$. Hence we can derive $\trans{P\;\backslash\;\{x\}}{\alpha}{P'\;\backslash\;\{x\}}$ by the \textsc{Res} rule.
        \proofcase{\textsc{Open}} Impossible, by Lemma~\ref{lemma:CCSnoBoundOut}.\qed
      \end{proofcasesdescription}
  \end{enumerate}


\begin{thebibliography}{{\AA}PBP{\etalchar{+}}13}

\bibitem[AF01]{abadi.fournet:mobile-values}
Mart{\'\i}n Abadi and C{\'e}dric Fournet.
\newblock Mobile values, new names, and secure communication.
\newblock In {\em Proceedings of POPL '01}, pages 104--115. ACM, January 2001.

\bibitem[{\AA}P10]{pohjola10:verifyingPsi}
Johannes {\AA}man~Pohjola.
\newblock Verifying psi-calculi.
\newblock M.~Sc.~thesis IT ; 10 052, Uppsala University, Department of
  Information Technology, 2010.

\bibitem[{\AA}P15]{sortedProofsurl}
Johannes {\AA}man~Pohjola.
\newblock Isabelle proof scripts for sorted psi-calculi.
\newblock Available at
  \url{http://www.it.uu.se/research/group/mobility/theorem/sortedPsi.tar.gz},
  2015.

\bibitem[{\AA}PBP{\etalchar{+}}13]{pohjola13:Reliable}
Johannes {\AA}man~Pohjola, Johannes Borgstr\"om, Joachim Parrow, Palle
  Raabjerg, and Ioana Rodhe.
\newblock Negative premises in applied process calculi.
\newblock Technical Report 2013-014, Department of Information Tecnology,
  Uppsala University, 2013.

\bibitem[Ben10]{Bengtson10.PhD}
Jesper Bengtson.
\newblock {\em Formalising process calculi}.
\newblock PhD thesis, Uppsala University, 2010.

\bibitem[BGP{\etalchar{+}}14]{borgstrom13:sorted}
Johannes Borgstr{\"o}m, Ram{\=u}nas Gutkovas, Joachim Parrow, Bj{\"o}rn Victor,
  and Johannes~{\AA}man Pohjola.
\newblock A sorted semantic framework for applied process calculi (extended
  abstract).
\newblock In Mart{\'i}n Abadi and Alberto Lluch~Lafuente, editors, {\em
  Trustworthy Global Computing}, number 8358 in Lecture Notes in Computer
  Science, pages 103--118. Springer, 2014.

\bibitem[BGRV15]{PWB15:TECS}
Johannes Borgstr{\"o}m, Ram{\=u}nas Gutkovas, Ioana Rodhe, and Bj{\"o}rn
  Victor.
\newblock A parametric tool for applied process calculi.
\newblock {\em ACM Transactions on Embedded Computing Systems}, 14(1), 2015.

\bibitem[BJPV11]{LMCS11.PsiCalculi}
Jesper Bengtson, Magnus Johansson, Joachim Parrow, and Bj{\"o}rn Victor.
\newblock Psi-calculi: a framework for mobile processes with nominal data and
  logic.
\newblock {\em LMCS}, 7(1:11), 2011.

\bibitem[Bla11]{BlanchetBook09}
Bruno Blanchet.
\newblock Using {H}orn clauses for analyzing security protocols.
\newblock In V{\'e}ronique Cortier and Steve Kremer, editors, {\em Formal
  Models and Techniques for Analyzing Security Protocols}, volume~5 of {\em
  Cryptology and Information Security Series}, pages 86--111. IOS Press, March
  2011.

\bibitem[CGK{\etalchar{+}}13]{mCRL2.TACAS13}
Sjoerd Cranen, Jan~Friso Groote, Jeroen J.~A. Keiren, Frank P.~M. Stappers,
  Erik~P. de~Vink, Wieger Wesselink, and Tim A.~C. Willemse.
\newblock An overview of the {mCRL2} toolset and its recent advances.
\newblock In Nir Piterman and Scott~A. Smolka, editors, {\em TACAS}, volume
  7795 of {\em Lecture Notes in Computer Science}, pages 199--213. Springer,
  2013.

\bibitem[CM03]{carbone.maffeis:expressive-power}
Marco Carbone and Sergio Maffeis.
\newblock On the expressive power of polyadic synchronisation in
  $\pi$-calculus.
\newblock {\em Nordic Journal of Computing}, 10(2):70--98, 2003.

\bibitem[DY83]{DY83}
Danny Dolev and Andrew~C. Yao.
\newblock On the security of public key protocols.
\newblock {\em IEEE Transactions on Information Theory}, 29(2):198--208, 1983.

\bibitem[EOW07]{Odersky.ECOOP07.matching.objects}
Burak Emir, Martin Odersky, and John Williams.
\newblock Matching objects with patterns.
\newblock In {\em Proceedings of the 21st European Conference on
  Object-Oriented Programming}, ECOOP'07, pages 273--298, Berlin, Heidelberg,
  2007. Springer-Verlag.

\bibitem[FG96]{fournet96.join}
C{\'e}dric Fournet and Georges Gonthier.
\newblock The reflexive {CHAM} and the join-calculus.
\newblock In {\em Proc. POPL}, pages 372--385, 1996.

\bibitem[FGM05]{FGM05:ATypeDisciplineForAuthiorizationPolicies}
C{\'e}dric Fournet, Andrew~D. Gordon, and Sergio Maffeis.
\newblock A type discipline for authorization policies.
\newblock In Mooly Sagiv, editor, {\em Proc. of ESOP 2005}, volume 3444 of {\em
  LNCS}, pages 141--156. Springer, 2005.

\bibitem[Gel85]{Gelernter.1985.LINDA}
David Gelernter.
\newblock Generative communication in {Linda}.
\newblock {\em ACM TOPLAS}, 7(1):80--112, January 1985.

\bibitem[Giv14]{given-wilson.express14.intensional}
Thomas Given{-}Wilson.
\newblock On the expressiveness of intensional communication.
\newblock In Johannes Borgstr{\"{o}}m and Silvia Crafa, editors, {\em
  Proceedings of {EXPRESS}/{SOS} 2014}, volume 160 of {\em {EPTCS}}, pages
  30--46, 2014.

\bibitem[Gor10]{Gorla:encoding}
Daniele Gorla.
\newblock Towards a unified approach to encodability and separation results for
  process calculi.
\newblock {\em Information and Computation}, 208(9):1031--1053, 2010.

\bibitem[GP01]{Gabbay01anew}
Murdoch~J. Gabbay and Andrew~M. Pitts.
\newblock A new approach to abstract syntax with variable binding.
\newblock {\em Formal Aspects of Computing}, 13:341--363, 2001.

\bibitem[GSV04]{DBLP:conf/fossacs/GiambiagiSV04}
Pablo Giambiagi, Gerardo Schneider, and Frank~D. Valencia.
\newblock On the expressiveness of infinite behavior and name scoping in
  process calculi.
\newblock In Igor Walukiewicz, editor, {\em Proceedings of FOSSACS 2004},
  volume 2987 of {\em LNCS}, pages 226--240. Springer, 2004.

\bibitem[GWGJ10]{Jay10.ConcurrentPattern}
Thomas Given-Wilson, Daniele Gorla, and Barry Jay.
\newblock Concurrent pattern calculus.
\newblock In Cristian Calude and Vladimiro Sassone, editors, {\em Theoretical
  Computer Science}, volume 323 of {\em IFIP Advances in Information and
  Communication Technology}, pages 244--258. Springer, 2010.

\bibitem[HJ06]{haack.jeffrey:pattern-matching-spi}
Christian Haack and Alan Jeffrey.
\newblock Pattern-matching spi-calculus.
\newblock {\em Information and Computation}, 204(8):1195--1263, 2006.

\bibitem[Hon93]{Honda:1993}
Kohei Honda.
\newblock Types for dyadic interaction.
\newblock In Eike Best, editor, {\em CONCUR '93, 4th International Conference
  on Concurrency Theory, Hildesheim, Germany, August 23-26, 1993, Proceedings},
  volume 715 of {\em Lecture Notes in Computer Science}, pages 509--523.
  Springer, 1993.

\bibitem[HU10]{Huffman10:NFN:2176728.2176735}
Brian Huffman and Christian Urban.
\newblock A new foundation for {N}ominal {I}sabelle.
\newblock In {\em Proceedings of the First international conference on
  Interactive Theorem Proving}, ITP'10, pages 35--50. Springer, 2010.

\bibitem[H{\"u}t11]{Hyttel.CONCUR11.TypedPsi}
Hans H{\"u}ttel.
\newblock Typed psi-calculi.
\newblock In Joost-Pieter Katoen and Barbara K{\"o}nig, editors, {\em CONCUR
  2011 -- Concurrency Theory}, volume 6901 of {\em LNCS}, pages 265--279.
  Springer, 2011.

\bibitem[H{\"u}t14]{hyttel.tgc13.resources}
Hans H{\"u}ttel.
\newblock Types for resources in $\psi$ -calculi.
\newblock In Mart{\'i}n Abadi and Alberto Lluch~Lafuente, editors, {\em
  Trustworthy Global Computing}, LNCS, pages 83--102. Springer International
  Publishing, 2014.

\bibitem[HV13]{BETTY13.stateoftheart}
Hans H{\"u}ttel and Vasco~T Vasconcelos.
\newblock The foundations of behavioural types.
\newblock State-of-the art report of WG1 of the BETTY project (EU COST Action
  IC1201). To appear, 2013.

\bibitem[JBPV10]{LICS10.WeakPsi}
Magnus Johansson, Jesper Bengtson, Joachim Parrow, and Bj{\"o}rn Victor.
\newblock Weak equivalences in psi-calculi.
\newblock In {\em Proc. of LICS 2010}, pages 322--331. IEEE, 2010.

\bibitem[JVP12]{JLAP12.SymbolicPsi}
Magnus Johansson, Bj{\"o}rn Victor, and Joachim Parrow.
\newblock Computing strong and weak bisimulations for psi-calculi.
\newblock {\em Journal of Logic and Algebraic Programming}, 81(3):162--180,
  2012.

\bibitem[Kri09]{Krishnaswami:POPL09:focusing.pattern.matching}
Neelakantan~R. Krishnaswami.
\newblock Focusing on pattern matching.
\newblock In {\em Proceedings of the 36th Annual ACM SIGPLAN-SIGACT Symposium
  on Principles of Programming Languages}, POPL '09, pages 366--378, New York,
  NY, USA, 2009. ACM.

\bibitem[LSD11]{dblp:conf/issre/liusd11}
Yang Liu, Jun Sun, and Jin~Song Dong.
\newblock {PAT} 3: An extensible architecture for building multi-domain model
  checkers.
\newblock In Tadashi Dohi and Bojan Cukic, editors, {\em ISSRE '11}, pages
  190--199. IEEE, 2011.

\bibitem[Mil89]{milner:cac}
Robin Milner.
\newblock {\em Communication and Concurrency}.
\newblock Prentice-Hall, Inc., 1989.

\bibitem[Mil93]{milner:polyadic-tutorial}
Robin Milner.
\newblock The polyadic $\pi$-calculus: A tutorial.
\newblock In Friedrich~L. Bauer, Wilfried Brauer, and Helmut Schwichtenberg,
  editors, {\em Logic and Algebra of Specification}, volume~94 of {\em Series
  F}. NATO ASI, Springer, 1993.

\bibitem[PBR{\AA}P13]{HOPSI}
Joachim Parrow, Johannes Borgstr{\"o}m, Palle Raabjerg, and Johannes
  {\AA}man~Pohjola.
\newblock Higher-order psi-calculi.
\newblock {\em Mathematical Structures in Computer Science}, FirstView, June
  2013.

\bibitem[Pit03]{PittsAM:nomlfo-jv}
Andrew~M. Pitts.
\newblock Nominal logic, a first order theory of names and binding.
\newblock {\em Information and Computation}, 186:165--193, 2003.

\bibitem[San93]{sangiorgi:expressing-mobility}
Davide Sangiorgi.
\newblock {\em Expressing Mobility in Process Algebras: First-Order and
  Higher-Order Para\-digms}.
\newblock PhD thesis, University of Edinburgh, 1993.
\newblock {CST-99-93} (also published as ECS-LFCS-93-266).

\bibitem[SLDC09]{Sun:2009:ISP:1607726.1608426}
Jun Sun, Yang Liu, Jin~Song Dong, and Chunqing Chen.
\newblock Integrating specification and programs for system modeling and
  verification.
\newblock In {\em TASE '09}, pages 127--135. IEEE Computer Society, 2009.

\bibitem[SNM07]{Syme.ICFP07.active.patterns}
Don Syme, Gregory Neverov, and James Margetson.
\newblock Extensible pattern matching via a lightweight language extension.
\newblock In {\em Proceedings of the 12th ACM SIGPLAN International Conference
  on Functional Programming}, ICFP '07, pages 29--40, New York, NY, USA, 2007.
  ACM.

\bibitem[SS05]{GC04Kell}
Alan Schmitt and Jean-Bernard Stefani.
\newblock The {Kell} calculus: A family of higher-order distributed process
  calculi.
\newblock In Corrado Priami and Paola Quaglia, editors, {\em Global Computing},
  volume 3267 of {\em LNCS}, pages 146--178. Springer Berlin Heidelberg, 2005.

\bibitem[SW01]{sangiorgi.walker:theory-mobile}
Davide Sangiorgi and David Walker.
\newblock {\em The $\pi$-calculus: a Theory of Mobile Processes}.
\newblock Cambridge University Press, 2001.

\bibitem[Urb08]{U07:NominalTechniquesInIsabelleHOL}
Christian Urban.
\newblock Nominal techniques in {I}sabelle/{HOL}.
\newblock {\em Journal of Automated Reasoning}, 40(4):327--356, May 2008.

\end{thebibliography}
\newcommand{\etalchar}[1]{$^{#1}$}
\newcommand{\noopsort}[1]{}

\end{document}

